\documentclass[aps,pra,twocolumn,british,superscriptaddress,nofootinbib,10pt]{revtex4-2}

\usepackage{header}

\begin{document}

\makeatletter
\let\stashsmartcomma\sm@rtcomma
\let\sm@rtcomma,
\makeatletter
\title{Expanding the Class of Free Fermions via Twin-Collapse Methods}

\author{Jannis Ruh}
\email{jannis.ruh1@uts.edu.au}
\affiliation{Centre for Quantum Software and Information, School of Computer Science, Faculty of Engineering \& Information Technology, University of Technology Sydney, NSW 2007, Australia}

\author{Samuel J. Elman}
\email{samuel.elman@uts.edu.au}
\affiliation{Centre for Quantum Software and Information, School of Computer Science, Faculty of Engineering \& Information Technology, University of Technology Sydney, NSW 2007, Australia}


\begin{abstract}
    We present a novel graph-theoretic approach to simplifying generic many-body Hamiltonians. Our primary result introduces a recursive twin-collapse algorithm, leveraging the identification and elimination of symmetric vertex pairs (twins), as well as line-graph modules, within the frustration graph of the Hamiltonian. This method systematically block-diagonalizes Hamiltonians, significantly reducing complexity while preserving the energetic spectrum. Importantly, our approach expands the class of models that can be mapped to non-interacting fermionic Hamiltonians (free-fermion solutions), thereby broadening the applicability of classical solvability methods. Through numerical simulations on spin Hamiltonians arranged in periodic lattice configurations and Majorana Hamiltonians, we demonstrate that the twin-collapse increases the identification of simplicial and claw-free graph structures, which characterize free-fermion solvability. Finally, we extend our framework by presenting a generalized discrete Stone-von Neumann theorem. This comprehensive framework provides new insights into Hamiltonian simplification techniques, free-fermion solutions, and group-theoretical characterizations relevant for quantum chemistry, condensed matter physics, and quantum computation.
\end{abstract}

\maketitle

\makeatletter
\let\sm@rtcomma\stashsmartcomma
\makeatother

Many-body Hamiltonians are of interest from a fundamental and applied perspective, for example in quantum chemistry simulations and condensed matter ~\cite{Nguyen2024Quantum,Watts24}. Particularly interesting is the technique of studying the complexity of solving Hamiltonians based on the commutation relations between the terms in the Hamiltonian. For example, while the complexity of the local Hamiltonian problem in general is known to be \acs{qma}-complete~\cite{Kempe2006Complexity}, the complexity of the commuting variant of the local Hamiltonian problem was studied in Refs.~\cite{bravyi2005Commutative,schuch_complexity_commuting+ham} and shown to be \acs{np}-complete. Similarly, the complexity of the non-contextual local Hamiltonian problem, defined by a restricted commutation structure, was also shown to be \acs{np}-complete in Ref.~\cite{Kirby2020}, also a reduction in complexity from the general case. References~\cite{Patel2024Exactly,Patel2024Extension} have also studied classical algorithms for solving Hamiltonians with reduced complexity based on commutation structure.
Furthermore, in special cases it is known that the commutation structure of a Hamiltonian allows for integrability, or even efficient solutions by classical means. One such example is the case where a many-body Hamiltonian admits a description in terms of non-interacting fermions.

The Jordan-Wigner transformation~\cite{jordan_wigner_transformation}, and its generalisations (see Ref.~\cite{chapman_solvable_spin_models} and references therein), provide a map between spin and fermionic representations for many-body models. This family of maps represent monomial-to-monomial transformations in the sense that one Pauli string is mapped to one fermionic string. In the special case when such transformations result in a fermionic Hamiltonian that is quadratic in fermionic operators, such that the model is described by a system of non-interacting fermions, the complexity of the solution is exponentially reduced, thus allowing a tractable solution on a classical machine~\cite{kaufman1949crystal,onsager1944crystal,jordan_wigner_transformation,schultz1964two,kitaev2006anyons,mandal2009exactly,fradkin1989jordan, wang1991ground, huerta1993bose, batista2001generalized, bravyi2002fermionic, verstraete2005mapping, nussinov2012arbitrary, chen2018exact, chen2019bosonization, backens2019jordan, tantivasadakarn2020jordan}. 

More recently, \citet{fendley2019free} presented an example of a model which admits a free-fermionic solution despite provably admitting no monomial-to-monomial map from spins to bilinear fermions. This example was then generalised to a whole family of models first in one dimension~\cite{elman_free_fermions_behind_the_disguise} and then arbitrary dimensions~\cite{chapman_unified_free_fermions} using graph theoretic principles based on the commutation structure of the Pauli terms. The binary commutation relations of Pauli operators allow the application of graph-theoretical methods to find free-fermionic solutions of the spin Hamiltonians, but also to characterise Pauli Lie algebras~\cite{kokcu2024classification,Wiersema2024,aguilar2024full}, which helps to further our understanding of the complexity of simulating such models~\cite{Pozsgay2024}.

In this work, we consider generic Hamiltonians written in any basis (spin or particle) that has a binary commutation rule. Our first main result is a simplification of Hamiltonians in terms of a reduction of number of terms, based the recognition and `collapsing' of graph theoretic structures known as twins. This twin-collapse technique allows to block diagonalise the Hamiltonians into symmetric subspaces, resulting in a simpler Hamiltonian within each block allowing to apply further analysis techniques. More specifically, we show that there exists a set of orthogonal complete projectors that effectively remove all terms in the Hamiltonian that correspond to a recursive collapse of all twins in the frustration graph. The projectors commute with the Hamiltonian and thus, we find that the simplified Hamiltonian is the same up to the projectors and the weights within each block. These results can be seen as a direct generalisation of Ref.~\cite{Kirby2020}. Furthermore, the twin-collapse technique allows us to not just eliminate twins, but also other structures, including modules that are line graphs.

We use the block-diagonalisation-by-collapse technique to extend the class of free-fermion models presented in Ref.~\cite{chapman_solvable_spin_models} and then numerically evaluate the ubiquity of such a solution method for several classes of Hamiltonians. Specifically, we implement the twin-collapse technique to simplify the frustration graph of various random spin Hamiltonians, as well as for generic Majorana Hamiltonians, and then employ a graph structure detection algorithm to determine the likelihood that the resulting graph has the requisite properties to admit a solution via Fendley's method~\cite{fendley2019free}.

Lastly, we discuss generalisations of the Pauli group, and similar groups to which our  previous results are applicable, for example, the Majorana group. This leads to a variation of the Stone-von Neumann theorem~\cite{heinrich_stabiliser_techniques}, showing the conditions under which two members of this family are equivalent to each other by unitary conjugation.

The paper has the following structure: In \cref{s"twin_block_diag}, we discuss our main result of how twin symmetries in the frustration graph lead to a block-diagonalisation of Hamiltonians. We then apply this technique to expand the class of free-fermion models in \cref{s"res_free_fermions}, followed by numerical examples. Finally, in \cref{s"res_stone_von_neumann} we discuss generalisations of the Pauli group and the discrete Stone-von Neumann theorem. \Cref{s"modular_decomposition,s"app_graphs} cover technical details on the modular decomposition of graphs relevant to the numerical simulations. In \cref{s"app_block_diag} we give a full constructive proof of the block diagonalisation, and in \cref{s"more_scf_examples} we extend the discussion on the free-fermion models. In \cref{s"app_isomorphisms} we fully proof the generalised discrete Stone-von Neumann theorem.

\section{Preliminaries}\label{s"preliminaries}

\subsection{Graph Theory}

A graph $G=(V,E)$ is a set of vertices $V$, and an edge set $E \subset V^{\times 2}$. We consider only simple graphs (undirected graphs with no self-loops). For a given $X \subseteq V$, the \textit{induced subgraph} $G[X] = (X, E \cap (X \times X))$ is the graph that contains vertices $X$ and all edges that have \textit{both} endpoints in $X$.
The \emph{order} of a graph $\Xi$, is the cardinality of its vertex set.

The \textit{open neighbourhood} of a vertex $x\in V$ is the set ${\neigh(x) = \set{y \in V}{(x, y)\in E}}$, with the \textit{complementary} open neighbourhood defined as ${\neigh^{c}(x)=\set{y \in V}{y\neq x, (x, y)\notin E}}$. The \textit{closed neighbourhood} is defined as $\neigh[x]=\neigh(x)\cup \{x\}$ with \textit{complementary closed neighbourhood} defined analogously ($\neigh^c[x]=V\setminus\neigh(x)$).

An \emph{independent set} $S\subseteq V$ is a subset of vertices with no edges between them. A \textit{clique} or complete subgraph $K\subseteq V$ is a subset of vertices such that all vertices in $K$ are pairwise connected. A \emph{simplicial clique}, $K_s$ is a non-empty clique such that, for every vertex $x \in K_s$, $\neigh(x)$ induces a clique in $G[V{\setminus}K_s]$. The \textit{claw} $K_{1,3}$ is the complete bipartite graph on one and three vertices, consisting of a central vertex neighbouring to every vertex in an independent set of order three.

A subset of vertices, $X \subseteq V$ are true (respectively false) siblings if for all $x, y \in X$ we have $\neigh[x] = \neigh[y]$ ($\neigh(x)=\neigh(y)$); in the special case of $\abs{X} = 2$, these vertices are referred to as true (false) \textit{twins}. A graph $G$ is called a \textit{cograph} if, and only if every non-trivial induced subgraph contains at least one pair of twins, false or true. 
A natural extension of siblings are modules defined as:
\begin{definition}[Modules]\label{.modules}
    Let $G = (V, E)$ be a graph. A module $X \subseteq V$ is defined through the following equivalent definitions:
    \begin{itemize}
        \item For all $y \in V{\setminus}X$ it holds
        \begin{equation}
            \exists x \in X: y \in \neigh(x) \iff \forall x \in X: y \in \neigh(x)\;.
        \end{equation}
        \item For all $x, y \in X$ it holds
        \begin{equation}
            \neigh(x) {\setminus} X = \neigh(y) {\setminus} X\;.
        \end{equation}
    \end{itemize}
\end{definition}
Recursively partitioning a graph into modules, i.e., constructing the modular decomposition tree of the graph (cf. \cref{s"modular_decomposition}), allows us to apply efficient algorithms to recursively detect and remove twins from a graph as well as detect whether a graph is claw-free and if so, if it contains simplicial cliques~\cite{chudnovsky_growing_without_cloning}.

A graph $G$ is a \textit{line graph} if there exists a root graph $R$, such that every edge in $R$ maps bijectively onto a vertex in $G$, and there is an edge in $G$ if, and only if, the two corresponding edges are incident to the same vertex in $R$.

\subsection{Linear Algebra}

Let $Y$ be a vector space of dimension \en[N] over a field $K$. We write ${\M_N(K) \cong \ls(Y)}$ for the set of all $N \times N$ matrices with entries in $K$, or linear mappings in $Y$. We denote by ${\GL_N(K) \cong \GL(Y)}$ the multiplicative group of all invertible matrices, or invertible linear mappings. If $K = \C$, we denote the set of unitary matrices (or operators acting on the vector space $Y$) by $\U(N) \cong \U(Y)$.

For even $N$, the standard symplectic form (also if $K$ is just a ring) is defined as
\begin{equation}
  \mathrm{\Omega}_N = \matp{0 & \1_{N/2} \\ -\1_{N/2} & 0}\;.
\end{equation}
Given $A, B \in \M_N(K)$, we define the Hilbert-Schmidt inner product as $\braket{A,
B}_{\mathrm{HS}} = \tr\w(B^\da A)$.

Given an operator $A$ and a projector $P$, that commutes with $A$, projecting into a space $U$, we denote the restriction of $A$ into $U$ by $\rst{A}{U} = \rst{A}{P}$.

\subsection{Group and Ring Theory}

Let $(J, \cdot)$ be a group that acts on a set $M$ from left and right. We write the
(group) commutator as
\begin{equation}
    \begin{split}
        \tbk{\cdot, \cdot} : J \times J &\to J,
  \\
  (g, h) &\mapsto ghg^{-1}h^{-1}\;,
    \end{split}
\end{equation}
and the conjugation action, $g*x$, is defined as
\begin{equation}
    \begin{split}
        * : J \times M &\to M,\\
        (g, x) &\mapsto gxg^{-1} \;.
    \end{split}
\end{equation}
For a subset $X \subseteq J$, $\braket{g \in X}$ denotes the subgroup generated by the
elements in $X$.

Given a ring $(R, +, \cdot)$, we define $R^\times$ to be the multiplicative group formed by set of units in $R$, i.e.,
invertible elements.

For the group $(J, \cdot)$ with representation $\mu: J \to \GL(Y)$, for some finite dimensional vector space $Y$ over a field $K$, the \emph{character} $\chi$ of the representation is the trace of the representation, i.e.,
\begin{equation}
    \chi = \tr \circ \mu: J \to K\;.
\end{equation}

\section{Twin-Symmetry Block-Diagonalisation} \label{s"twin_block_diag}

In this section, we describe our first result, and outline how to simplify Hamiltonians based on collapsing twins in the frustration graph. Mathematical details can be found in \cref{s"mathematical_details,s"app_block_diag}.


Let $\hi$ be a finite-dimensional complex Hilbert space. We denote by $S\subset \U(\hi)$ a group of unitary operators on $\hi$ such that for all $g,h\in S$ ${\tbk{g, h}\in \{-1,+1\}}$.
Given a hermitian subset $V \subseteq S$, we study generic Hamiltonians of the form
\begin{equation}
    H = \sum_{g \in V} w_g g \;, \label{e"hamiltonian}
\end{equation}
with $w_g \in \R{\setminus}\bc{0}$ for all $g \in V$. Without loss of generality, we assume that for all $g \in V$ there is no $h \in V$ with $g \propto h$.

One example of such a group $S$ is the Pauli group which provides an orthonormal hermitian basis of $\M_{2^n}(\C)$; another example is the Majorana group. In~\cref{s"res_stone_von_neumann}, we study a generalised family of such groups and show the conditions under which they are equivalent to each other. The characteristic binary commutation relations of the operators in $S$ allow us to study the Hamiltonian with graph-theoretical methods. To this end we define the frustration graph as follows:
\begin{definition}[Frustration graph]
\label{def"frustration_graph}
The \emph{frustration graph} of $H$ is defined as $G \coloneqq \frust{H} \coloneqq
\frust{V} \coloneqq (V, E)$, where the edge set is given by
\begin{equation}
    E = \set{(g, h) \in V \times V}{\tbk{g, h} = -1}\;,
\end{equation}
that is, vertices are neighbouring if, and only if, the corresponding operators
anticommute.
\end{definition}

Let $\bc{g, h} \subseteq V$ be twins in the frustration graph $G$ of $H$. For all $x \in V{\setminus}\bc{g, h}$ we have $\tbk{gh, x} = 1$; furthermore if $\bc{g, h}$ are false twins, then we also have $\tbk{gh, g} = \tbk{gh, h} = 1$; thus, the product of the false twins, $gh$, defines a symmetry of the Hamiltonian. When $\bc{g,h}$ are false twins, it is possible to find a projector with the same property, effectively removing one of the twins, as we show below.
When $\bc{g,h}$ are true twins, one can find a unitary operator that merges the twins in \cref{e"hamiltonian} into a single vertex, but leaves the rest of the Hamiltonian invariant. Despite the fact that only the false twins represent true symmetries of the Hamiltonian, we abuse terminology and refer to all twins (both false and true) collectively as \textit{twin symmetries}, since we are able to remove both without affecting the spectrum.

By repeatedly applying these projections and rotations to all sibling sets within the frustration graph, we can simplify the Hamiltonian while preserving its spectrum. Specifically, the process involves first rotating all primitive true siblings so that each true-sibling set collapses into a single vertex. This step is then followed by a projection that further condenses the graph by merging each distinct set of false siblings into single vertices. The entire procedure is recursively repeated until no sibling structures remain in the graph, leading us to the following result:
\begin{theorem}[Informal]\label{.result_informal}
Let $H$ be a Hamiltonian with frustration graph $G = (V, E)$, and let $X = \bc{-1, +1}^m$, for some \en[m], be a parameter space. Then there exists a complete set of commuting, orthogonal projectors $\bc{P(x)}_{x \in X}$, which commute with $H$, and set of unitary rotations $\bc{U(x)}_{x \in X}$, such that
\begin{equation}\label{e"rotated_blocks}
    H = \sum_{x \in X} P(x) \rst{H}{P(x)} P(x)
\end{equation}
and
\begin{equation}
    \rst{H}{P(x)} = \rst{\w(U\ur(x) * H_C(x))}{P(x)}\;,
\end{equation}
with
\begin{equation}
    H_C(x) = \sum_{g \in V'} w'_g\w(x) g \;\;\;\;\; (w'_g \in \R)
\end{equation}
where $V'$ is the vertex set obtained by recursively collapsing all twins in $G$. Furthermore $P(x)$ commutes with $U\ur(x) * H_C(x)$ for all $x \in X$.
\end{theorem}



\Cref{e"rotated_blocks} describes a block diagonalisation of the Hamiltonian, where each block represents a distinct symmetric subspace of the projectors, $P(x)$. Within each symmetric subspace, the Hamiltonian may be rotated independently such that the Hamiltonian may be described by a reduced Hamiltonian. We refer to this process as \textit{collapsing} the twins, and the reduced Hamiltonian as a \textit{collapsed Hamiltonian}. Notably, the frustration graph is the same for all symmetric subspaces. The reduction of the frustration graph corresponds to the removal of summands in the Hamiltonian, simplifying the model. Note that this process of collapsing cannot be done in one step in general, i.e., multiple alternating rounds of collapsing false twins and true twins may be necessary as the example in \cref{f"minimal_alternating_tree} shows. This iterative process can be implemented recursively on the modular decomposition tree of the frustration graph as we describe in \cref{s"modular_decomposition}.

In the special case where $G$ is a cograph, the Hamiltonian may be collapsed to a single operator (\cref{.cograph_collapse}), for example as the graph in \cref{f"minimal_alternating_tree}, i.e., the Hamiltonian is effectively diagonalised; however, it should be noted that the parameter space $X$ may be exponential in size. In this way, we can see this special case as an extension of the result from Ref.~\cite[Lemma~1]{Kirby2020}.

In \cref{s"app_block_diag}, we discuss how the sequence of twin-collapses can be extended to also allow collapsing modules that are not cographs but line graphs (we refer to these modules which are line graphs as \textit{line-graph modules}), if the group $S$ is unitarily equivalent to the Pauli group. We discuss a set of unitarily equivalent Pauli groups in \cref{s"res_stone_von_neumann}; one example would be the Majorana group.
\Cref{.result_informal} can then be extended as follows:
\begin{corollary}\label{.full_result_informal}
In \cref{.result_informal}, we can set $V'$ to be the vertex set obtained by recursively collapsing all twins \emph{and line-graph modules} in $G$ if the group $S$ is unitarily equivalent to the Pauli group.
\end{corollary}

More generally, further techniques may be applied to the simplified model to characterise the solvability of the model.

\begin{figure}[t]
  \centering
 \includegraphics{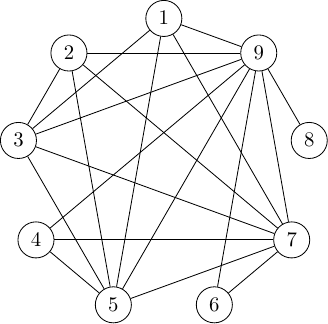}
  \caption{
    Example of a graph that requires multiple alternating rounds of false and true twins collapses. The graph can be fully collapse onto a single vertex by collapsing the following false and true twins in that order: $\bc{1, 2} \mapsto \bc{2}$, $\bc{2, 3} \mapsto \bc{3}$, $\bc{3, 4} \mapsto \bc{4}$, $\bc{4, 5} \mapsto \bc{5}$, $\bc{5, 6} \mapsto \bc{6}$, $\bc{6, 7} \mapsto \bc{7}$, $\bc{7, 8} \mapsto \bc{8}$, $\bc{8, 9} \mapsto \bc{9}$. Note that this order is strict since, for example, $\{3, 4\}$ is not a twin of any kind until $\{1, 2, 3\}$ have been merged.
    }\label{f"minimal_alternating_tree}
  \vspace{-2ex}
\end{figure}

\subsection{Mathematical Details}\label{s"mathematical_details}

\subsubsection{Twin Symmetries}\label{s"twin_symmetries}

The goal is to block-diagonalise $H$ by recognising graphical structures; specifically, we shall use false-twin symmetry projections and true-twin rotations to simplify the frustration graph $G$ of $H$. Our results are based on the following proposition, which we shall quickly prove here for completeness: First, let us deal with the false-twin symmetries.
\begin{lemma}[\cite{chapman_solvable_spin_models}]\label{.false_twin_symmetries}
    Let $G = (V, E)$ be the frustration graph of a Hamiltonian $H$. Define the group
    \begin{equation}
      L' = \groupset{gh}{(g, h) \in V \times V \text{ are false twins}}\;.
    \end{equation}
    If $-1 \notin L'$, set $L = L'$. Otherwise let the group $L$ be generated by representatives of $L'/\braket{-1}$. $L$
    is abelian, $-1 \notin L$ and its group elements commute with each term in $H$ and are hermitian and unitary; specifically, they define symmetries.
\end{lemma}
\begin{proof}
    $L'$ is clearly a unitary, abelian group and its group elements commute with each term in $H$ (there is always an even number of minus signs for the commutators). Furthermore, its generators are hermitian, since they are products of two commuting hermitian operators. With that and the abelianess, all elements in $L'$ are Hermitian. These properties transfer analogously to $L$ and it remains to show $-1 \notin L$. Without loss of generality, let $-1 \in L'$. Let $L'/\braket{-1}$ be independently generated by $g_1\braket{-1}, \ldots g_m\braket{-1}$, for some $g_1, \ldots, g_m \in L'$, \en[m]. Assume $-1 \in L$. Then it must be $-1 = g_1\cdots g_l$ for some \en[l] with $2 \leq l \leq m$, where we assume w.l.o.g. that the generators are ordered accordingly (note that it cannot be $\abs{L'} = 1$). But then, it is $g_2\braket{-1}\cdots g_l\braket{-1} = \w(g_1\braket{-1})^{-1} = g_1\braket{-1}$, which contradicts the independence of the generators.
\end{proof}
As with all symmetries of a Hamiltonian, we are able to project the Hamiltonian into the subspace of the symmetries identified by false twins leading to the following proposition:
\begin{proposition}[\cite{chapman_solvable_spin_models}]\label{.symmetry_restriction}
    Let $G = (V, E)$ be the frustration graph of a Hamiltonian $H$ and $L = \braket{g_1, \ldots, g_m}$ as in \cref{.false_twin_symmetries}, where $\bc{g_1, \ldots, g_m} \subset L$ are independent generators, \en[m]. For $x \in \bc{0, 1}^m$ define the stabiliser group $L_x = \braket{(-1)^{x_1} g_1, \ldots, (-1)^{x_m}g_m}$ and the stabilised space $\hi_x = \set{y \in \hi}{\forall g \in L_x : gy = y}$. Furthermore, define the mapping
    \begin{align}
        \begin{split}
            \beta_x : \bc{\text{false twins of $G$}} &\to \bc{\pm 1}, \\
            (g, h) &\mapsto \sign\w(\rst{\w(gh)}{\hi_x}).
        \end{split}
    \end{align}
    Then, the graph $G_x = \frust{\rst{H}{\hi_x}}$ contains no false twins. More specifically, let $M$ be the set of maximal false siblings sets where we allow $\abs{T} = 1$ for $T \in M$, and fix one $g_T \in T$ for all $T \in M$; then it holds
    \begin{equation}
        \rst{H}{\hi_x} = \sum_{T \in M} \w(\sum_{h \in T} \beta_x(g_T, h) w_h) \rst{g_T}{\hi_x}\;.
    \end{equation}
\end{proposition}
\begin{proof}
    Let $x \in \bc{0, 1}^m$. Firstly, note that since $L$ is abelian with $-1 \notin L$, $L_x$ is indeed a stabiliser group and $\hi_x$ is not trivial. Now let $(g, h) \in V^2$ be false twins. Since either $+gh \in L$ or $-gh \in L$ it follows that $gh = a s$ for some $s \in L_x$ and $a \in \bc{\pm 1}$. But then we have $\rst{\w(gh)}{\hi_x} = a \rst{s}{\hi_x} = a$ which shows that $\beta_x$ is well-defined. Furthermore, we have
    \begin{subalign}
        \rst{H}{\hi_x} &= \sum_{T \in M} \sum_{h \in T} w_h \rst{h}{\hi_x}\\
        &= \sum_{T \in M} \sum_{h \in T} w_h \beta_x(g_T, h)\rst{(g_Th)}{\hi_x}\rst{h}{\hi_x}\\
        &= \sum_{T \in M} \sum_{h \in T} w_h \beta_x(g_T, h)\rst{g_T}{\hi_x}\;.
    \end{subalign}
\end{proof}
\Cref{.symmetry_restriction} describes how the frustration $G$ can be simplified by removing false twins by projecting onto the stabilised spaces $\hi_x$, more specifically, the projection causes each set of false siblings to collapse into a single vertex, respectively. Next, we show how to use true twin rotations to further simplify the graph.
\begin{lemma}\label{.rotation}
    Let $g, h \in S$ be hermitian and anticommuting, and $a, b \in \R$. The operator $U = \ee^{\theta gh / 2}$, $\theta \in \R$, is unitary and it holds
\begin{align}
    U * \w(ag + bh) &= g \w(a \cos \theta + b \sin \theta) \nonumber\\
    &\, + h \w(b \cos \theta - a \sin \theta)\label{e"y_rotation}
\end{align}
\end{lemma}
\begin{proof}
    Set $\rho = -i gh$. $\rho$ is obviously hermitian and unitary; therefore $U = \ee^{i \theta \rho / 2}$ is unitary. Since $\rho^2 = \1$ it holds $U = \cos \frac{\theta}{2} + i \rho \sin \frac{\theta}{2}$, and it follows
    \begin{subalign}
        U * \w(ag + bh) &= \w(\cos^2 \frac{\theta}{2} - \sin^2 \frac{\theta}{2}) \w(ag + bh)
        \\
        &\, + i \cos \frac{\theta}{2} \sin \frac{\theta}{2} \w(a i h - i b g + i a h - i b
        g) \nonumber\\
        &= \w(ag + bh) \cos \theta + \w(- a h + b g)  \sin \theta\;,
    \end{subalign}
    which is the statement after sorting the terms.
\end{proof}
We can use the rotation in \cref{.rotation} to merge true twins while leaving the rest of the graph invariant; the projectors in the following are place-holders for the false twin projections:
\begin{proposition}\label{.merge_paulis}
    Let $G = (V, E)$ be the frustration graph of a Hamiltonian $H$, $g, h, f \in V$ such that $g$ and $h$ are true twins, and $a, b \in \R$. Furthermore let $P$ be a projector, such that $\bk{P, p} = 0$ for all $p \in \bc{g, h, f}$. Set $U = \ee^{\theta gh / 2}$ with $\theta =\arctan\w(b/a))$. Then we have
    \begin{align}
        U*(agP + bhP) &= \sqrt{a^2 + b^2} gP\;,\\
        U*(fP) &= fP\;.
    \end{align}
\end{proposition}
\begin{proof}
    It holds $b \cos \theta - a \sin \theta = 0$ and
    \begin{equation}
        a \cos \theta + b \sin \theta = \frac{a + \frac{b^2}{a}}{\sqrt{1 + \frac{b^2}{a^2}}} = \sqrt{a^2 + b^2}\;.
    \end{equation}
    Therefore, we have $U * (ag + bh) = \sqrt{a^2 + b^2} g$, according to \cref{.rotation}. Since $\bk{gh, P} = \bk{gh, f} = 0$ it follows
    \begin{subalign}
        U * (agP + bhP) &= \w(U * (ag + bh))(U * P)\\
        &= \sqrt{a^2 + b^2} g P
    \end{subalign}
    and analogously $U * (fP) = fP$.
\end{proof}
Applying \cref{.merge_paulis} iteratively allows us to collapse sets of true siblings, leaving the rest of the graph and Hamiltonian invariant.

\subsubsection{Block-Diagonalisation}\label{s"block_diag}

We now give a proof sketch of \cref{.result_informal}.
The idea is to apply the methods of twin collapse from \cref{s"twin_symmetries} alternately to simplify the graph. That is, we first collapse all sets of false siblings; then on the new graph, we collapse all sets of true siblings. We recursively continue this process until both the set of false and true siblings are empty. The detailed, constructive version of the proofs can be found in \cref{s"app_block_diag}. First, we need the graph sequence that defines the sets of siblings:
\begin{definition}[Twin Collapse]\label{.main_collapse_graph}
    Let $G = (V, E)$. Set $G^0 = G$. We define the following graph sequence $\w(G^i)_{0 \leq i \leq c}$, \en[c]:
    \begin{itemize}
        \item For odd \en[i]: For all maximal sets $T$ of false siblings in $G^{i-1}$, fix one of the siblings and remove the other vertices.
        \item For even \en[i], $i \geq 2$: For all maximal sets $T$ of true siblings, fix one of the siblings and remove the other vertices.
        \item Set \en[c] even and minimal, such that $G^c = G^{c+1}$, and $r = c/2 - 1$.
    \end{itemize}
\end{definition}
\begin{proof}[Proof sketch of \Cref{.result_informal}]
    The full proof can be found in \cref{s"app_block_diag}.

    Given a Hamiltonian $H$ with frustration graph $G = (V, E)$, define the graph sequence $\w(G^i)_{0 \leq i \leq c}$ as in \cref{.main_collapse_graph}.

    The idea is to define the projectors $\bc{P(x)}_x$ appropriately, such that they describe the alternation of false sibling projections and true sibling rotations according to the graph sequence $\w(G^i)_{0 \leq i \leq c}$. More specifically, they are constructed by conjugating  the projectors that describe the stabilised spaces by the false twins as in \cref{.symmetry_restriction}, with the true twin rotations from \cref{.merge_paulis}. One then shows that these projectors are orthogonal, complete and that they commute using the graph symmetries. The proof then follows inductively by showing that these projectors commute with the Hamiltonian.
\end{proof}

\section{Expansion of the Free-Fermions Class}\label{s"res_free_fermions}

We extend the graph-theoretic framework introduced in Ref.~\cite{chapman_unified_free_fermions} to characterise spin-$\tfrac{1}{2}$ Hamiltonians solvable via mappings to free fermions, in two distinct ways. The first generalization is in the subtle change in the definition of the frustration graph in \Cref{def"frustration_graph}. This generalisation lifts the prior restriction to Pauli operators, instead encompassing Hamiltonians expressed in any operator basis satisfying the characteristic binary commutation relations. This allows us to apply Fendley's solution method to Hamiltonians written in a broader family of representations without the need for fermion-to-qubit mappings including quantum chemistry Hamiltonians of interest written as interacting  fermions~\cite{gorman_electronic_structure_qma_complete}.

The second generalisation is how the block diagonalisation of \cref{s"twin_block_diag} expands class of the graphs that admit a free-fermion solution, which we discuss below.

We then numerically compare the free-fermionic solvability before and after applying our collapsing technique; specifically, we study Hamiltonians on a two-dimensional lattice with random interactions, as well as generic Majorana Hamiltonians, and apply graph-theoretical algorithms based on the modular decomposition tree \cite{habib_survey_modular_decomposition} of the frustration graph to remove twins and find \acl{scf} Hamiltonians.

\subsection{Free-Fermionic Solvability}

In Ref.~\cite{chapman_unified_free_fermions} it was show that a many-body quantum spin system
allows a description of non-interacting fermions (free-fermion) if the Hamiltonian is \acf{scf}; a
Hamiltonian is \ac{scf} if, and only if, its frustration graph contains no claws but at
least one simplicial clique (for each maximal connected subgraph). As discussed, our definition extends to Hamiltonians written in any basis with a binary commutation relation.
Our block diagonalisation method allows us to expand further the class of models which may be solved via a mapping to free fermions, as we may now increase the class of graphs to which Fendley's solution method may be applied~\cite{fendley2019free}.
\begin{corollary}
\label{.free_fermion}
A Hamiltonian $H$ is \textit{generically} free fermion if, after recursively collapsing all twin symmetries, as well as line-graph modules if $S$ is unitarily equivalent to the Pauli group, the frustration graph is simplicial and claw free.
\end{corollary}
Note that this result is generic in the weights of the Hamiltonian, i.e., it only depends
on the binary commutation relations of the operators.

\begin{figure}[b]
  \centering
  \includegraphics{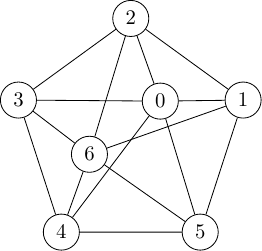}
  \caption{
    Example of a graph that is only simplicial after removing a twin vertex. Up to labelling, the non-empty cliques are $\bc{0}, \bc{1}, \bc{1, 2}, \bc{0, 1}, \bc{0, 1, 2}$, which are not simplicial. However, by removing, the vertex $6$, which is a false twin of $0$, the graph has, for example, the simplicial clique $\bc{1, 2}$.
  }\label{f"create_simplicial_clique}
  \vspace{-2ex}
\end{figure}

It is immediately obvious that collapsing all twins recursively does indeed extend the
class of free-fermion Hamiltonians: It is clear that collapsing twins does not introduce
new claws; further, the collapsing algorithm may remove claws. For example, let $G = K_{1,
3}$, then $G$ collapses to a single vertex: all leaves of the claw are false twins, after
collapsing the false twins, we are left with a complete graph on two vertices, since these
vertices are true twins, we may then collapse them to a single vertex using a unitary
rotation. It is also true that the collapsing algorithm strictly increases the probability
of a graph being simplicial. It is known that that simpliciality in claw-free graphs is
hereditary~\cite{chudnovsky_claw_free_independence_polynomial}, (for completeness we
provide a full proof of this in \cref{.preserve_simpliciality}), thus collapsing twins can
not remove a simplicial clique. However, removing twins may in fact introduce a simplicial
clique. Consider the graph depicted in \Cref{f"create_simplicial_clique}. The graph does
not contain a simplicial clique, however, the vertices $\{0,6\}$ are false twins. After
collapsing the false twins the graph is simplicial and claw free with the simplicial
clique being, for example, the subgraph induced by the labelled vertices $K_s=\{1,2\}$.

We focus here on the class of \ac{scf} Hamiltonians, as this class covers other classes like the class of line graphs~\cite{chapman_solvable_spin_models} or (even-hole, claw)-free graphs~\cite{elman_free_fermions_behind_the_disguise,fukai_free_fermions_with_claws}. However, we note that the twin-collapsing method may also be applied to these classes. For example, in the case of line graphs, collapsing twins does expand the class, as the potential removal of forbidden claws shows. Furthermore, the collapsing method can also remove symmetries induced by cycles in the root graph of the line graph (e.g., consider a triangle in the root graph), which may further simplify the solution (at least in the case of true twins). Regarding (even-hole, claw)-free graphs, it is clear that here too collapsing twins can expand the class as the simple example of a square as an even hole in the frustration graph shows.

\subsection{Numerical Simulations}\label{s"numerics}

We now quantify by how much the collapsing algorithm expands the class of Hamiltonians solvable by free-fermion methods. We do this by random sampling Hamiltonians for different Hamiltonian classes. Since the result of \cref{.free_fermion} pertains to Hamiltonians that are generically free, the coefficients in the Hamiltonians are not considered but only the operators in the Pauli basis. Given a Hamiltonian, we check whether its frustration graph is \ac{scf} or not. The algorithms are based on the modular decomposition tree of a graph as described in \cref{s"modular_decomposition}. The relevant code can be found at the \hr{https://github.com/taeruh/free_fermions}[taeruh/free\_fermion] repository (the plotted data was commited at \hr{https://github.com/taeruh/free_fermions/tree/91b7044326ec56ebf608dfa0a66c7e7930bdd6df}[91b7044]) \cite{complementarycode,complementarydata}; for the modular decomposition algorithm we used the library provided by~\cite{modular_decomposition_library,spinner_evaluation_of_modular_decomposition_algorithms}.

\subsubsection{Erdös R\'enyi Graphs}

As first example we discuss the ubiquitous Erdös R\'enyi graph model, or $\gnp$ model. The distribution of $\gnp(n, p)$ graphs, with \en and $p \in \bk{0, 1}$, is the distribution of graphs on $n$ vertices where each edge is drawn with probability $p$. For every $\gnp$ graph $G$ it is possible to construct a Hamiltonian such that its frustration graph is $G$, however, these models may not be of direct physical relevance.

The results for the numerical simulation are presented in \cref{f"gnp}.
We plot the probability $p_{\mathrm{SCF}}$ that a given Hamiltonian is \ac{scf} against the edge probability $p$, as well as the change in the order of the graph, $\Delta \Xi,$ and the increase in probability of a given model being free fermion due to the collapsing algorithm, $\Delta p_{\mathrm{SCF}}$. We see that $\pscf$ is close to symmetric around $p = 1/2$. This can be explained qualitatively by the following argument: Given four vertices, the expected number of claws is $4 p^3 (1-p)^3$. It is clear, this number approaches $0$ symmetrically for $p \to 0$ and $p \to 1$, in which case the graph is claw-free. For small $p$, the graph is likely to be sparsely connected. Thus, for any clique $K \subseteq \vertm(G)$, the neighbourhood $\neigh(v){\setminus} K$ for any $v \in K$ is likely to be small, and therefore likely to be fully connected, meaning $K$ is simplicial. Similarly, for $p$ close to $1$, the graph is likely to be densely connected, and therefore any neighbourhood is likely to be fully connected. This explains the symmetric form of $\pscf$. Similar arguments explain the symmetric form of the number of collapsed twins.
\begin{figure}[t]
    \centering
    \includegraphics{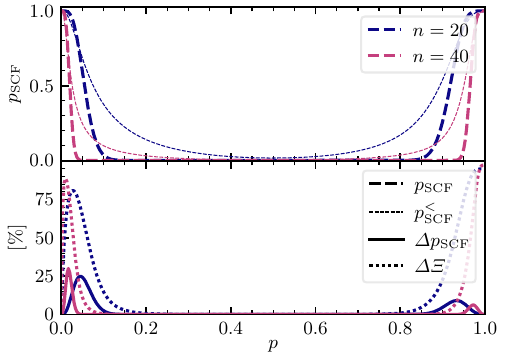}
    \caption{Probability, $p_{\mathrm{SCF}}$, that the $\gnp(n, p)$ Hamiltonians are \ac{scf}, and the effect of the block diagonalisation. The $x$-axis is the edge probability $p$. In the upper plot, the thicker dashed line, $p_{\mathrm{SCF}}$, shows the probability that the according Hamiltonian is \ac{scf} after the block diagonalisation. The thinner dashed line, $p_{\mathrm{SCF}}^<$, is the analytical upper bound on $p_{\mathrm{SCF}}$ before the block diagonalisation. In the lower plot, the dotted line, $\Delta \Xi$, shows how many vertices have been removed by the block diagonalisation (independent vertices in the frustration graph were ignored). The solid line, $\Delta p_{\mathrm{SCF}}$, shows the difference 
    between the number of Hamiltonians that are \ac{scf} before and after the block diagonalisation.} \label{f"gnp}
\end{figure}

The simplicity of the $\gnp$ model allows finding \ac{scf} lower and upper bounds, at least before the block diagonalisation, via the probabilistic first and second moment methods. The calculations are in \cref{s"erdos_renyi_bounds}. For small $p$, we have
\begin{equation}
    \pscf \geq \w(1 - \binom{n}{4} 4 p^3 (1-p)^3) (1-p)^{2np}\;.
\end{equation}
The first term in the parentheses is a lower bound for the probability that a $\gnp$ graph is claw-free, and the second term is a lower bound for the probability that a claw-free graph is simplicial. One can show that in the limit $n \to \infty$, if $p < \order[n^{-3/4}]$ or $1-p < \order[n^{-3/4}]$ the graph is almost surely claw-free and simplicial (cf. \cref{s"erdos_renyi_bounds}). As upper bound, for any $p$, we have
\begin{align}
    \pscf & \leq 1 - \binom{n}{4} \Big/ \bcdot{\binom{n-4}{4} + 4 \binom{n-4}{3}}\nonumber\\
    &+ \frac{3}{2} \binom{n-4}{2} \frac{1}{p(1-p)} + \frac{1}{4 p^3 (1-p)^3}\nonumber\\
    &+ \dotbc{\frac{n-4}{4} \w(\frac{3}{p^2 (1-p)} + \frac{1}{(1-p)^3})}\;.
\end{align}
This bound is plotted in \cref{f"gnp}; we see that for small and large $p$ the block diagonalisation overcomes this bound.

\subsubsection{Two-Dimensional Spin Lattices}

Next, we consider a 2-local spin Hamiltonian on a two-dimensional periodic brick lattice lattice, depicted in \cref{f"brick_lattice}. We define the Hamiltonian by assigning to each link of the lattice a linear combination of 2-local Pauli interaction, drawn uniformly at random from the set of all 2-local Pauli terms, with probability $p$. It is sufficient to only consider the lattice of the size shown, since larger lattices would repeat the same pattern. The chosen lattice is large enough to prevent periodic boundary structures in the frustration graph. This can be checked by drawing the line graph of the lattice, which depicts the range of neighbourhoods in the frustration graph: Potential operators are located on the vertices of the line graph, and potential anticommutators are located on the edges of the line graph; therefore, by ensuring that in the line graph a neighbourhood at the periodic boundaries does not overlap with itself, we ensure that this does not happen in the frustration graph.
\begin{figure}[t]
    \centering
    \includegraphics{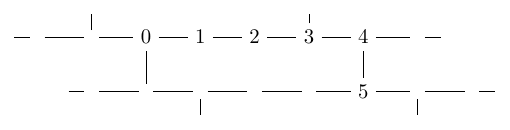}
    \caption{Two-dimensional periodic brick lattice modelling the Hamiltonians. The spins are located on the vertices and we randomly draw 2-local Pauli interactions between neighboured spins. Each of the nine Pauli interactions is accepted with a certain probability $p \in (0, 1]$; the probability is the same for all edges, and we enforce that each edge has at least one non-trivial interaction. We draw the interactions separately on the edges between the vertices $1$ to $5$ and then extend the lattice periodically with periodic boundary conditions. } \label{f"brick_lattice}
\end{figure}
\begin{figure}
    \centering
    \includegraphics{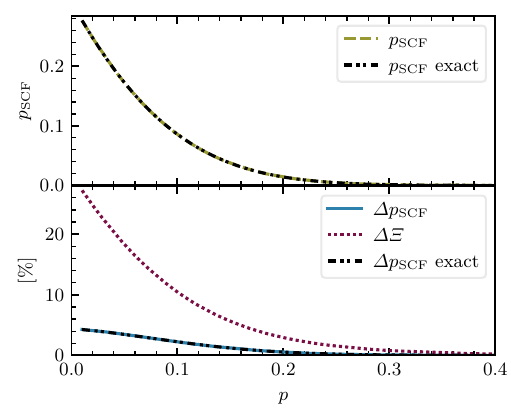}
    \caption{\ac{scf} probability for Hamiltonians on a periodic brick lattice. The $x$-axis is the probability $p$ that a Pauli interaction is accepted on an edge (cf. \cref{f"brick_lattice}). $p_{\mathrm{SCF}}$, $\Delta p_{\mathrm{SCF}}$ and $\Delta \Xi$ are as in \cref{f"gnp}. The dash-dotted lines are exact results for $p_{\mathrm{SCF}}$ and $\Delta p_{\mathrm{SCF}}$, respectively (cf. \cref{s"brick_exact}). For higher densities than the ones shown here, all plots are $0$ or close to $0$.}\label{f"numerical_results}
\end{figure}

The results for the numerical simulations for the lattice are presented in \Cref{f"numerical_results}. We see that the probability of a given Hamiltonian being \ac{scf} decreases with increasing interaction probability $p$; this is expected as one can easily see that the number of claws increases with $p$. Interestingly, we also observed that for all the Hamiltonians we drew, if the frustration graph was claw-free, then it also had a simplicial clique; it would be interesting to investigate this correspondence further for different physical lattice structures in future work. Furthermore, we see that with increasing $p$, relatively fewer terms are removed by the collapsing algorithm, and correspondingly the expansion of the class of free-fermion Hamiltonians is less effective; again, it is easy to see that as the probability approaches 1 the frustration graph does not contain any twins. When the interactions are sparse, the removal of twins expands the free-fermion class by approximately $4\%$.

As the unit cell of periodic brick lattice is small enough, it is possible to calculate exact \ac{scf} probabilities (cf. \cref{s"brick_exact}). We see that these probabilities are nearly identical with the sampled results, verifying the accuracy of the sampling method.

While the results for the periodic brick lattice are intuitive, the considered model itself may be argued to be artificial. Next we consider the model of a periodic square lattice. We draw the Hamiltonian similarly as in \cref{f"brick_lattice}, however on a periodic square lattice. Additionally, each vertex is coupled to an additional local spin (via a two-local interaction), e.g., to a local nuclei. The calculated results are shown in \cref{f"square_lattice}. We see that the model is very unlikely to be \ac{scf}; only for small probabilities $p$ there is a small chance that the model is \ac{scf} ($p_{\mathrm{SCF}} \sim 10^{-2}$). This is because the edges in the physical square lattice are on both ends connected to vertices with at least degree $4$, which is very likely to produce claws. In contrast, the edges in brick lattice are connected to vertices that have on average a lower degree.
\begin{figure}[t]
    \centering
    \includegraphics{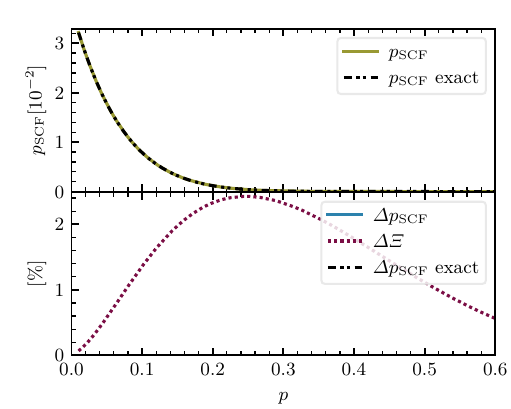}
    \caption{\ac{scf} probability for a Hamiltonian on a periodic square lattice with local nuclei. The Hamiltonian is drawn as in \cref{f"brick_lattice} but on square lattice instead of the brick lattice. As in \cref{f"brick_lattice} we enforce that the model is two-dimensional. We include an additional interaction for each vertex to a local spin, e.g., a nuclei (also drawn with $p$); the interaction is the same for all vertices. $p_{\mathrm{SCF}}$, $\Delta p_{\mathrm{SCF}}$ and $\Delta \Xi$ are as in \cref{f"gnp}; the dash-dotted lines are exact results for $p_{\mathrm{SCF}}$ and $\Delta p_{\mathrm{SCF}}$, respectively (cf. \cref{s"square_exact}). $\Delta p_{\mathrm{SCF}}$ is always $0$ and the other lines are also close to $0$ for higher densities.}\label{f"square_lattice}
\end{figure}

\begin{figure}[t]
  \centering  \includegraphics{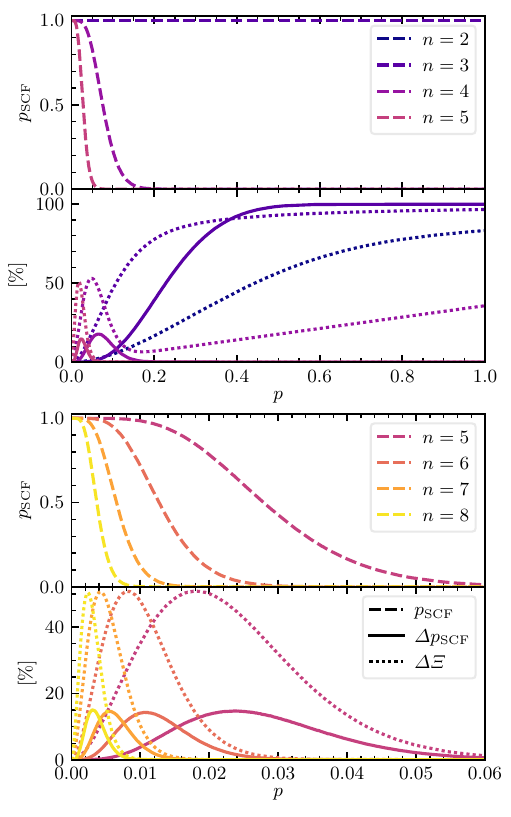}
  \caption{%
    Probabilty, $p_{\mathrm{SCF}}$, that the Majorana Hamiltonian in \cref{e"majorana_hamiltonian} is \ac{scf} depending on the interaction probability $p$. The plots are drawn for different numbers of orbitals $n$ and we label them analogously as in \cref{f"numerical_results}. Note that the solid line for $n=2$ is constantly $0$.}\label{f"e_structure}
\end{figure}

\subsubsection{Electronic Structure Hamiltonians}

As previously alluded to, our results are not restricted to spin systems. We can also, for example, apply the graph-theoretic formalism directly to Hamiltonians written in the Majorana basis. Majorana fermions are self-adjoint operators, defined as
\begin{equation}
    \gamma_{2j-1}=\frac{1}{2}(c_j^\dagger+c_j), \qquad\gamma_{2j}=\frac{-i}{2}(c_j^\dagger-c_j),\label{e"majorana_operators}
\end{equation}
where $\{c_j^{(\dagger)}\}_{j\in \Z_n}$ for some \en, are complex fermionic creation and annihilation operators. Majoranas obey the canonical anticommutation relations
\begin{equation}
    \{\gamma_i,\gamma_j\}=2\delta_{ij},
\end{equation}
meaning the graph-theoretic framework is appropriate for Majorana models. Furthermore, they are unitarily equivalent to the Pauli group (cf. \cref{s"conjugation_isomorphism}). As a showcase, we run similar simulations on the following Majorana Hamiltonian:
\begin{equation}
  H_{\mathrm{M}} = \frac{i}{2}\sum_{a, b = 1}^{2n} w_{ab} \gamma_a \gamma_b + \sum_{a, b, c, d =
  1}^{2n} w_{abcd} \gamma_a \gamma_b \gamma_c \gamma_d\;,\label{e"majorana_hamiltonian}
\end{equation}
where \en is the number of complex fermionic orbitals and the weights $w_{ab}, w_{abcd} \in \R{\setminus}\{0\}$ are non-zero with a probability $p \in \bk{0, 1}$. This Hamiltonian encompasses, for example the electronic-structure Hamiltonian which, in general, has been shown to be \acs{qma}-complete \cite{gorman_electronic_structure_qma_complete}. Here, we investigate the likelihood of the model being free fermion as a function of the total number of orbitals $n$, as well as the probability $p$ of drawing a given Majorana string (\cref{f"e_structure}).

As expected, with increasing complexity of the model (that is, increasing number of
orbitals and interaction probability), the probability that the $H_{\text{M}}$ is \ac{scf}
decreases. We also observe the effect of the twin collapse is smaller for increased $n$.
For smaller orbital numbers however, especially $n \le3$, we observe that,
counter-intuitively, the reduction in the order of the graph $\Delta\Xi$, increases with
$p$. For $n = 2$ the frustration graph is always an induced subgraph of the octahedral
graph which is \ac{scf} (ignoring the independent vertex due to
the four-body interaction, since it necessarily commutes with all other terms). Moreover, the octahedral graph is a \textit{cograph} and therefore fully collapses to a single vertex (see \cref{f"octahedral}), which explains why $\Delta \Xi$ converges to $\frac{5}{6}$ as $p$ approaches 1.

For $n = 3$ we observe a large difference in $p_{\mathrm{SCF}}$ due to twin collapse, in fact, after the twin collapse the Hamiltonian is always \ac{scf}. This can be explained as follows: firstly consider the frustration created by all two-body interaction, $\sum_{a, b = 1}^{6}\gamma_a \gamma_b$; this graph has no twins, however, it is a line graph by definition~\cite{chapman_solvable_spin_models} and therefore \ac{scf}~\cite{chudnovsky_claw_free_independence_polynomial,elman_free_fermions_behind_the_disguise}, as are all induced subgraphs. More specifically, our block-diagonalisation algorithm even collapses the line graph into a single vertex. Now, if we have a four-local term in the Hamiltonian, e.g., $\gamma_1\gamma_2\gamma_3\gamma_4$, it is easy to show that this is a false twin of $\gamma_5\gamma_6$. Therefore, after the twin collapse, this term has been removed (or can be effectively replaced by $\gamma_5\gamma_6$ in the frustration graph). After removing all four-local terms, the Hamiltonian is a line graph and therefore collapses further onto a single vertex. This also explains why we see $\Delta \Xi \to \frac{29}{30}$ for $p \to 1$.

In \cref{s"majorana_analytical}, we argue that in the limit $n \to \infty$, the Hamiltonian $H_{\mathrm{M}}$ is almost surely simplicial, claw-free, if the number of operators in $H_{\mathrm{M}}$ is upper bounded by $\order[n^{3/4}]$.

\begin{figure}[t]
  \centering
  \includegraphics{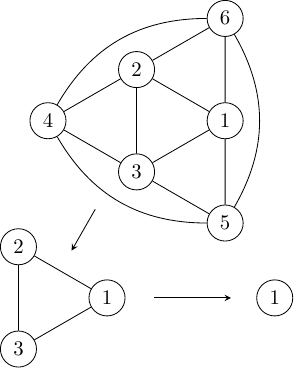}
  \caption{
    Recursive twin collapse of the octahedral graph \cite{weisstein_octahedral_graph}. The shown graph with six vertices has three pairs of false twins, $\bc{1, 4}$ $\bc{2, 5}$ and $\bc{3, 6}$. After collapsing those, the new graph consists only of a single true sibling sets. Collapsing this set results in a single vertex. 
  }\label{f"octahedral}
  \vspace{-2ex}
\end{figure}

In \cref{s"uniform_random_paul_strings}, we discuss models where we draw uniformly random Pauli Strings.

\section{A Variation of the Discrete Stone-von Neumann Theorem}
\label{s"res_stone_von_neumann}

We were able to apply the full block-diagonalisation technique in \cref{.full_result_informal} on Hamiltonians written in terms of Majorana operators since the Majorana group is unitarily equivalent to the Pauli group as stated by the discrete Stone–von Neumann theorem~\cite{heinrich_stabiliser_techniques} (cf.~\cref{x"pauli_to_majorana}). Below, we consider a generalisation of this theorem, characterising the conditions under which groups --- such as the Pauli group or, more generally, the Weyl–Heisenberg group --- are isomorphic, potentially via unitary conjugation. A comprehensive treatment, including proofs, generalisations and further mathematical details, is provided in \cref{s"app_isomorphisms}.

Let us first give a definition of the groups we are interested in, and then show under which conditions there exists a unitary conjugation between them.

\subsection{The Polar Commutator Group}

Let \en[d,~n,~N], with $d$ prime, and $\omega \in \C$ be a primitive $d$-th root of unity. Let $F \subset \C^\times$ be a fixed set of representatives of the multiplicative quotient group $\C^\times\!/\braket{\omega}$; for example, ${F = \set{r \ee^{i\phi}}{r \in \R{\setminus}\bc{0}, \phi \in [0, 2\pi/d)}}$. Without loss of generality let $1 \in F$, and let $r: \C^\times \to F$, $u: \C^\times \to \Z_d$ such that $a = r(a) \omega^{u(a)}$ for all $a \in \C^\times$. We now define the polar commutator group as:

\begin{definition}[The polar commutator group]\label{.main_polar_group}
Let ${W \in \M_n\w(\Z_d)}$ and set ${\Omega = W - W\ut}$. The \emph{polar commutator group} is the tuple ${\gk_d^n(W) = \w(F, \Z_d, \Z_d^n, \cdot)}$ with multiplication defined as
\begin{align}
    \cdot: \gk_d^n \times \gk_d^n &\to \gk_d^n\;,\\
    (a, p, x), (b, q, y) &\mapsto (r(ab), u(ab) + p + q + x\ut W y, x + y)\nonumber
\end{align}
\end{definition}
We may then define the \textit{polar commutator representation}
\begin{equation}
    \begin{split}
        \mu: \gk_d^n(W) &\to \GL\w(\C^N),\\
        (a, p, x) & \mapsto a \omega^p\tau(x),
    \end{split}
    \label{e"main_representation}
\end{equation}
where $\tau : \Z_d^n \to \GL\w(\C^N)$ is a mapping into the general linear group of $\C^N$
such that $\mu$ is a multiplicative monomorphism. Where clear by context, we drop indices and argument and write $\gk$ instead of the full form, $\gk_d^n(W)$. As shorthand, we write $\mu(\cdot, \cdot, \cdot) \coloneqq \mu((\cdot, \cdot, \cdot))$ and we call $\mu(\gk)$ the representation of $\gk$ via $\tau$.

Let us now consider the form of elements from $\gk_d^n(W)$, $\w(a, p, x)$. The first component allows us to include scalar factors up to multiples of $\omega$ in the representation; here, we allow the scalars to be any complex numbers, but one can also generalise it to multiplicative subgroups of $\C^\times$ (see \cref{s"app_isomorphisms}). The second component, $p\in\Z_d$ then accounts for multiples of $\omega$ as well as the commutation rules via $W$. The third component describes the group elements in the representation via $\tau$. Characteristically, the commutator of the representation is always a scalar, more specifically, it is restricted to the roots of unity, hence the name of the group: For $(a, p, x), (b, q, y) \in \gk$, we have (see \cref{s"app_isomorphisms})
\begin{equation}
  \tbk{\mu(a, p, x), \mu(b, q, y)} = \omega^{x\ut \Omega y},
\end{equation}
and the commutator Lie bracket is given by
\begin{equation}
  \bk{\mu(a, p, x), \mu(b, q, y)} = \w(1 - \omega^{-x\ut \Omega y}) \mu(a, p, x)\mu(b, q,
  y)\;.
\end{equation}
Note that because of this, closing a subset $X$ of $\mu(\gk)$ under the Lie bracket is equivalent, up to scalar factors, to simply closing $X$ under multiplication.


Given a group which equals the representation $\mu\w(\gk(W))$ for some $W \in \M_n(\Z_d)$, $\gk$ allows us to handle this group in the $\Z_d^n$ vector space additionally with $\Z_d$ to capture the commutation rules (instead of $\GL\w(\C^N)$).

We can describe the Pauli group as representation of polar commutator group: For some
\en[m], set $N = 2^m$, $n = 2m$, $d = 2$, $W = \smatp{0 & 0 \\ -\1^{m\times m} & 0} \in
\M_n\w(\Z_d)$ and define $\tau$ as
\begin{equation}
\begin{split}
  \tau: \Z_d^n \cong \Z_d^m \times \Z_d^m &\to \GL\w(\C^{2^m}),\\ (z, x) &\mapsto
  \bigotimes_{j=1}^m Z_j^{z_j} X_j^{x_j}.
  \end{split}
\end{equation}
The representation $\mu$ as in \cref{e"main_representation} then gives us the Pauli group with complex prefactors. This description of the Pauli group is also known as the tableau or Heisenberg description; or rather, the Pauli group is a representation of the Heisenberg group $\hei_{2m}(\Z_2)$~\cite{gross_hudsons_theorem_finite_quantum_system,heinrich_stabiliser_techniques}. In \cref{s"weyl_heisenberg_and_parafermions}, we describe the more general Weyl-Heisenberg group in terms of $\gk$ and $\mu$. The Pauli group, with $\Omega$ being the standard symplectic form, is the canonical example of the polar commutator group, for $d = 2$, corresponding to physical spin-$\frac{1}{2}$ systems.

Another group that we can describe is the group of Majorana operators (or more general,
the parafermions; cf \cref{s"weyl_heisenberg_and_parafermions}): Again, for some \en[m],
we set $N = 2^m$, $n = 2m$, and $d = 2$, but now $W = P$, where $P$ is the parity matrix,
i.e, $P_{ij} = 1$ if $i > j$ and $0$ otherwise for all $i, j \in \bc{1, \ldots, n}$, and
\begin{equation}
    \begin{split}
        \tau: \Z_d^n &\to \GL\w(\C^{2^m}),\\
        x &\mapsto\gamma_1^{x_1} \cdots \gamma_n^{x_n},
    \end{split}
    \label{e"majorana_representation}
\end{equation}
with the Majorana operators define as in \cref{e"majorana_operators}.
We see that the only difference between the Pauli group and the Majorana group is the Matrix $W$. 

\subsection{The Conjugation Isomorphism}\label{s"conjugation_isomorphism}

We now discuss under which requirements different polar commutator groups are isomorphic to each other. The proof is in \cref{s"app_isomorphisms}. As explicit example we show that there is a unitary equivalence between the Pauli and the Majorana group. This equivalence is already known as a special case of the Stone–von Neumann theorem~\cite{heinrich_stabiliser_techniques}, under which both groups form irreducible representations of the Heisenberg group (compare, for example, Ref.~\cite{room_synthesis_of_clifford_matrices}, which implicitly defines the Majorana group as a representation of the Heisenberg group; cf.~\cref{x"pauli_to_majorana}). Nonetheless, we present it here as an illustrative application of \cref{.main_isomorphism} due to its familiarity (in \cref{s"weyl_heisenberg_and_parafermions} we discuss the general case of the Weyl–Heisenberg and parafermion groups).
\begin{theorem}\label{.main_isomorphism}
    Let $N_1,N_2\in \N$, $W_1,W_2 \in \M_n\w(\Z_d)$, and $\gk_d^n(W_1)$, $\gk_d^n(W_2)$ with representations $\mu_1$ and $\mu_2$, respectively. If $\Omega_i = W_i - W_i\ut$, $i \in \bc{1, 2}$, both have full rank, i.e., are symplectic, then there exists an isomorphism
    \begin{equation}
      \phi: \gk = \gk_d^n(W_1) \to \gk_d^n(W_2)\;.
    \end{equation}
    Furthermore, if the representations $\mu_1$ and $\mu_2 \circ \phi$ have the same character, i.e, the same trace of the representation, and it exists a set $B \subseteq \gk$ such that $\mu_1(B)$ is a set of hermitian (or unitary) generators of the vector space $\M_{N_1}(\C)$ and $\mu_2(\phi(g))$ is hermitian (unitary) for all $g \in B$, then it exists $S \in \U\w(N)$, where $N = N_1 = N_2$, such that $\mu_2\w(\phi(g)) = S \mu_1(g) S\ui$ for all $g \in \gk$.
\end{theorem}
The first statement of \cref{.main_isomorphism} shows that many of the $\gk$ groups are the same from a group theoretical point of view. The canonical representation is probably the Weyl-Heisenberg group, where the commutator matrix $\Omega$ is the standard symplectic form and representatives of the group of a unitary Schmidt inner product. In this case of $d = 2$, the basis is additionally hermitian. More properties of the Weyl-Heisenberg are detailed in \cref{.weyl_heisenberg_properties}.

The second statement is particularly interesting for applications in quantum mechanics, namely, that we can map between different operator groups, preserving quantum expectation values.

We end this section by showing that the theorem in its full form applies, for example, to the Pauli group and the Majorana group (see \cref{.parafermion_weyl_isomorphism} for an alternative generalised proof). To see this we have to check the conditions in \cref{.main_isomorphism}: In the case of the Pauli group, we already know that the commutator matrix $\Omega$ is symplectic, and that the group contains a unitary basis; furthermore, the trace, i.e., the character, is zero for all non-identity elements. Regarding the Majorana group, it is clear that the strings in \cref{e"majorana_representation} are unitary, i.e., all Majoranas are unitary when we restrict the scalar prefactors in the groups to have absolute value $1$; furthermore, non-identity elements have trace zero: if the string has even support, just cycle one element from the front to the back and then use the cyclic property of the trace, and if the support is odd, there exists an $i \in \bc{1, \ldots, m}$ such that either $\gamma_{2i-1}$ or $\gamma_{2i}$ is in the string, but not both, which also implies that the trace is zero, since $\braket*{\lambda}[(c_i\ur \pm c_i)][\lambda]>i> = 0$, where $\ket{\lambda}_i \in \bc{\ket{0}_i, \ket{1}_i}$ is the fermionic number basis for the $i$th orbital. With that, the Majorana group has the same character as the Pauli group. It remains to show that the commutator matrix $\Omega$ of the Majorana group has full rank , but this is equivalent to saying that for each Majorana string, there exists another one that anticommutes with it; the argument for this is the same as for the trace: if the string has even support, then any $\gamma_{j}$ with $j$ being in the support, anticommutes with the string, and if the support is odd, it exists an $i \in \bc{1, \ldots, n}$ such that $\gamma_{i}$ is not in the string - this element anticommutes with the string. With that, all conditions of \Cref{.main_isomorphism} are fulfilled, and therefore, there exists an unitary $S \in \GL\w(\C^{2^m})$ that maps the Pauli group to the Majorana group under conjugation.

\section{Discussion}

We have presented a graph-theoretic method to simplify Hamiltonians written in the Pauli basis, or any operator basis with binary commutation rules. The block-diagonalisation recursively removes terms in the Hamiltonian that correspond to siblings in the frustration graph, simplifying the Hamiltonian. Our twin-collapsing approach can be seen as an extension of an important result of Ref.~\cite{Kirby2020}, which leads to the reduction in complexity of a class of Hamiltonians defined by certain commutation structures. An interesting further direction would be to fully explore how many more classes of Hamiltonians can be reduced in complexity due to these results, or whether such results extend to systems made of qudits rather than qubits~\cite{mann_free_parafermions}. In the latter case, the frustration graph is a directed graph (and the model has a free-parafermion solution if the frustration graph is an oriented indifference graph), and false twins are defined as before, with the additional requirement that for each shared neighbour of the twins, one of the two edges from the twins must be directed to the neighbour and the other away from the neighbour.

Another immediate application of this method is in the recognition and expansion of the class of free-fermion Hamiltonians. Numerical simulations show that the collapsing algorithm  applied to spin Hamiltonians on a two-dimensional periodic brick lattice can remove up to approximately $26\%$ of the terms in the Hamiltonian, when interactions are sparse. This leads to approximately $4\%$ more free-fermion Hamiltonians in that case. Since our collapsing algorithm works through modular decomposition of the frustration graph, an immediately apparent extension to our work would be to investigate which other collections of terms may be identified through their graphical structures that can be removed through unitary or projective means. While this work has focused on \textit{generic} free fermion solutions, this may help in the pursuit of a general theory that includes non-generic models; that is, models which admit a free-fermion solution only for finely tuned coefficients~\cite{fendley2024free,fukai_free_fermions_with_claws}. A first step would be to identify the family of unitaries that preserve claw-free-ness of a frustration graph.

As we have shown, in the special case where the frustration graph of the Hamiltonian is a \textit{cograph}, the block-diagonalisation through twin collapse results in a \textit{full diagonalisation} of the model. Further work could investigate whether this could lead to a general diagonalisation technique where one manipulates or perturbs a general Hamiltonian such that its frustration graph becomes a cograph.

We have also presented a variation of the discrete Stone-von Neumann theorem and studied a family of groups which can be used as drop-in replacements of the Pauli group in the previous results, broadening the application of the graph-theoretical methods. The groups are characterised by that they are nearly Abelian in the sense that the commutator is restricted to roots of unity. We showed that if the commutator defining matrix is symplectic then the groups are isomorpic to each other, potentially under a unitary conjugation. While the presented application of the theorem to the Weyl-Heisenberg group and the parafermion group recovers already known results, we believe theory gives insights into those groups from an abstract (but potentially simpler, especially for numerical calculations) point of view. It may also help in characterising similar groups in the future or to define new groups with different commutator matrices and represent them via the isomorphism to the Weyl-Heisenberg group.

\section*{Acknowledgements}
The authors would like to thank Adrian Chapman and Ryan Mann for insightful discussions. Jannis Ruh was supported by the Sydney Quantum Academy, Sydney, NSW, Australia.

\bibliography{literature.bib}

@misc{aguilar2024full,
  title = {{Full classification of Pauli Lie algebras}},
  author = {Aguilar, Gerard and Cichy, Simon and Eisert, Jens and Bittel,
            Lennart},
  eprint = {2408.00081},
  archivePrefix = {arXiv},
  year = {2024},
}

@article{backens2019jordan,
  title = {{Jordan-Wigner transformations for tree structures}},
  author = {Backens, Stefan and Shnirman, Alexander and Makhlin, Yuriy},
  journal = {Scientific Reports},
  volume = {9},
  number = {1},
  pages = {2598},
  year = {2019},
  publisher = {Nature Publishing Group UK London},
  doi = {10.1038/s41598-018-38128-8},
  archivePrefix = {arXiv},
  eprint = {1810.02590},
}

@article{batista2001generalized,
  title = {{Generalized Jordan-Wigner transformations}},
  author = {Batista, CD and Ortiz, Gerardo},
  journal = {Physical Review Letters},
  volume = {86},
  number = {6},
  pages = {1082},
  year = {2001},
  publisher = {APS},
  doi = {10.1103/physrevlett.86.1082},
  archivePrefix = {arXiv},
  eprint = {cond-mat/0008374},
}

@article{bolt_clifford_stuff_1,
  title = {On the {Clifford} collineation, transform and similarity groups. {I}},
  volume = {2},
  DOI = {10.1017/S1446788700026379},
  number = {1},
  journal = {Journal of the Australian Mathematical Society},
  publisher = {Cambridge University Press},
  author = {Bolt, Beverley and Room, T. G. and Wall, G. E.},
  year = {1961},
  pages = {60–79},
}

@article{bolt_clifford_stuff_2,
  title = {On the {Clifford} collineation, transform and similarity groups. {II}
           },
  volume = {2},
  DOI = {10.1017/S1446788700026380},
  number = {1},
  journal = {Journal of the Australian Mathematical Society},
  publisher = {Cambridge University Press},
  author = {Bolt, Beverley and Room, T. G. and Wall, G. E.},
  year = {1961},
  pages = {80–96},
}

@article{bolt_clifford_stuff_3,
  title = {On the {Clifford} collineation, transform and similarity groups. {III
           }. {Generators} and involutions},
  volume = {2},
  DOI = {10.1017/S1446788700026926},
  number = {3},
  journal = {Journal of the Australian Mathematical Society},
  publisher = {Cambridge University Press},
  author = {Bolt, Beverley},
  year = {1962},
  pages = {334–344},
}

@article{bravyi2002fermionic,
  title = {{Fermionic quantum computation}},
  author = {Bravyi, Sergey B and Kitaev, Alexei Yu},
  journal = {Annals of Physics},
  volume = {298},
  number = {1},
  pages = {210--226},
  year = {2002},
  publisher = {Elsevier},
  doi = {10.1006/aphy.2002.6254},
  archivePrefix = {arXiv},
  eprint = {quant-ph/0003137},
}

@article{bravyi2005Commutative,
  author = {Bravyi, Sergey and Vyalyi, Mikhail},
  title = {{Commutative version of the local Hamiltonian problem and common
           eigenspace problem}},
  year = {2005},
  publisher = {{Rinton Press, Incorporated}},
  address = {{Paramus, NJ}},
  volume = {5},
  number = {3},
  issn = {1533-7146},
  journal = {{Quantum Info. Comput.}},
  url = {https://www.rintonpress.com/journals/doi/QIC5.3-2.html},
  pages = {187–215},
  eprint = {quant-ph/0308021},
  archivePrefix = {arXiv},
  primaryClass = {quant-ph},
}

@article{chapman_solvable_spin_models,
  doi = {10.22331/q-2020-06-04-278},
  title = {{Characterization of solvable spin models via graph invariants}},
  author = {Chapman, Adrian and Flammia, Steven T.},
  journal = {{Quantum}},
  volume = {4},
  pages = {278},
  year = {2020},
  archivePrefix = {arXiv},
  primaryclass = {quant-ph},
  eprint = {2003.05465},
}

@misc{chapman_unified_free_fermions,
  title = {{A Unified Graph-Theoretic Framework for Free-Fermion Solvability}},
  author = {Adrian Chapman and Samuel J. Elman and Ryan L. Mann},
  year = {2023},
  eprint = {2305.15625},
  archivePrefix = {arXiv},
  primaryClass = {quant-ph},
}

@article{chen2018exact,
  title = {{Exact bosonization in two spatial dimensions and a new class of
           lattice gauge theories}},
  author = {Chen, Yu-An and Kapustin, Anton and Radi{\v{c}}evi{\'c}, {DJ}or{dj}e
            },
  journal = {Annals of Physics},
  volume = {393},
  pages = {234--253},
  year = {2018},
  publisher = {Elsevier},
  doi = {10.1016/j.aop.2018.03.024},
  archivePrefix = {arXiv},
  eprint = {1711.00515},
}

@article{chen2019bosonization,
  title = {{Bosonization in three spatial dimensions and a 2-form gauge theory}},
  author = {Chen, Yu-An and Kapustin, Anton},
  journal = {Physical Review B},
  volume = {100},
  number = {24},
  pages = {245127},
  year = {2019},
  publisher = {APS},
  doi = {10.1103/physrevb.100.245127},
  archivePrefix = {arXiv},
  eprint = {1807.07081},
}

@article{chudnovsky_claw_free_independence_polynomial,
  title = {{The roots of the independence polynomial of a clawfree graph}},
  journal = {Journal of Combinatorial Theory, Series B},
  volume = {97},
  number = {3},
  pages = {350-357},
  year = {2007},
  issn = {0095-8956},
  doi = {10.1016/j.jctb.2006.06.001},
  url = {https://www.sciencedirect.com/science/article/pii/S0095895606000876},
  author = {Maria Chudnovsky and Paul Seymour},
}

@article{chudnovsky_growing_without_cloning,
  title = {{Growing without cloning}},
  author = {Maria Chudnovsky and Paul Seymour},
  year = {2012},
  month = {sep},
  doi = {10.1137/100817255},
  volume = {26},
  pages = {860--880},
  journal = {SIAM Journal on Discrete Mathematics},
  issn = {0895-4801},
  publisher = {Society for Industrial and Applied Mathematics Publications},
  number = {2},
}

@article{elman_free_fermions_behind_the_disguise,
  author = {Elman, Samuel J. and Chapman, Adrian and Flammia, Steven T.},
  title = {{Free Fermions Behind the Disguise}},
  journal = {Communications in Mathematical Physics},
  year = {2021},
  month = {Dec},
  volume = {388},
  number = {2},
  pages = {969-1003},
  doi = {10.1007/s00220-021-04220-w},
  archivePrefix = {arXiv},
  primaryclass = {quant-ph},
  eprint = {2012.07857},
}

@article{fendley_parafermionic_edge_zero_modes,
  doi = {10.1088/1742-5468/2012/11/P11020},
  url = {https://doi.org/10.1088/1742-5468/2012/11/P11020},
  year = {2012},
  month = {nov},
  publisher = {IOP Publishing and SISSA},
  volume = {2012},
  number = {11},
  pages = {P11020},
  author = {Fendley, Paul},
  title = {Parafermionic edge zero modes in Zn-invariant spin chains},
  journal = {Journal of Statistical Mechanics: Theory and Experiment},
}

@article{fendley2019free,
  title = {{Free fermions in disguise}},
  author = {Fendley, Paul},
  journal = {Journal of Physics A: Mathematical and Theoretical},
  volume = {52},
  number = {33},
  pages = {335002},
  year = {2019},
  publisher = {IOP Publishing},
  doi = {10.1088/1751-8121/ab305d},
  archivePrefix = {arXiv},
  eprint = {1901.08078},
}

@article{fendley2024free,
  title = {{Free fermions beyond Jordan and Wigner}},
  author = {Fendley, Paul and Pozsgay, Bal{\'a}zs},
  journal = {SciPost Physics},
  volume = {16},
  number = {4},
  pages = {102},
  year = {2024},
  doi = {10.21468/SciPostPhys.16.4.102},
  archivePrefix = {arXiv},
  eprint = {2310.19897},
}

@article{fradkin1989jordan,
  title = {{Jordan-Wigner transformation for quantum-spin systems in two
           dimensions and fractional statistics}},
  author = {Fradkin, Eduardo},
  journal = {Physical Review Letters},
  volume = {63},
  number = {3},
  pages = {322},
  year = {1989},
  publisher = {APS},
  doi = {10.1103/physrevlett.63.322},
}

@misc{gorman_electronic_structure_qma_complete,
  title = {{Electronic Structure in a Fixed Basis is {QMA}-complete}},
  author = {Bryan O'Gorman and Sandy Irani and James Whitfield and Bill
            Fefferman},
  year = {2021},
  eprint = {2103.08215},
  archivePrefix = {arXiv},
  primaryClass = {quant-ph},
  url = {https://arxiv.org/abs/2103.08215},
}

@article{gross_hudsons_theorem_finite_quantum_system,
  author = {Gross, D.},
  title = {{Hudson’s theorem for finite-dimensional quantum systems}},
  journal = {Journal of Mathematical Physics},
  volume = {47},
  number = {12},
  pages = {122107},
  year = {2006},
  month = {12},
  issn = {0022-2488},
  doi = {10.1063/1.2393152},
  archivePrefix = {arXiv},
  eprint = {quant-ph/0602001},
}

@article{habib_survey_modular_decomposition,
  title = {{A survey of the algorithmic aspects of modular decomposition}},
  journal = {Computer Science Review},
  volume = {4},
  number = {1},
  pages = {41-59},
  year = {2010},
  issn = {1574-0137},
  doi = {10.1016/j.cosrev.2010.01.001},
  author = {Michel Habib and Christophe Paul},
  archivePrefix = {arXiv},
  primaryClass = {cs.DM},
  eprint = {0912.1457},
}

@phdthesis{heinrich_stabiliser_techniques,
  school = {Universität zu Köln},
  title = {{On stabiliser techniques and their application to simulation and
           certification of quantum devices}},
  author = {Markus Heinrich},
  year = {2021},
  url = {https://kups.ub.uni-koeln.de/50465/},
}

@article{huerta1993bose,
  title = {{Bose-Fermi transformation in three-dimensional space}},
  author = {Huerta, Luis and Zanelli, Jorge},
  journal = {Physical Review Letters},
  volume = {71},
  number = {22},
  pages = {3622},
  year = {1993},
  publisher = {APS},
  doi = {10.1103/physrevlett.71.3622},
  archivePrefix = {arXiv},
  eprint = {cond-mat/9308006},
}

@article{jordan_wigner_transformation,
  author = {Jordan, P. and Wigner, E.},
  title = {{Über das Paulische Äquivalenzverbot}},
  journal = {Zeitschrift für Physik},
  year = {1928},
  month = {Sep},
  day = {01},
  volume = {47},
  number = {9},
  pages = {631-651},
  issn = {0044-3328},
  doi = {10.1007/BF01331938},
  url = {https://doi.org/10.1007/BF01331938},
}

@article{kaufman1949crystal,
  title = {{Crystal statistics. II. Partition function evaluated by spinor
           analysis}},
  author = {Kaufman, Bruria},
  journal = {Physical Review},
  volume = {76},
  number = {8},
  pages = {1232},
  year = {1949},
  publisher = {APS},
  doi = {10.1103/PhysRev.76.1232},
}

@article{Kempe2006Complexity,
  author = {Kempe, Julia and Kitaev, Alexei and Regev, Oded},
  title = {{The Complexity of the Local Hamiltonian Problem}},
  journal = {SIAM Journal on Computing},
  volume = {35},
  number = {5},
  eprint = {quant-ph/0406180},
  archiveprefix = {arXiv},
  pages = {1070-1097},
  year = {2006},
  doi = {10.1137/S0097539704445226},
}

@article{Kirby2020,
  title = {Classical simulation of noncontextual {P}auli {H}amiltonians},
  volume = {102},
  DOI = {10.1103/physreva.102.032418},
  number = {3},
  pages = {032418},
  journal = {Physical Review A},
  publisher = {American Physical Society (APS)},
  author = {Kirby, William M. and Love, Peter J.},
  year = {2020},
  primaryClass = {quant-ph},
  archiveprefix = {arXiv},
  eprint = {2002.05693},
}

@article{kitaev2006anyons,
  title = {{Anyons in an exactly solved model and beyond}},
  author = {Kitaev, Alexei},
  journal = {Annals of Physics},
  volume = {321},
  number = {1},
  pages = {2--111},
  year = {2006},
  publisher = {Elsevier},
  doi = {10.1016/j.aop.2005.10.005},
  archivePrefix = {arXiv},
  eprint = {cond-mat/0506438},
}

@misc{kokcu2024classification,
  title = {{Classification of dynamical Lie algebras generated by spin
           interactions on undirected graphs}},
  author = {K{\"o}kc{\"u}, Efekan and Wiersema, Roeland and Kemper, Alexander F
            and Bakalov, Bojko N},
  eprint = {2409.19797},
  archiveprefix = {arXiv},
  year = {2024},
}

@article{lehot_line_graph,
  author = {Lehot, Philippe G. H.},
  title = {{An Optimal Algorithm to Detect a Line Graph and Output Its Root
           Graph}},
  year = {1974},
  issue_date = {Oct. 1974},
  publisher = {Association for Computing Machinery},
  address = {New York, NY, USA},
  volume = {21},
  number = {4},
  issn = {0004-5411},
  url = {https://doi.org/10.1145/321850.321853},
  doi = {10.1145/321850.321853},
  journal = {J. ACM},
  month = oct,
  pages = {569–575},
  numpages = {7},
}

@article{mandal2009exactly,
  title = {Exactly solvable {Kitaev} model in three dimensions},
  author = {Mandal, Saptarshi and Surendran, Naveen},
  journal = {Physical Review B},
  volume = {79},
  number = {2},
  pages = {024426},
  year = {2009},
  publisher = {APS},
  doi = {10.1103/PhysRevB.79.024426},
  archivePrefix = {arXiv},
  eprint = {0801.0229},
}

@article{mann_free_parafermions,
  title = {{A graph-theoretic framework for free-parafermion solvability}},
  volume = {481},
  ISSN = {1471-2946},
  url = {http://dx.doi.org/10.1098/rspa.2024.0671},
  DOI = {10.1098/rspa.2024.0671},
  number = {2312},
  journal = {Proceedings of the Royal Society A: Mathematical, Physical and
             Engineering Sciences},
  publisher = {The Royal Society},
  author = {Mann, Ryan L. and Elman, Samuel J. and Wood, David R. and Chapman,
            Adrian},
  year = {2025},
  pages = 20240671,
}

@article{Nguyen2024Quantum,
  author = {Nguyen, Nam and Watts, Thomas W. and Link, Benjamin and Williams,
            Kristen S. and Sanders, Yuval R. and Elman, Samuel J. and Kieferova,
            Maria and Bremner, Michael J. and Morrell, Kaitlyn J. and Elenewski,
            Justin and Isaacs, Eric B. and Johnson, Samuel D. and Mathieson, Luke
            and Obenland, Kevin M. and Otten, Matthew and Sundareswara, Rashmi
            and Holmes, Adam},
  title = {Quantum computing for corrosion simulation: workflow and resource
           analysis},
  journal = {npj Quantum Information},
  year = {2026},
  month = {Jan},
  volume = {12},
  number = {1},
  pages = {27},
  issn = {2056-6387},
  doi = {10.1038/s41534-025-01171-1},
  url = {https://doi.org/10.1038/s41534-025-01171-1},
  eprint = {2406.18759},
  archivePrefix = {arXiv},
  primaryClass = {quant-ph},
}

@article{nussinov2012arbitrary,
  title = {Arbitrary dimensional {Majorana} dualities and architectures for
           topological matter},
  author = {Nussinov, Zohar and Ortiz, Gerardo and Cobanera, Emilio},
  journal = {Physical Review B},
  volume = {86},
  number = {8},
  pages = {085415},
  year = {2012},
  publisher = {APS},
  doi = {10.1103/physrevb.86.085415},
  archivePrefix = {arXiv},
  eprint = {1203.2983},
}

@article{onsager1944crystal,
  title = {{Crystal statistics. I. A two-dimensional model with an
           order-disorder transition}},
  author = {Onsager, Lars},
  journal = {Physical Review},
  volume = {65},
  number = {3-4},
  pages = {117},
  year = {1944},
  publisher = {APS},
  doi = {10.1103/PhysRev.65.117},
}

@article{Patel2024Extension,
  title = {{Extension of Exactly-Solvable Hamiltonians Using Symmetries of Lie
           Algebras}},
  volume = {128},
  ISSN = {1520-5215},
  doi = {10.1021/acs.jpca.4c00993},
  number = {20},
  journal = {The Journal of Physical Chemistry A},
  publisher = {American Chemical Society (ACS)},
  author = {Patel, Smik and Yen, Tzu-Ching and Izmaylov, Artur F.},
  year = {2024},
  pages = {4150–4159},
}

@article{Patel2024Exactly,
  title = {{Exactly solvable Hamiltonian fragments obtained from a direct sum of
           Lie algebras}},
  volume = {160},
  ISSN = {1089-7690},
  doi = {10.1063/5.0207195},
  number = {19},
  journal = {The Journal of Chemical Physics},
  publisher = {AIP Publishing},
  author = {Patel, Smik and Izmaylov, Artur F.},
  year = {2024},
  pages = {194107},
}

@misc{perkins_structure_of_dense_claw_free_graphs,
  title = {{The typical structure of dense claw-free graphs}},
  author = {Will Perkins and Sam van der Poel},
  year = {2025},
  eprint = {2501.17816},
  archivePrefix = {arXiv},
  primaryClass = {math.CO},
}

@article{Pozsgay2024,
  title = {Quantum circuits with free fermions in disguise},
  volume = {58},
  ISSN = {1751-8121},
  url = {http://dx.doi.org/10.1088/1751-8121/adcd18},
  DOI = {10.1088/1751-8121/adcd18},
  number = {17},
  journal = {Journal of Physics A: Mathematical and Theoretical},
  publisher = {IOP Publishing},
  author = {Fukai, Kohei and Pozsgay, Balázs},
  year = {2025},
  month = apr,
  pages = {175202},
}

@misc{fukai_free_fermions_with_claws,
  title = {A free fermions in disguise model with claws},
  author = {Kohei Fukai and István Vona and Balázs Pozsgay},
  year = {2025},
  eprint = {2508.05789},
  archivePrefix = {arXiv},
  primaryClass = {cond-mat.stat-mech},
}

@article{room_synthesis_of_clifford_matrices,
  title = {{A Synthesis of the {Clifford} Matrices and Its Generalization}},
  author = {Thomas Gerald Room},
  journal = {American Journal of Mathematics},
  year = {1952},
  volume = {74},
  pages = {967},
  url = {https://api.semanticscholar.org/CorpusID:124153636},
}

@article{roussopoulos_line_graph,
  title = {{A max {m,n} algorithm for determining the graph H from its line
           graph G}},
  journal = {Information Processing Letters},
  volume = {2},
  number = {4},
  pages = {108-112},
  year = {1973},
  issn = {0020-0190},
  doi = {10.1016/0020-0190(73)90029-X},
  url = {https://www.sciencedirect.com/science/article/pii/002001907390029X},
  author = {Nicholas D. Roussopoulos},
}

@article{schuch_complexity_commuting+ham,
  author = {Schuch, Norbert},
  title = {{Complexity of commuting Hamiltonians on a square lattice of qubits}},
  year = {2011},
  publisher = {{Rinton Press, Incorporated}},
  address = {{Paramus, NJ}},
  volume = {11},
  number = {11–12},
  issn = {1533-7146},
  journal = {{Quantum Info. Comput.}},
  pages = {901–912},
  url = {https://www.rintonpress.com/journals/doi/QIC11.11-12-1.html},
  eprint = {1105.2843},
  archivePrefix = {arXiv},
  primaryClass = {quant-ph},
}

@article{schultz1964two,
  title = {{Two-dimensional Ising model as a soluble problem of many fermions}},
  author = {Schultz, Theodore D and Mattis, Daniel C and Lieb, Elliott H},
  journal = {Reviews of Modern Physics},
  volume = {36},
  number = {3},
  pages = {856},
  year = {1964},
  publisher = {APS},
  doi = {10.1103/RevModPhys.36.856},
}

@mastersthesis{spinner_evaluation_of_modular_decomposition_algorithms,
  doi = {10.5445/IR/1000170363},
  title = {{A practical evaluation of modular decomposition algorithms}},
  author = {Spinner, Jonas},
  year = {2024},
  school = {Karlsruher Institut für Technologie (KIT)},
}

@article{strassen_algorithm,
  author = {Strassen, Volker},
  title = {{Gaussian elimination is not optimal}},
  journal = {Numerische Mathematik},
  year = {1969},
  month = {Aug},
  day = {01},
  volume = {13},
  number = {4},
  pages = {354-356},
  doi = {10.1007/BF02165411},
}

@article{tantivasadakarn2020jordan,
  title = {{Jordan-Wigner} dualities for translation-invariant {Hamiltonians} in
           any dimension: {Emergent} fermions in fracton topological order},
  author = {Tantivasadakarn, Nathanan},
  journal = {Physical Review Research},
  volume = {2},
  number = {2},
  pages = {023353},
  year = {2020},
  publisher = {APS},
  doi = {10.1103/physrevresearch.2.023353},
  archivePrefix = {arXiv},
  eprint = {2002.11345},
}

@article{verstraete2005mapping,
  title = {Mapping local {Hamiltonians} of fermions to local {Hamiltonians} of
           spins},
  author = {Verstraete, Frank and Cirac, J Ignacio},
  journal = {Journal of Statistical Mechanics: Theory and Experiment},
  volume = {2005},
  number = {09},
  pages = {P09012},
  year = {2005},
  publisher = {IOP Publishing},
  doi = {10.1088/1742-5468/2005/09/P09012},
  archivePrefix = {arXiv},
  eprint = {cond-mat/0508353},
}

@misc{Watts24,
  title = {{Fullerene-encapsulated cyclic ozone for the next generation of
           nano-sized propellants via quantum computation}},
  author = {Thomas W Watts and Matthew Otten and Jason T Necaise and Nam Nguyen
            and Benjamin Link and Kristen S Williams and Yuval R Sanders and
            Samuel J Elman and Maria Kieferova and Michael J Bremner and Kaitlyn
            J Morrell and Justin E Elenewski and Samuel D Johnson and Luke
            Mathieson and Kevin M Obenland and Rashmi Sundareswara and Adam
            Holmes},
  year = {2024},
  eprint = {2408.13244},
  archiveprefix = {arXiv},
  url = {https://arxiv.org/abs/2408.13244},
}

@misc{weisstein_octahedral_graph,
  author = {Eric W. Weisstein},
  title = {{Octahedral Graph}},
  year = {n.d.},
  howpublished = {From MathWorld--A Wolfram Web Resource},
  url = {https://mathworld.wolfram.com/OctahedralGraph.html},
  note = {Accessed: March 18, 2025},
}

@article{Wiersema2024,
  title = {{Classification of dynamical Lie algebras of 2-local spin systems on
           linear, circular and fully connected topologies}},
  volume = {10},
  DOI = {10.1038/s41534-024-00900-2},
  number = {1},
  pages = {110},
  journal = {npj Quantum Information},
  publisher = {Springer Science and Business Media LLC},
  author = {Wiersema, Roeland and K\"{o}kc\"{u}, Efekan and Kemper, Alexander F.
            and Bakalov, Bojko N.},
  year = {2024},
  month = nov,
}

@article{wang1991ground,
  title = {{Ground state of the two-dimensional antiferromagnetic {Heisenberg}
           model studied using an extended {Wigner-Jordon} transformation}},
  author = {Wang, YR},
  journal = {Physical Review B},
  volume = {43},
  number = {4},
  pages = {3786},
  year = {1991},
  publisher = {APS},
  doi = {10.1103/physrevb.43.3786},
}

@misc{complementarycode,
  author = {Jannis Ruh},
  title = {Complementary code to the numerical experiments},
  note = {avalaible at \href{https://github.com/taeruh/free_fermions}{
          github.com/taeruh/free\_fermions}},
  year = {2025},
}

@misc{complementarydata,
  author = {Jannis Ruh},
  title = {Complementary data to the numerical experiments},
  note = {avalaible at commit \href{
          https://github.com/taeruh/free_fermions/tree/91b7044326ec56ebf608dfa0a66c7e7930bdd6df
          }{91b7044} (in the results directory},
  year = {2025},
}

@misc{modular_decomposition_library,
  author = {Jonas Spinner},
  title = {{Modular Decomposition}},
  note = {avalaible at \href{
          https://github.com/jonasspinner/modular-decomposition}{
          github.com/jonasspinner/modular-decomposition}},
  year = {2025},
}

\clearpage
\onecolumngrid
\appendix

\section{Modular Decomposition}\label{s"modular_decomposition}
We shall now introduce the modular decomposition tree and discuss the relevant observations for the algorithms we use in our numerical simulations. More additional details and proofs can be found in \cref{s"app_graphs}. The modular decomposition tree is a unique description of a graph that allows us to apply algorithms to recursively detect and remove twins as well as detect whether a graph is claw-free and if so, if it contains simplicial cliques \cite{chudnovsky_growing_without_cloning}. This is exactly what we need to apply~\cref{.result_informal} and detect \ac{scf} Hamiltonians in practice. Importantly, the complexity to create the modular decomoposition tree is linear in $|V| + |E|$~\cite{habib_survey_modular_decomposition}.

We begin by extending our graph-theoretical definitions. Abusing notation, we define $\vertm$ and $\edgem$ to be the mappings from a graph to its vertex and edge set, respectively, that is, $\vertm(G) = V$ and $\edgem(G) = E$. Given a set ${X \subseteq V}$, its complement is ${X^c = V {\setminus} X}$ and the graph complement is $G^c = \w(V, \w(V \times V) {\setminus} \w(E \cup \set{(x, x)}{x \in V}))$. We write ${H \,(<)\!\leq G}$ when $H$ is an induced subgraph of $G$, i.e., ${H = G[X]}$ for some ${X \,(\subset)\!\subseteq V}$. For $X \subseteq V$, the semi-open neighbourhood of $X$ is defined as $\neigh\braket{X} = \bigcup_{x \in X} \neigh(x)$ and similarly, the complementary version. The open and closed neighbourhoods are given by $\neigh(X) = \neigh\braket{X}{\setminus}X$ and $\neigh\bk{X} = \neigh\braket{X} \cup X$, respetively, and analogously for the complementary versions. For some $Y \subseteq V$, the open neighbourhood of $X$ in $Y$ is denoted as $\neigh_Y(X) = \neigh(X) \cap Y$; analogously for the  closed, semi-open, and complementary versions. We also allow graph subscripts for the neighbourhood to specify the graph in which the neighbourhood is taken, e.g., we have $\neigh = \neigh_G$ and $\neigh^c = \neigh_{G^c}$. $G[X]$ is a (complementary) component if, and only if, $\neigh^{(c)}[X] = X$ (which is equivalent to $\neigh^{(c)}(X) = \emptyset$).\\[0.5em] Let us repeat the definition of modules:
\begin{definition}[Modules]\label{.app_modules}
    Let $G = (V, E)$ be a graph. A module $X \subseteq V$ is defined through the following equivalent definitions:
    \begin{itemize}
        \item For all $y \in X^c$ it holds
        \begin{equation}
            y \in \neigh(X) \iff \forall x \in X: y \in \neigh(x)\;.
        \end{equation}
        \item For all $x, y \in X$ it holds
        \begin{equation}
            \neigh(x) {\setminus} X = \neigh(y) {\setminus} X\;.
        \end{equation}
    \end{itemize}
\end{definition}
The idea of the modular decomposition is to describe the graph as a quotient graph with respect to modules. To do this, we need a class of graphs that complement the definition of modules, namely prime graphs:
\begin{definition}[Prime graph]
    A graph $G = (V, E)$ is called \emph{prime} if, and only if, $\abs{V} \geq 4$ and it only has trivial modules; that is, the only modules are the empty set, single vertices and $V$ itself.
\end{definition}
We require $\abs{V} \geq 4$ since for $\abs{V} \leq 2$, $G$ is trivially prime, and for $\abs{V} = 3$, $G$ is never prime. Of specific interest are maximal modules and maximal prime subgraphs:
\begin{definition}[maximal, strong]
    Let $G = (V, E)$ be a graph with $\abs{V} > 1$.
    \begin{enumerate}
        \item A module $M \subseteq V$ is called \emph{maximal} if, and only if, $M \neq V$ and there is no module $M'$ such that $M \subset M' \subset V$.
        \item A module $M \subseteq V$ is called \emph{strong}, if it does not overlap with any other module, i.e., for all modules $M'$ it holds either $M' \subseteq M$, $M \subseteq M'$ or $M \cap M' = \emptyset$.
        \item An induced prime subgraph $H \leq G$ is called maximal if, and only if, there is no induced prime subgraph $H'$ such that $H < H' \leq G$.
    \end{enumerate}
\end{definition}
The maximal modules and prime graphs allow us to define the quotient graph with respect to a modular partition, which will lead to the main theorem of modular decompositions.
\begin{definition}[Modular partitions]
    Let $G = (V, E)$ be a graph. A partition $P$ of $V$ is called a (maximal) modular partition of $G$ if $X$ is a (maximal) module for all $X \in P$.
\end{definition}
Given a modular partition we can define the quotient graph:
\begin{definition}[Quotient graph]
    Let $G = (V, E)$ be a graph and $P$ a modular partition of $V$. The quotient graph $G/P$ is the graph $G/P = (P, E/P)$ where
    \begin{subalign}
        E/P &= \set{(X, Y) \in P^2}{Y \subseteq \neigh(X)}\\
        &= \set{(X, Y) \in P^2}{X \subseteq \neigh(Y)}\;.
    \end{subalign}
    $P = V/P$ is the set of equivalence classes of $V$ with respect to the canonical equivalence class induced by $P$. We write $\tilde{x}$ for the elements of $V/P$, for some representative $x \in V$. Let $\bc{x_1, \ldots x_{\abs{P}}}$ be a set of representatives; the induced subgraph $G[\bc{x_1, \ldots x_{\abs{P}}}]$ is isomorphic to the quotient graph $G/P$ and we call it a representative of $G/P$. Because of that, we include $G/P$ in the list of subgraphs of $G$ meaning all of the representatives of $G/P$.
\end{definition}
The key observation by Gallai is that if neither $G$ nor $G^c$ are disconnect, then the maximal modules of $G$ are strong and with that they form a partition to decompose the graph:
\begin{theorem}[Modular decomposition, Edmonds-Gallai; e.g.,~\cite{habib_survey_modular_decomposition}]\label{.modular_decomposition}
    Let $G = (V, E)$ be a graph. Then one and only one of the following holds:
    \begin{description}
        \item[Single] $G$ is a single vertex.
        \item[Parallel] $G$ is disconnected, i.e., there are more than one components.
        \item[Serial] $G^c$ is disconnected, i.e., there are more than one complementary components.
        \item[Prime] $G$ and $G^c$ are connected, $\abs{V} \geq 4$, the maximal modules are strong, i.e., they form a maximal modular decomposition $P$, and it holds that $G/P$ is maximal prime in $G$.
    \end{description}
\end{theorem}
The decomposition described in \cref{.modular_decomposition} is unique. If we are in the parallel case, we describe the graph as a quotient graph that is an independent set together with the information about each module, i.e., vertex in the quotient graph. Analogously, in the serial or prime case, the quotient graph is a clique or a prime graph, respectively. The idea of the modular decomposition tree is to apply the decomposition recursively, that is, apply \cref{.modular_decomposition} to each module:
\begin{definition}[Modular decomposition tree]\label{.decomposition_tree}
    Let $G = (V, E)$ be a graph. We define the modular decomposition tree $\tree(G)$ of $G$ recursively accordingly to the four cases in \cref{.modular_decomposition}:
    \begin{description}
        \item[Single] $\tree(G)$ consists only of the root node which contains the vertex label.
        \item[Parallel] The root node of $\tree(G)$ is labelled ``(p)arallel'' and its children are the decomposition trees of all (i.e., minimal) components.
        \item[Serial] The root node of $\tree(G)$ is labelled ``(s)erial'' and its children are the decomposition trees of all (i.e, minimal) complementary components.
        \item[Prime] The root node of $\tree(G)$ is labelled ``prime'' (we sometimes use an empty label for that) and its children are the decomposition trees of all modules in the maximal modular partition $P$ of $G$. Furthermore, the root node contains a description of $G/P$ (we sometimes draw $G/P$ between the children).
    \end{description}
\end{definition}
The modular decomposition tree fully, and uniquely, describes a graph $G$; \cref{f"mod_decom_example} shows an example graph with its decomposition tree. An important special case is a cograph, whose modular decomposition tree is a cotree:

\begin{definition}[Cograph and cotree]\label{.cotree}
    A graph $G$ is called cograph if $\tree(G)$ is a cotree, that is, if, and only if, the modular decomposition tree of $G$ does not contain any prime nodes.
\end{definition}
While the decomposition tree is not necessarily the most efficient description when performing graph transformations, the information contained in the tree is advantageous for detecting structures in graphs. Astonishingly, there exist different algorithms to compute the modular decomposition tree in linear time with respect to the number of edges and vertices in the graph; Ref.~\cite{habib_survey_modular_decomposition} provides an introduction to some of these.
\begin{figure}[t]
    \centering
    \includegraphics{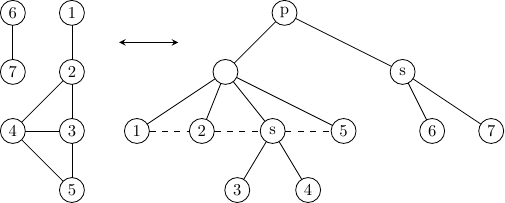}
    \caption{Example of a graph $G$ (left-hand side) and its modular decomposition tree $\tree(G)$ (right-hand side). ``p'' stands for a parallel node, ``s'' for a serial node, and an empty label for a prime node. For prime nodes we draw the edges of the quotient graph of the according module between the children (the quotient graph is trivial for parallel and serial nodes).} \label{f"mod_decom_example}
\end{figure}

The first observation we use to construct our detection algorithms is that twins are easily detected in the modular decomposition tree:
\begin{proposition}\label{.siblings_in_tree}
    Let $G = (V, E)$ be a graph. Each set of (false) true siblings of $G$ is given by collecting the leaves, i.e, ``single'' nodes, of a (parallel) serial node in the decomposition tree $\tree(G)$.
\end{proposition}
The proof is in \cref{.siblings_in_tree_with_proof}. When applying \cref{.result_informal} on a Hamiltonian $H$ with frustration graph $G$, we apply \cref{.siblings_in_tree} to find and collapse all twins recursively: We start at the root node. On each node that is not a single node, we recursively collapse all twins in all child-modules that are not leaves. Then, on the current node, if it is a parallel or serial node, we remove all but one of the leaves. If this leaf is the only remaining child of the current node, we replace the current node with the leaf.

The complexity of this is roughly $\order(\abs{V}^2)$, while a naive approach would roughly be $\order(\abs{V}^3)$ (in both cases one factor $\abs{V}$ accounts for the actual removal of the vertex in $G$).

This collapse sequence can be easily extended to also collapse line graphs. In \cref{.remove_line_graphs} we argue that we only have to check whether prime modules that have only leaves are line graphs. This can be done in linear time with respect to the size of the module \cite{roussopoulos_line_graph,lehot_line_graph}.

The next observation is about how we can detect whether a graph is claw-free or not, based on the decomposition tree. It turns out that claw-free-ness heavily restricts the structure of the decomposition tree:
\begin{proposition}\label{.claw_free_tree}
    Let $G = (V, E)$ be a claw-free graph. Then one of the following holds for the decomposition tree $\tree(G)$:
    \begin{enumerate}
        \item The root node is prime and all its children are cliques, i.e., leaves or serial nodes.
        \item The root node is serial and all its children are leaves, parallel nodes with two children, or prime nodes, and in the latter two cases all of their children are cliques, i.e., leaves or serial nodes.
        \item The root node is parallel and all its children are leaves or types of the above two cases.
    \end{enumerate}
\end{proposition}
While the tree structure in \cref{.claw_free_tree} is necessary for graphs to be claw-free, it is not sufficient. However, we can characterise claws given such a decomposition tree structure:
\begin{proposition}\label{.simple_tree_claws}
    Let $G = (V, E)$ be a graph, not necessarily claw-free, such that $\tree(G)$ has the form described in \cref{.claw_free_tree}. Each claws is covered in exactly one of the following two cases:
    \begin{enumerate}
        \item Let $A$ be a prime node with partition $P$. For every claw $\bc{\tilde{x}_0, \ldots, \tilde{x}_3} \subseteq G[\vertm(A)]/P$, the set $\bc{x_0, \ldots, x_3}$ is a claw, for arbitrary representatives.
        \item Let $A$ be a prime node, with partition $P$, that is a child of a serial node $B$. For every independent set $\bc{\tilde{x}_1, \tilde{x}_2, \tilde{x}_3} \subseteq G[\vertm(A)]/P$, the set $\bc{x_0, \ldots x_3}$, arbitrary representatives, is a claw for every vertex $x_0$ in another child of $B$.
    \end{enumerate}
\end{proposition}
The proofs of these two propositions are in
\cref{.claw_free_tree_with_proof,.simple_tree_claws_with_proof}, respectively. Note that if all twins have been recursively removed from a graph, the cliques in the \cref{.claw_free_tree} collapse into single vertices. The naive approach of detecting claws in a graph is of complexity $\order(\abs{V}^4)$, however, by translating this problem into a search for triangles and using efficient matrix algorithms, e.g., a variant of the Strassen algorithm \cite{strassen_algorithm}, this can be reduced to $\order(\abs{V}^{3,8})$. While we cannot strictly improve this complexity (e.g., in the case where the root node is prime with only leaves as children), searching for cliques based on \cref{.claw_free_tree,.siblings_in_tree} does help in practice (note that we already have the tree due to the twin collapse). Firstly, we check whether the tree $\tree(G)$ has the form as described in \cref{.claw_free_tree}; if not, we can early stop and conclude that $G$ is not claw-free. If the tree has the form, we search for the claws described in \cref{.simple_tree_claws}, which essentially requires to search for triangles in smaller subgraphs.

The last ingredient we need is to detect simplicial cliques. A naive approach would be of exponential complexity, however, for claw-free graphs, Ref.~\cite{chudnovsky_growing_without_cloning} provides an efficient algorithm to find simplicial cliques in $\order(\abs{V}^4)$ time. The algorithm requires checking whether induced (quotient) subgraphs are prime, which can be accomplished via the modular decomposition tree. For more details on the algorithm, we refer the reader to Ref.~\cite{chudnovsky_growing_without_cloning}.

\section{More on the Modular Decomposition and Simplicial Claw-Free Graphs}\label{s"app_graphs}

\noindent In this section we state and prove more technical details required to find \ac{scf} graphs.

\begin{proposition}\label{.module_properties}
    Some well known and basic properties of modules: Let $G = (V, E)$ be a graph and $X, Y \subseteq V$ be two modules.
    \begin{enumerate}
        \item If $X \cap Y \neq \emptyset$ then $X \cup Y$ is a module.
        \item $X$ is also a module of $G^c$.
        \item $X$ and $Y$ are either all-to-all connected (complete-adjacent) or all-to-all disconnected (complete-anti-adjacent).
        \item Let $W \subseteq V$, then $X \cap W$ is a module of $G[W]$.
    \end{enumerate}
\end{proposition}
\begin{proposition}
    Let $G = (V, E)$ be a prime graph. Then $G^c$ is also prime and both are connected.
\end{proposition}
\begin{proof}
    Clear (cf. \cref{.module_properties}).
\end{proof}
\begin{lemma}\label{.induced_two-edge_path}
    Let $G = (V, E)$ be a prime graph. Then every vertex, is an endpoint of an induced two-edge path.
\end{lemma}
\begin{proof}
    Let $x \in V$. Since $G$ is prime, it holds $V {\setminus} \neigh[x] \neq \emptyset$ (otherwise $\neigh(x)$ would be a module), and again, because $G$ is prime, which implies that $G$ is connected, there exists a $z \in V {\setminus} \neigh[x]$ and $y \in \neigh(x)$ such that $z$ neighbours $y$. But then $\bc{x, y, z}$ is an induced two-edge path.
\end{proof}
\begin{proposition}\label{.no_subsequent_components}
    Let $G = (V, E)$ be a graph. In the modular decomposition tree, a serial node is never child of a serial node and a parallel node is never child of a parallel node.
\end{proposition}
\begin{proof}
    Clear, because we decompose into the minimal (complementary) components.
\end{proof}
\begin{proposition}\label{.siblings_in_tree_with_proof}
    Let $G = (V, E)$ be a graph. Each set of (false) true siblings of $G$ is given by collecting the leaves, i.e, ``single'' nodes, of a (parallel) serial node in the decomposition tree $\tree(G)$.
\end{proposition}
\begin{proof}
    Firstly, let $S$ be the set of all leaf nodes of a (parallel) serial parent node $a_n$, where \en is the layer in $\tree(G)$ (root node is in layer $1$), and let $M_n \subseteq V$ be the module corresponding to $a_n$ ($M_1 = V$). It is clear that $S$ is a set of (false) true siblings in $G[M_n]$. Now let $a_{n-1}$ be the parent node of $a_n$ and $M_{n-1} \subseteq V$ be the corresponding module. In general, it is clear that (false) true siblings in $G[M_n]$ are also (false) true siblings in $G[M_{n-1}]$ (because $M_n$ is a module). Therefore, it follows inductive that $S$ is a set of (false) true siblings in $G$.
    
    Now let $S$ be a set of (false) true siblings in $G$. We prove the statement via induction with respect to the size of the subgraphs that contain $S$. The base case is the graph $G[S]$: Since $S$ is (an independent set) a clique in $G[S]$, the decomposition tree of $G[S]$ has two layers, where the root node is a (parallel) serial node and the second layer contains all vertices of $S$ as leaves. Now let $H \leq G$ with $S \subseteq \vertm(H)$ and let the statement be true for all graphs $H'$ with $\abs{\vertm(H')} < \abs{\vertm(H)}$. Assume that we are not in the trivial case where the root node $a_1$ of $\tree(H)$ is a serial or parallel node and all vertices in $S$ are leaves of $a_1$. We state that $S$ is fully contained in one of modules of the children of $a_1$. Assuming the contrary, there exist two vertices $x, y \in S$, $x \neq y$, such that $x$ is in one module $M_x$ and $y$ is in another module $M_y$ (corresponding to two different child nodes $a_x$ and $a_y$ of $a_1$). We show that this leads to a contradiction, considering three cases:
    
    Firstly, consider the case where $a_1$ is a parallel node (this is already a contradiction if $S$ is a set of true siblings). without loss of generality, let $\abs{M_x} > 1$ (otherwise, if $\abs{M_x} = \abs{M_y} = 1$, we are back in the trivial case). Then there is a $z \in \neigh(x) \subseteq M_x$ because otherwise $x$ would be in a leaf node; but then $y$ cannot neighbour $z$, so it cannot be a sibling of $x$; contradiction.
    
    Secondly, consider the case where $a_1$ is a serial node (this is already a contradiction if $S$ is a set of false siblings). Again, without loss of generality, we have $\abs{M_x} > 1$. Then there is a $z \in M_x {\setminus} \neigh(x)$ because otherwise $x$ would be a leaf node; but then $y$ neighbours $z$, so it cannot be a sibling of $x$; contradiction.
    
    Thirdly, consider the case where $a_1$ is a prime node. Since $\bc{x, y}$ is a (false) true twin in $H$, $\bc{\tilde{x}, \tilde{y}}$ (remember that $\tilde{x} = M_x$ and $\tilde{y} = M_y$) is a (false) true twin in $H/P$, where $P$ is the maximal modular partition of $H$: Let $\tilde{z} \in \neigh_{H/P}(\tilde{x}) {\setminus} \bc{\tilde{x}, \tilde{y}}$, if existent. Then we also have $z \in \neigh_H(x)$ and therefore also $z \in \neigh_H(y)$. This implies that $\tilde{z} \in \neigh_{H/P}(\tilde{y}) {\setminus} \bc{\tilde{x}, \tilde{y}}$. Vice versa, we repeat the argument for $\tilde{z} \in \neigh_{H/P}(\tilde{y}) {\setminus} \bc{\tilde{x}, \tilde{y}}$, if existent, and it follows that $\bc{\tilde{x}, \tilde{y}}$ is a non-trivial module in the prime graph $H/P$, more specifically a (false) true twin; contradiction.
    
    Therefore, $S$ is fully contained in one of the children modules of $a_1$, let this module be $M_2 \subset \vertm(H)$. $M_2$ is strictly smaller than $V$; thus the induction hypothesis applies on $G[M_2]$, and there is a (parallel) serial node $a$ in $\tree(G[M_2])$ such that all vertices of $S$ are leaves of $a$. However, $a$ is obviously also a node in $\tree(G)$.
\end{proof}

\begin{proposition}[\cite{chudnovsky_growing_without_cloning}]
\label{.max_claw_free_cliques}
    Let $G = (V, E)$ be a claw-free graph where and $G$ and $G^c$ are connected and $\abs{V} > 1$. Then, the maximal modules of $G$ are cliques.
\end{proposition}
\begin{proof}
    Follows easily with \cref{.induced_two-edge_path,.modular_decomposition}.
\end{proof}
\begin{proposition}[\cite{chudnovsky_growing_without_cloning}]
\label{.serial_to_parallel_cliques}
    Let $G = (V, E)$ be a claw-free graph. If, in $\tree(G)$, a parallel node is a child of a serial node, it consists of exactly two components and both are cliques.
\end{proposition}
\begin{proof}
    Assuming the contrary, there is a parallel node $A$, child of a serial node $B$, that contains three independent vertices $x_1, x_2, x_3 \in A$. However, then every vertex $x_0$ in one of the other children of $B$ (at least one exists) is the central vertex in the claw $\bc{x_0, \ldots, x_3}$.
\end{proof}
\begin{proposition}\label{.claw_free_tree_with_proof}
    Let $G = (V, E)$ be a claw-free graph. Then one of the following holds for the decomposition tree $\tree(G)$:
    \begin{enumerate}
        \item The root node is prime and all its children are cliques, i.e., leaves or serial nodes.
        \item The root node is serial and all its children are leaves, parallel nodes with two children, or prime nodes, and in the latter two cases all of their children are cliques, i.e., leaves or serial nodes.
        \item The root node is parallel and all its children are leaves or types of the above two cases.
    \end{enumerate}
\end{proposition}
\begin{proof}
    This is just a combination of \cref{.no_subsequent_components,.max_claw_free_cliques,.serial_to_parallel_cliques}.
\end{proof}
\begin{corollary}
    The decomposition tree of claw-free graphs has at most 5 layers. 
\end{corollary}
\begin{proposition}\label{.simple_tree_claws_with_proof}
    Let $G = (V, E)$ be a graph, not necessarily claw-free, such that $\tree(G)$ has the form described in \cref{.claw_free_tree}. All claws are covered in exactly one of the following two cases:
    \begin{enumerate}
        \item Let $A$ be a prime node with partition $P$. For every claw $\bc{\tilde{x}_0, \ldots, \tilde{x}_3} \subseteq G[\vertm(A)]/P$, the set $\bc{x_0, \ldots, x_3}$ is a claw, for arbitrary representatives.
        \item Let $A$ be a prime node, with partition $P$, that is a child of a serial node $B$. For every independent set $\bc{\tilde{x}_1, \tilde{x}_2, \tilde{x}_3} \subseteq G[\vertm(A)]/P$, the set $\bc{x_0, \ldots x_3}$, arbitrary representatives, is a claw for every vertex $x_0$ in another child of $B$.
    \end{enumerate}
\end{proposition}
\begin{proof}
    It is clear that the above cases describe claws. Now let $K \in G$ be a claw. If \cref{.claw_free_tree_with_proof}~(a) holds, it is clear that we are in case (a) here, since each module that is a clique can contain only one vertex of a claw (this is true in general). Analogously, it follows that if \cref{.claw_free_tree_with_proof}~(b) holds, we are either in case (a) or (b) here, since only the prime nodes can contain independent sets of size $3$.
\end{proof}
\begin{remark}
    The decomposition tree in \cref{.claw_free_tree_with_proof}, becomes even simpler after all twins have been collapsed recursively as the cliques collapse into a single vertex. Therefore in case \textit{(a)}, the graph is a prime graph, and in case \textit{(b)}, the root node is serial with maximally one leaf and all other children are prime graphs.
\end{remark}
\begin{lemma}
    Let $G = (V, E)$ be a graph. A module $A$ in $G$ is a line graph if, and only if, all modules in $\tree(A)$ are line graphs.
\end{lemma}
\begin{proof}
    This is clear by the characterization of line graphs via forbidden subgraphs.
\end{proof}
\begin{remark}\label{.remove_line_graphs}
    When collapsing twins and line graph recursively, we only have to check for line graphs on modules that are prime and only have leaves. Assume, that a module $A$ has a child that is not a leaf. Then this child must have a module $B$ in its decomposition tree that is prime and has only leaves (after the collapse). This module $B$ cannot be a line graph, otherwise it would have been collapsed. Therefore $A$ cannot be a line graph.
\end{remark}
\begin{lemma}[\cite{chudnovsky_claw_free_independence_polynomial}]
\label{.neighbourhood_is_simplicial}
    Let $G = (V, E)$ be a claw-free graph with a simplicial clique $K \subseteq V$. Then $\neigh(k){\setminus} K$ is either a simplicial clique in $G\bk{V{\setminus}K}$ or the empty set.
\end{lemma}
\begin{proof}
    See~\cite{chudnovsky_claw_free_independence_polynomial} 2.4.
\end{proof}
\begin{lemma}\label{.contains_induced_path}
    Let $G = (V, E)$ be a graph. Every path contains an induced path from start to end.
\end{lemma}
\begin{proof}
    Let $p = \bc{x_1, \ldots, x_n} \subseteq V$ be a path, for \en, with start $x_1$ and end $x_n$. Let \en[m] be maximal such that $(x_1, x_m) \in E$. Remove all $x_i$ from $p$ with $1 < i \leq m$. This gives us a possibly new $p' = \bc{x_1, x_m, x_{m+1}, \ldots x_n}$ where $x_1$ only neighbours the next vertex ($x_m$) in the path. Repeat this procedure with $\bc{x_m, \ldots, x_n}$ and so on.
\end{proof}
\begin{proposition}[\cite{chudnovsky_claw_free_independence_polynomial}]
\label{.preserve_simpliciality}
    Let $G = (V, E)$ be a simplicial claw-free connected graph. Then $G[U]$ is a simplicial claw-free graph for all $\emptyset \neq U \subseteq V$.
\end{proposition}
\begin{proof}
    Let $K \subseteq V$ be a simplicial clique and $U \subseteq V$. If $K \cap U \neq \emptyset$ then $K \cap U$ is clearly a simplicial clique in $G[U]$. Assume $K \cap U = \emptyset$. Let $p = \bc{x_1, \ldots x_n} \subseteq V$, \en, be a path from some $x_n \in U$ to some $x_1 \in K$ such that $x_{\geq 2} \notin K$, and without loss of generality, $p$ is an induced path (\cref{.contains_induced_path}). Define $K_m = \neigh(x_m){\setminus}\w(K_0 \cup \ldots \cup K_{m-1})$ for $m = 1, \ldots, n-1$ and $K_0 = K$. Inductively it follows that $K_m$ is a simplicial clique in $G\bk{V{\setminus}\w(K_0 \cup \ldots \cup K_{m-1})}$, for all $m = 1, \ldots, n-1$: We have $x_2 \in K_1$, therefore, $K_1 \neq \emptyset$ and with \cref{.neighbourhood_is_simplicial} it follows that $K_1$ is a simplicial clique in $G[V{\setminus}K_0]$. Now let the statement be true for $m-1$, $m \geq 2$. Since $p$ is an induced path, it holds $x_{m+1} \in \neigh\w(x_{m}){\setminus}\w(K \cup \neigh\w(x_{1}) \cup \ldots \cup \neigh\w(x_{m-1})) \subseteq K_m$; therefore, with \cref{.neighbourhood_is_simplicial}, $K_m$ is a simplicial clique in $G\bk{\w(V{\setminus}\w(K_0 \cup \ldots \cup K_{m-2})) {\setminus} K_{m-1}}$.
    
    Now choose the smallest $s \in \bc{1, \ldots, n-1}$ such that $K_s \cap U \neq \emptyset$ (this exists since $x_n \in K_{n-1} \cap U$). $K_s$ is a simplicial clique in $G\bk{V{\setminus}\w(K_0 \cup \ldots \cup K_{s-1})}$, and therefore it is clearly also a simplicial clique in $G[U] = G\bk{V{\setminus}U^c}$ since $K_0 \cup \ldots \cup K_{s-1} \subseteq U$.
\end{proof}
\begin{proposition}
    Let $G = (V, E)$ be a connected graph. By removing an arbitrary number of siblings and line-graph modules, possibly recursively, $G$ can become simplicial claw-free, but never lose this property.
\end{proposition}
\begin{proof}
    It is clear that removing vertices does not create claws, and it does not remove simplicial cliques in a simplicial claw-free graph according to \cref{.preserve_simpliciality}. \Cref{f"create_simplicial_clique} shows that we can create simplicial cliques by removing siblings.
\end{proof}

\section{Details on the Block-Diagonalisation}\label{s"app_block_diag}

\noindent In this section we give a constructive proof of \Cref{.result_informal}. To do so, we shall first standardise our notation.
\begin{notation}
    Let $\w(a_{\kappa_i})_{1 \leq i \leq n}$, \en, be a sequence in a semigroup for some arbitrary, but strictly totally ordered, indices $\w(\kappa_i)_{1 \leq i \leq n}$. For $i, j \in \bc{1, \ldots, n}$ with $i \leq j$, we define
    \begin{align}
        a_{\kappa_i \lps \kappa_j} &\coloneqq \prod_{\kappa \in \kappa_i \lps \kappa_j}
        a_\kappa \coloneqq a_{\kappa_i} \cdots a_{\kappa_j}\;,\\
        a_{\kappa_j \gps \kappa_i} &\coloneqq \prod_{\kappa \in \kappa_j \gps \kappa_i}
        a_\kappa \coloneqq a_{\kappa_j} \cdots a_{\kappa_i}\;,\\
        a_{\kappa_i \lps} &\coloneqq \prod_{\kappa \in \kappa_i \lps} a_\kappa \coloneqq
        \prod_{\kappa \in \kappa_i \lps \kappa_n} a_\kappa\;,\\
        a_{\lps \kappa_i} &\coloneqq \prod_{\kappa \in \lps \kappa_i} a_\kappa \coloneqq
        \prod_{\kappa \in \kappa_1 \lps \kappa_i} a_\kappa\;,\\
        a_{\kappa_i \gps} &\coloneqq \prod_{\kappa \in \kappa_i \gps} a_\kappa \coloneqq
        \prod_{\kappa \in \kappa_i \gps \kappa_1} a_\kappa\;,\\
        a_{\gps \kappa_i} &\coloneqq \prod_{\kappa \in \gps \kappa_i} a_\kappa \coloneqq
        \prod_{\kappa \in \kappa_n \gps \kappa_i} a_\kappa\;.
    \end{align}
    Furthermore, we define the shorthand $\bc{x_{\leq k}} = x_1, x_2, \ldots, x_k$ for \en[k], and some $x_i \in \bc{0, 1}^{m_i}$, $m_i \in \N$.\\
    The above shorthands are analogously defined for up-indices.
\end{notation}
The idea of the proof is to inductively construct the projectors $P$ in \Cref{.result_informal}, such that they correspond to alternating false twin projections and true twin rotations. We shall then extend the statement by including additional unitary rotations that allow more complex collapse sequence including collapsing line graph modules.

Let us first redefine the twin collapse sequence (\cref{.main_collapse_graph}) in terms of the modular decomposition tree.
\begin{definition}[Collapse sequence]\label{.collapse_graph}
    Let $G = (V, E)$. Set $G^0 = G$. We define the following sequence of graphs: $\w(G^i)_{0 \leq i \leq c}$, \en[c]:
    \begin{itemize}
        \item For odd \en[i]: For all maximal sets $T$ of false siblings in $G^{i-1}$, fix one of the siblings and remove the other vertices. That is, for each parallel node in $\tree\w(G^{i-1})$ remove all leaves but one, and if this leaf is the only child, then replace the parallel node by that leaf.
        \item For even \en[i], $i \geq 2$: For all maximal sets $T$ of true siblings, fix one of the siblings and remove the other vertices. That is, for each serial node in $\tree\w(G^{i-1})$ remove all leaves but one, and if this leaf is the only child, then replace the serial node by that leaf.
        \item Set \en[c] even and minimal, such that $G^c = G^{c+1}$.
    \end{itemize}
\end{definition}
\begin{definition}[Symmetry and rotation sequences]\label{.base_sequences}
    Let $G = (V, E)$ be the frustration graph of a Hamiltonian $H$, and $\w(G^i)_{0 \leq i \leq c}$, \en[c], be the according graph sequence obtained by the twin collapse. Set $r = c/2 - 1$. The sequence $\w(L^i)_{0 \leq i \leq r}$ of pseudo-symmetries, is defined by $L^i$ being the group generated by false twin symmetries of $G^{2i}$ as defined in \cref{.false_twin_symmetries}, for $i = 0, \ldots r$. The rotation-exponent sequences $\w(\w(\rho^i_j)_{1 \leq j \leq q_i})_{0 \leq i \leq r}$, $q_i \in \N$, are defined as follows: For $i = 0, \ldots r$, let $\w(T_k)_{1 \leq k \leq l}$ be the (ordered) sequence of maximal true sibling sets in $G^{2i + 1}$, \en[l], and let $\w(h_k)_{1 \leq k \leq l} \subseteq \vertm\w(G^{2i+1})$ such that $h_j \in T_j$, for all $j = 1, \ldots, l$. For all $k = 1, \ldots, l$ and for each $h \in T_k{\setminus}\bc{h_k}$ set $\rho = h_k h$ and append it to the $\w(\rho^i_j)_j$ sequence. It holds $q_i = \sum_{k = 1}^l \w(\abs{T_k} -
    1)$.
\end{definition}
In the following, we implicitly refer to the above definitions.
\begin{definition}[False-twin projectors]\label{.projectors}
    Let the group $L^i$ of false-twin symmetries of $G^i$ be independently generated by $\bc{g^i_j}_{1\le j\le m_i}$, $m_i \in \N$. The false-twin projectors are defined as
    \begin{equation}
        P^i(x) = \prod_{j = 1}^{m_i} \frac{1 + x g^i_j}{2}\;,
    \end{equation}
    for $x \in \bc{-1, +1}^{m_i}$ and  all $i \in \bc{0, \ldots, r}$.
\end{definition}
\begin{proposition}
    Given $i \in \bc{0, \ldots r}$, the false-twin projectors $P^i(x)$ are orthogonal and complete with respect to $x$, i.e.,
    \begin{equation}
        \sum_{x \in \bc{-1, +1}^{m_i}} P^i(x) = \1\;.
    \end{equation}
    Furthermore, they commute with each other.
\end{proposition}
\begin{proof}
    Let $i \in \bc{0, \ldots, r}$. $L^i$ is a group of hermitian, unitary and commuting elements (\cref{.false_twin_symmetries}); therefore the operators are clearly commuting orthogonal projectors. Furthermore, we have
    \begin{equation}
        \sum_{x \in \bc{-1, +1}^{m_i}} P^i(x) = \prod_{j=1}^{m_i} \sum_{x_j \in \bc{-1, +1}} \frac{1 + x g^i_j}{2} = \prod_{j=1}^{m_i} \1 = \1\;.
    \end{equation}
\end{proof}
\begin{remark}
    Remember that for a complete set of orthogonal commuting projectors, products of different projectors give zero, i.e., let $\bc{P_1, \ldots, P_n}$, \en be such a set, then it is $P_i P_j = 0$ for $i \neq j$. This is because we can diagonalize all projectors in the same basis, and all have eigenvalues of $1$; therefore overlapping eigenvectors would contradict $\sum_i P_i = \1$.
\end{remark}
\begin{corollary}
    Let $P^i(x)$ denote the false-twin projectors of $G^i$ in the sequence. Then 
    \begin{equation}
        \bk{P^i(x), g} = 0
    \end{equation}
    for all $g \in \vertm\w(G^j)$ where $i \in \bc{0, \ldots, r}$, $j \in \bc{2i, \ldots, c}$ and $x \in \bc{-1, +1}^{m_i}$.
\end{corollary}
\begin{proof}
Clear, per construction: Let $i, j$ as above. Each $h \in L^i$ commutes with all $g \in
\vertm\w(G^{2i}) \supseteq \vertm\w(G^j)$.
\end{proof}

\begin{corollary}\label{.projectors_commute}
The false-twin projectors of
    \cref{.projectors} commute with each other and with all subsequent rotation-exponents, i.e.,
    \begin{equation}
        \bk{P^i(x), P^j(y)} = \bk{P^i(x), \rho^k_l} = 0\;,
    \end{equation}
    for all $i, j, k \in \bc{0, \ldots, r}$, $k \geq i$, $x \in \bc{-1, +1}^{m_i}$, $y \in \bc{-1, +1}^{m_j}$ and $l \in \bc{1, \ldots q_k}$.
\end{corollary}
\begin{proof}
    Let $i, j, k, l$ as above, and without loss of generality, $j \geq i$. $P^i(x)$ commutes with all $g \in \groupset{h}{h \in \vertm\w(G^{2j})} \geq L^j$, as well as with all $g \in \groupset{h}{h \in \vertm\w(G^{2k + 1})} \ni \rho^k_l$.
\end{proof}
\begin{definition}[Rotated projectors]\label{.rotated_projectors}
   For $j = 0, \ldots, r$, the \emph{rotated projectors} are defined as
    \begin{equation}
        P_R^j\w(\bc{x^{\leq j}}) = \w(U^{j-1 \gps}\w(\bc{x^{\leq j-1}}))\ur * P^{j}\w(x^j)\;,
    \end{equation}
    with $U^j\w(\bc{x^{\leq j}}) = \prod_{k = 1}^{q_j} U^j_k\w(\bc{x^{\leq j}})$ where
    \begin{equation}
        U^j_k\w(\bc{x^{\leq j}}) = \exp\w(i \theta^j_k\w(\bc{x^{\leq j}}) \rho^j_k / 2)
    \end{equation}
    for some (for now arbitrary) angles $\theta^j_k\w(\bc{x^{\leq j}}) \in \R$ (shall be specified later) and $x^j \in \bc{-1, +1}^{m_j}$, for all $k = 1, \ldots, q_j$.
\end{definition}
\begin{proposition}\label{.rotated_projectors_commute}
    The rotated projectors are orthogonal projections, and they commute, i.e.,
    \begin{equation}
        \bk{P_R^i\w(\bc{x^{\leq i}}), P_R^j\w(\bc{y^{\leq j}})} = 0
    \end{equation}
    for all $i, j \in \bc{0, \ldots, r}$, $x^i \in \bc{-1, +1}^{m_i}$, $y^j \in \bc{-1, +1}^{m_j}$, and they are complete with respect to $x^j$, i.e.,
    \begin{equation}
        \sum_{x^j \in \bc{-1, +1}^{m_j}} P_R^j\w(\bc{x^{\leq j}}) = \1\;.
    \end{equation}
\end{proposition}
\begin{proof}
    Let $i, j \in \bc{0, \ldots, r}$ and without loss of generality $i \leq j$. For conciseness we shall drop the parameters. Since conjugation is a homomorphism it is clear that the operators are complete orthogonal projectors. Furthermore, we have
    \begin{subalign}
        P_R^i P_R^j &= \w(U^{i-1 \gps})\ur P^{i} \w(U^{j-1 \gps i})\ur P^{j} U^{j-1 \gps}\\
        &\!\overset{(*)}{=} \w(U^{j-1 \gps})\ur P^{j} U^{j-1 \gps i} P^{i} U^{i-1 \gps}\\
        &= \w(U^{j-1 \gps})\ur  P^{j} U^{j-1 \gps} \w(U^{i-1 \gps})\ur P^{i} U^{i-1
        \gps}\\
        &= P_R^j P_R^i\;,
    \end{subalign}
    where in $(*)$ we used \cref{.projectors_commute} (commute $P^{i} \w(U^{j-1 \gps i})\ur$, then $P^{i}P^{j}$, then $P^{i} U^{j-1 \gps i}$).
\end{proof}

\begin{proposition}\label{.result_full}
    Set $H^0 = H$ and $U^{-1} = \1$. For $i = 0, \ldots, r$, recursively define the Hamiltonian sequence $\w(H^i)_i$ as
    \begin{equation}
        H^i\w(\bc{x^{\leq i}}) = P_R^i\w(\bc{x^{\leq i}}) H^{i-1}\w(\bc{x^{\leq i-1}})
        P_R^i\w(\bc{x^{\leq i}})\;,
    \end{equation}
    and the angles $\theta^i_j\w(\bc{x^{\leq i}})$ in $U^i\w(\bc{x^{\leq i}})$ accordingly to \cref{.merge_paulis} such that true twins are merged (which depends on the subsequently updated weights in $H^i\w(\bc{x^{\leq i}})$ after each rotation $\theta^i_{j-1}\w(\bc{x^{\leq i}})$) for $x^j \in \bc{-1, +1}^{m_j}$, $j = 0, \ldots, q_i$. For all $i = 0, \ldots, r$ we have
    \begin{subalign}
        H &= \sum_{x^0, \ldots, x^i} H^i\w(\bc{x^{\leq i}})\\
        &= \sum_{x^0, \ldots, x^i} \w(U^{i \gps}\w(\bc{x^{\leq i}}))\ur * H_{CP}^i\w(\bc{x^{\leq i}})\\
        &= \sum_{x^0, \ldots, x^i} P_R^{i \gps}\w(\bc{x^{\leq i}}) \w(\w(U^{i \gps}\w(\bc{x^{\leq i}}))\ur * H_C^i\w(\bc{x^{\leq i}})) P_R^{i \lps}\w(\bc{x^{\leq i}})
        \label{e"rotated_collapsed_inner}
    \end{subalign}
    where the last two equations are true per summand; the sums go over $x^j \in \bc{-1, +1}^{m_j}$ for $j = 0, \ldots, i$ and we set
    \begin{subalign}
        H_{CP}^i\w(\bc{x^{\leq i}}) &= U^{i \gps}\w(\bc{x^{\leq i}}) * H^i\w(\bc{x^{\leq i}})\\
        &= \w(\prod_{j \in i \gps} U^j\w(\bc{x^{\leq j}}) P^j\w(x^j)) H
        \w(\prod_{j \in \lps i} P^j\w(x^j) U^{j \dagger}\w(\bc{x^{\leq j}}))\;.\label{e"collapse_sequence}
    \end{subalign}
    Furthermore, for $x^j, y^j \in \bc{-1, +1}^{m_j}$, $j = 0, \ldots, i$, it holds
    \begin{align}
        H^i_{CP} \w(\bc{x^{\leq i}}) &= P^{i \gps}\w(\bc{x^{\leq i}}) H^i_C \w(\bc{x^{\leq i}}) P^{\lps i}\w(\bc{x^{\leq i}})\;,\label{e"full_projector_symmetry}\\
        H^i\w(\bc{x^{\leq i}}) &= P_R^{i \gps}\w(\bc{x^{\leq i}}) H^i \w(\bc{x^{\leq i}}) P_R^{\lps i}\w(\bc{x^{\leq i}})\;,\label{e"full_rotated_projector_symmetry}
    \end{align}
    where $H_{CP}$ consists of the $P^{\lps i}$-symmetry projected operators of the collapsed graph, that is
    \begin{align}
        H_C^i\w(\bc{x^{\leq i}}) &= \sum_{g \in V'} w'_g\w(\bc{x^{\lps i}}) g \qquad\quad (w'_g \in \R)\;,\\
        H_{CP}^i\w(\bc{x^{\leq i}}) &= H^i_C\w(\bc{x^{\leq i}}) P^{\lps i}\w(\bc{x^{\leq i}})\;, \label{e"collapsed_hamiltonian_projected}
    \end{align}
    with $V' = \vertm\w(G^{2i}) \subseteq V$ where $\bk{g, P^{\lps i}\w(\bc{x^{\leq i}})} = 0$ for $g \in V'$, and we have $\frust{H^i_{CP}\w(\bc{x^{\leq i}})} = \frust{H^i_C\w(\bc{x^{\leq i}})} \coloneqq \frust{V'} = G^{2i}$. Moreover, it holds $\bk{g P^{\lps i-1}\w(\bc{x^{\leq i-1}}), P^i\w(\bc{x^{\leq i}})} = 0$ for $g \in \vertm\w(G^{2(i-1)})$, and for $j \leq i$
    \begin{align}
        \bk{H^{i-1}_{C}\w(\bc{x^{\leq i-1}}), P^j\w(y^j)} &= 0\;, \label{e"projector_symmetry_construction}\\
        \bk{H^{i-1}_{CP}\w(\bc{x^{\leq i-1}}), P^j\w(y^j)} &= 0\;, \label{e"projector_symmetry}\\
        \bk{H^{i-1}\w(\bc{x^{\leq i-1}}), P_R^i\w(\bc{x^{\leq i}})} &=0\;,
        \label{e"rotated_projector_symmetry}\\
        \bk{\w(U^{i \gps}\w(\bc{x^{\leq i}}))\ur * H_C^i\w(\bc{x^{\leq i}}), P_R^j\w(\bc{x^{\leq i}})} &=0\;.
        \label{e"rotated_projector_on_rotated_ham}
    \end{align}
\end{proposition}
\begin{proof}
    Generally, let $x^i, y^i \in \bc{-1, +1}^{m_i}$, $i \in \bc{0, \ldots, r}$, and for conciseness, we shall occasionally drop these arguments, if they are not important. Let $i \in \bc{0, \ldots, r}$. We shall prove the statement via induction. For $i=0$ the statement is trivial. Now let it be true for $i-1$.\\
    Firstly, we show \cref{e"collapse_sequence}, but this is clear via induction:
    \begin{equation}
        U^{i\gps} * H^i = U^{i\gps} * \w(P_R^i H^{i-1} P_R^i) = U^i P^i \w(U^{i-1\gps} * H^{i-1}) P^i \w(U^i)\ur\;.
    \end{equation}
    Now let $j \in \bc{0, \ldots, i}$. Per construction, it is clear that $P^j$ is a symmetry projector of $H_C^{i-1}$. More specifically, we know that $H_{CP}^{i-1}$ is a Hamiltonian with $P^{\lps i-1}$-symmetry-projected operators, such that the frustration graph is $G^{2(i-1)}$. $P^i$ is chosen such that it is a symmetry projector of the non-projected Hamiltonian, $H_C^i$ with frustration graph $G^{2(i-1)}$. But then, since $P^j$, $P^{\lps i-1}$ and $g \in \vertm\w(G^{2(i-1)})$ commute pairwise, we have $\bk{g P^{\lps i-1}, P^j} = 0$ and get \cref{e"projector_symmetry_construction,e"projector_symmetry}. Analogously it is clear that $\bk{g, P^{\lps i}} = 0$ for $g \in \vertm\w(G^{2i})$, and per construction (\cref{e"collapse_sequence}) it is clear that $H^i_{CP}$ has the form described in \cref{e"collapsed_hamiltonian_projected} with $\frust{H^i_C \w(\bc{x^{\leq i}})} = G^{2i}$. This then also proves \cref{e"full_projector_symmetry} as $P^{i \gps}$ is a projector.
    
    Next, we show that $P_R^i\w(\bc{x^{\leq i}})$ is a symmetry projector of $H^{i-1}\w(\bc{x^{\leq i-1}})$:
    \begin{subalign}
        H&^{i-1} P_R^i = P_R^{i-1}\w(\w(U^{i-2 \gps})\ur * H_{CP}^{i-2})P_R^{i-1} P_R^i\\
        &= \w(U^{i-2 \gps})\ur P^{i-1} H_{CP}^{i-2} P^{i-1} \w(U^{i-1})\ur P^{i} \w(U^{i-1 \gps})\\
        &= \w(U^{i-1 \gps})\ur H_{CP}^{i-1} P^{i} \w(U^{i-1 \gps})\\
        &= \w(U^{i-1 \gps})\ur P^i H_{CP}^{i-1} \w(U^{i-1 \gps})\\
        &= \w(U^{i-1 \gps})\ur P^{i} (U^{i-1}) P^{i-1} H_{CP}^{i-2} P^{i-1}  \w(U^{i-2 \gps})\\
        &= P_R^i P_R^{i-1}\w(\w(U^{i-2 \gps})\ur * H_{CP}^{i-2})P_R^{i-1}\\
        &= P_R^i H^{i-1}\;.
    \end{subalign}
    This proves \cref{e"rotated_projector_symmetry}. \Cref{e"full_rotated_projector_symmetry} is simply true because $P_R^{\lps i}$ is a projector. We then get
    \begin{subalign}
        H &= \sum_{x^0, \ldots, x^{i-1}} H^{i-1}\w(\bc{x^{\leq i-1}})\\
        &= \sum_{x^0, \ldots, x^i} H^{i-1}\w(\bc{x^{\leq i-1}}) P_R^i(\bc{x^{\leq i}})\\
        &= \sum_{x^0, \ldots, x^i} P_R^i(\bc{x^{\leq i}}) H^{i-1}\w(\bc{x^{\leq i-1}})
        P_R^i(\bc{x^{\leq i}})\\
        &= \sum_{x^0, \ldots, x^{i}} H^{i}\w(\bc{x^{\leq i}})\;.
    \end{subalign}
    It remains to show 
    Eqs.~\eqref{e"rotated_collapsed_inner}~and~\eqref{e"rotated_projector_on_rotated_ham}
    the first follows inductively (repeat the argument on $P^{i-1 \gps}$) via
    \begin{subalign}
        \w(U^{i\gps})\ur * H_{CP}^i &= \w(U^{i\gps})\ur * \w(P^{i\gps} H_C^i P^{\lps i})\\
        &\!\overset{(*)}{=} \w(U^{i-1\gps})\ur * \w(P^{i\gps} \w(\w(U^i)\ur * H_C^i) P^{\lps i})\\
        &= \w(U^{i-1\gps})\ur * \w(P^i U^{i-1\gps} \w(U^{i-1\gps})\ur P^{i-1\gps} \w(\w(U^i)\ur * H_C^i) P^{\lps i-1}U^{i-1\gps} \w(U^{i-1\gps})\ur P^i)\\
        &= P^i_R \w(\w(U^{i-1\gps})\ur * \w(P^{i-1\gps} \w(\w(U^i)\ur * H_C^i) P^{\lps i-1})) P^i_R\;,
    \end{subalign}
    where in $(*)$ we commuted $\w(U^i)\ur$ through $P^{i\gps}$, and for \cref{e"rotated_projector_on_rotated_ham} we have
    \begin{subalign}
        \w(\w(U^{i \gps})\ur * H_C^i) P_R^j &= \w(U^{j-1 \gps})\ur \w(\w(U^{i \gps j-1})\ur * H_C^i) P^j U^{j-1 \gps}\\
        &= \w(U^{j-1 \gps})\ur P^j \w(\w(U^{i \gps j-1})\ur * H_C^i) U^{j-1 \gps}\\
        &= P_R^j \w(\w(U^{i \gps})\ur * H_C^i)\;,
    \end{subalign}
    where we again commute $P^j$ through $U^{i \gps j-1}$ (and $H^i_C$) in the second step.
\end{proof}
\begin{corollary}\label{.cograph_collapse}
    In \cref{.result_full}, if the frustration graph $G = (V, E)$ of $H$ is a cograph, that is, $\tree\w(G)$ is a cotree with depth $2(r+1)$ or $2(r+1) - 1$, \en[r] (cf. \cref{.cotree}), then we have 
\begin{equation}
  H_C^{r}\w(\bc{x^{\lps r}}) = w\w(\bc{x^{\lps r}}) g P^{\lps r}_R\w(\bc{x^{\lps r}})\;,
\end{equation}
for appropriate $g \in V$ and $w \in \R$.
\end{corollary}
\begin{proof}
Clear, cf. \cref{.collapse_graph}.
\end{proof}
\begin{remark}\label{.extension}
    The proofs of \cref{.rotated_projectors_commute,.result_full} are built around two main observations: Firstly the operator sequence of false twin projections, $P^j$, and true twin rotations, $U^j$, act locally in the frustration graph preserving the decomposition tree otherwise. This is what allows us to define the sequences in the first place without having them interfering with each other. Secondly, the rotations, $U^j$ commute with all previous projections $P^k$, $k \leq j$, which implies that the rotated projectors commute with each other.

    However, the specific form of the $U^j$ rotations is not important; we could have chosen other unitaries as long as they have the desired action on the graph and the according commutation rules. Moreover, we can actually do more than just collapsing true twins. For example, consider the case where a module as the three edge path with vertices $\bc{a, b, c, d}$, connected in that order. Then the rotation $U = \ee^{\theta cd}$, $\theta \in \R$, fulfils the required commutation rules and changes the path to a cograph (with appropriately chosen $\theta$) but leaves the rest of the graph invariant. This cograph can then be collapsed with the usual twin collapses. In fact, we can reduce any module that is an odd-length edge path to an even-length edge path as we show in \cref{f"path_collapse}.

    More generally, one can extend the sequences in \cref{.collapse_graph,.base_sequences} to allow any unitaries that change a module locally as long as they commute with the previous projections and everything outside the module.
\end{remark}
\begin{figure}[!hbt]
    \centering
    \includegraphics{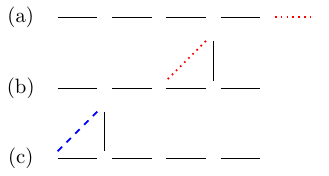}
    \caption{
        Elementary reduction of an odd-length path. Rotating around the red dotted edge in (a), with appropriate angles, turns the path into (b). Again, rotating around the red dotted edge in (b) turns the path into (c). The dashed blue edge in (c) is a true twin that can be collapsed as usual. This procedure works for any odd-length path.
    } \label{f"path_collapse}
\end{figure}
\begin{proposition}
    If the group $S$ that provides the basis operators for $H$ as in \cref{e"hamiltonian} is unitarily equivalent to the Pauli group\footnote{For example groups covered in \cref{.main_isomorphism}.}, we can extend \cref{.result_full} to include unitary rotations that cause any module that is a line graph to fully collapse into one vertex.
\end{proposition}
\begin{proof}
    \Cref{.extension} describes how one can include additional unitaries in general. Without loss of generality, we can assume that $S$ is the Pauli group. We need to show that given a line graph in a module, potentially after previous collapse operations, that there exist unitaries that commute with all previous false twin projections and act locally on the module such that the module can be collapsed into a single vertex.

    It is known that one can find the hermitian Pauli generators of unitaries such that the according unitaries transform a line graph into an independent set \cite{chapman_solvable_spin_models}. We argue that these generators can always be chosen such that they commute with the previous projectors and everything outside the module. Assume we have a generator $g$, with according unitary $U = \ee^{i \theta g}$ for some $\theta \in \R$, that anticommutes with an operator $h \in V$ outside the module. Extend the group $S$ by another spin via the tensorproduct; let the index of this spin be \en[i]. Now define $h' = h \otimes Z_i$, $g' = g \otimes X_i$ and $U' = \ee^{i \theta g'}$. Let $H'$ be the Hamiltonian where $h$ is replaced by $h'$. It is $H \cong PH'P$ where $P = (\1 + Z_i)/2$. $U'$ acts under conjugation on $H'$ as $U$ would act on $H$ with the difference that $U'$ now commutes with $h'$; furthermore $H$ and $H'$ have the same frustration graph.
    This can be done for all $h \in V$ that anticommute with $g$; we only have to project out the spin extensions in the end. The same argument holds if there are any false twin projectors that we need $g$ to commute with, since these projectors products of $(\1 + p)$ operators where $p$ is a Pauli string (extend $p$ by $\otimes Z_i$ if $p$ anticommutes with $g$).
\end{proof}

\section{More on the SCF Examples}\label{s"more_scf_examples}

\subsection{Uniform Random Paul Strings}\label{s"uniform_random_paul_strings}

In this section we discuss models that are constructed by uniformly drawing random Pauli strings for a fixed number $n$ of spins.

In \cref{f"two_local} we restrict to two-local Pauli strings and in \cref{f"sparse} the Pauli strings can be of arbitrary length. We plot the results against the probability $p$ that a Pauli string is accepted. As expected, for increasing $p$, all lines go to $0$. For lower $n$ the decrease is generally slower. However, note that this would swap if plotted against the absolute number $m$ of Pauli strings in the Hamiltonian; this is because for $n \to \infty$, two randomly drawn two-local Pauli strings commute almost surely. In fact, for $n \to \infty$, one can choose
\begin{equation}
  m \approx \w(\frac{3}{4})^{3/4} \w(\frac{1}{2})^{1/2} \varepsilon^{1/4} n^{3/4}
\end{equation}
and the Hamiltonian is almost surely \ac{scf} with probabilty $1 - \varepsilon$, $\varepsilon \in \R_+$ (cf. \cref{s"analytical_k_local}).

In \cref{f"sparse} we plot the results against the absolute number of Pauli strings in the
Hamiltonian. Notably, we see that for different numbers $n$ of spins, all lines are essentially the same. This is because for two random Pauli strings, the probability that they commute is $\frac{1}{2}$, independently of $n$. Therefore, if the number of drawn Pauli strings is small with respect to $2^{2n}$, i.e., the total number of Pauli strings, the effective frustration graph distributions we draw from for different $n$ are the same. More specifically, the effective frustration graphs we draw form a subset of the $\gnp\w(\text{\# Pauli strings}, 1/2)$ graphs, that have a vanishing \ac{scf} probability.
\begin{figure}[t]
    \centering
    \includegraphics{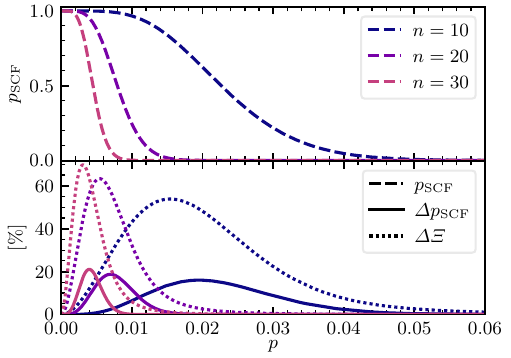}
    \caption{\ac{scf} probability for a random two-local Pauli Hamiltonian. The two-local Pauli strings are uniformly randomly drawn with probability $p$ for different numbers $n$ of spins. $p_{\mathrm{SCF}}$, $\Delta p_{\mathrm{SCF}}$ and $\Delta \Xi$ are as in \cref{f"numerical_results}. For higher densities, all lines go to $0$.} \label{f"two_local}
\end{figure}
\begin{figure}[b]
    \centering
    \includegraphics{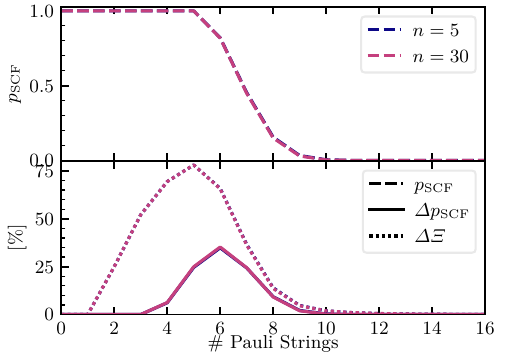}
    \caption{\ac{scf} probability for a random Pauli Hamiltonian. Pauli strings are uniformly randomly drawn with probability for different numbers $n$ of spins. We show the results against the absolute number of Pauli strings in the Hamiltonian. $p_{\mathrm{SCF}}$, $\Delta p_{\mathrm{SCF}}$ and $\Delta \Xi$ are as in \cref{f"numerical_results}. For higher numbers of Pauli strings, all lines go to $0$.} \label{f"sparse}
\end{figure}

\subsection{Analytical Bounds on the SCF Class}

In the following, let $\pcf$, $\psgcf$ and $\pscf$ be the probability that a graph is claw-free, a graph is simplicial given that it is claw-free, and is \acl{scf}, respectively.

\subsubsection{Erdös R\'enyi Model}\label{s"erdos_renyi_bounds}

Given the $\gnp(n, p)$ class with \en and $p \in \bk{0, 1}$ we are interested in upper and lower bounds on the probability that the graphs are \ac{scf}. Let $C$ be the number of claws in a graph $G \in \gnp(n, p)$, then we have
\begin{equation}
    \Exp[C] = \binom{n}{4} 4 p^3 (1-p)^3\;,
\end{equation}
where we count the number of claws in each set of four vertices. The first moment methods gives us
\begin{equation}
    \pcf = 1 - \Exp[C > 0] \geq 1 - \binom{n}{4} 4 p^3 (1-p)^3\;.
\end{equation}
For the second moment we have
\begin{align}
    \Exp[C^2] &= \binom{n}{4} \binom{4}{0} \binom{n-4}{4} \w(4 p^3 (1-p)^3)^2\nonumber\\
    &+ \binom{n}{4} \binom{4}{1} \binom{n-4}{3} \w(4 p^3 (1-p)^3)^2\nonumber\\
    &+ \binom{n}{4} \binom{4}{2} \binom{n-4}{2} \w(2 p^3(1-p)^2 2 p^2 (1-p)^3 + 2 p^3 (1-p)^3 2 p^3 (1-p)^2)\nonumber\\
    &+ \binom{n}{4} \binom{4}{3} \binom{n-4}{1} \w(3 p^3 (1-p)^3 p (1-p)^2 + p^3 (1-p)^3 p^3)\nonumber\\
    &+ \binom{n}{4} \binom{4}{4} \binom{n-4}{0} 4p^3 (1-p)^3
\end{align}
With the second moment method, this gives us
\begin{subalign}
    \pcf &= 1 - \Exp[C > 0] \leq 1 - \frac{\Exp[C]^2}{\Exp[C^2]}\\
    &= 1 - \frac{\binom{n}{4}}{\binom{n-4}{4} + 4 \binom{n-4}{3} + \frac{3}{2} \binom{n-4}{2} \frac{1}{p(1-p)} + \frac{1}{4} \binom{n-4}{1}\w(\frac{3}{p^2 (1-p)} + \frac{1}{(1-p)^3}) + \frac{1}{4 p^3 (1-p)^3}}\;.\label{e"claw_free_upper}
\end{subalign}
Now let $S$ be the number of vertices in $G$, which is not necessarily claw-free for now, that are simplicial cliques; it holds
\begin{equation}
  \Exp[S] = n \sum_{l = 0}^{n-1} \binom{n-1}{l} p^l (1 - p)^{n - 1 - l} p^{\binom{l}{2}} = n \sum_{l = 0}^{n-1} \binom{n-1}{l} p^{l + \binom{l}{2}} (1 - p)^{n -
  1 - l}\;,
\end{equation}
where, for each vertex we summed over the probabilities that it has a neighbourhood of size $l$ that is a clique. Assuming that $G$ is claw-free we know that there are less cases where a single vertex has a neighbourhood that is not a clique, since some of those cases are not allowed. Therefore, it holds
\begin{equation}
    \Exp[S | \text{$G$ is claw-free}] \geq \Exp[S].
\end{equation}
Let us temporarily assume that $G$ is connected. By upper bounding the second moment of $S$ (given $G$ is claw-free) by $n^2$ we can apply the second moment method and get
\begin{equation}
    \psgcf \geq \Exp[S]^2 / n^2\;.
\end{equation}
For small $p$, we can trivially lower bound $\Exp[S]$ by only taking the term for $l = 0$; this gives
\begin{equation}
    \psgcf \geq (1-p)^{2(n-1)} \geq (1-p)^{2n}\;.\label{e"simp_given_claw_free}
\end{equation}
This formula is additive in $n$, therefore, it also holds if $G$ is not connected, that is, if $G$ has \en[m] components of sizes $n_1, \ldots n_m$, then we have
\begin{equation}
    \psgcf \geq \prod_{i=1}^m (1-p)^{2n_i} = (1-p)^{2n}.
\end{equation}
This gives us the following lower bound on $\pscf$, for small $p$:
\begin{equation}
    \pscf = \pcf \psgcf \geq \w(1 - \binom{n}{4} 4 p^3 (1-p)^3) (1-p)^{2n}\;.
\end{equation}
For $n \to \infty$ and large $p$ we can make the statement that if the graph is claw-free, then it is almost surely also simplicial. This follows from Thm. 1.6 (i) in Ref. \cite{perkins_structure_of_dense_claw_free_graphs}, which states that in that case the graph is co-bipartite with high probability. For $n \to \infty$ and small $p$ we can make a similar statement, i.e., if the graph is almost surely claw-free, then it is also almost surely simplicial: Choose $p$ such that $\pcf \geq 1 - \epsilon$ for some small $\epsilon \in \R$ according to our bounds, i.e.,
\begin{equation}
    1 - \binom{n}{4} 4 p^3 (1-p)^3 \geq 1 - \epsilon
\end{equation}
As we consider $n \to \infty$ and small $p$ (i.e., $1-p \approx 1$), this gives us
\begin{equation}
    p \leq \sqrt[3]{\frac{6\epsilon}{n^4}}\;.
\end{equation}
Using \cref{e"simp_given_claw_free}, we then have (l'H\^opital)
\begin{equation}
    \psgcf \geq \w(1 - \sqrt[3]{\frac{6\epsilon}{n^4}})^{2n} \overset{n \to \infty}{\longrightarrow} 1\;.
\end{equation}
As upper bound we can choose \cref{e"claw_free_upper} as $\pscf \leq \pcf$, for any $p$.

\subsubsection{\texorpdfstring{$K$}{k}-Local Random Pauli Strings}\label{s"analytical_k_local}

Let $H$ be a Hamiltonian of $m$ $k$-local Pauli strings on $n$ spins, \en[m, k, n] with $k \ll n$. We can estimate a trivial, approximated, threshold for $m$ such that the frustration graph $G$ of $H$ is almost surely \ac{scf} in the limit $n \to \infty$. Let $X$ be the number of sets of $4$ vertices that are connected. If $X = 0$ we know that $G$ is claw-free and every component has a simplicial clique since the size of every component is strictly upper bounded by $4$. Let $x$ and $y$ be two random $k$-local Pauli strings where $x \neq y$; the probability $p$ that $x$ and $y$ anticommute is given by (the fraction after the first sum is the hypergeometric distribution formula)
\begin{equation}
    p = \sum_{s = 0}^k \frac{\binom{k}{s}\binom{n-k}{k-s}}{\binom{n}{k}} \sum_{a = 1,a+=2}^s \binom{s}{a} \w(\frac{2}{3})^a \w(\frac{1}{3})^{s-a}\label{e"p_k_local}
\end{equation}
We use this as approximation for any two distinct Pauli operators in the Hamiltonian under the assumption that $m \ll 2^{2n}$.
In the limit $n \to \infty$ only the leading order term w.r.t. $n-k$ in \cref{e"p_k_local} (which is in the $s = 1$ summand) and $1/n$, matters, which gives
\begin{equation}
    p \to \frac{k \frac{(n-k)^{k-1}}{(k-1)!}}{\frac{n^k}{k!}} \frac{2}{3} = \frac{2k^2 (n-k)^{k-1}}{3 n^k} \to \frac{2k^2}{3n}\;.
\end{equation}
For $n \to \infty$, $p$ goes to zero, therefore the most likely set of $4$ vertices that is fully connected is the three-edge path (other connected $4$-sets have more edges), therefore we can reduce $X$ to count only these case and find a threshold for them. Given $4$ vertices, there are $\binom{4}{2} 2 = 12$ possible three-edge paths (choose two endpoints) and we get
\begin{align}
    \Exp[X] \approx \binom{m}{4} 12 p^3 (1-p)^3\;.
\end{align}
Approximating $(1-p) \approx 1$ and requiring $\Exp[X] \leq \varepsilon$ for some small $\varepsilon \in \R_+$ (because then $\Exp[X = 0] = 1 - \Exp[X > 0] \geq 1 - \Exp[X] \geq 1 - \varepsilon$), we have
\begin{IEEEeqnarray}{LrCl}
    & p &\leq& \sqrt[3]{\frac{\varepsilon}{12 \binom{m}{4}}}\\
    \Leftrightarrow & \frac{2k^2}{3n} &\leq& \sqrt[3]{\frac{\varepsilon}{12 \binom{m}{4}}}\\
    \Leftrightarrow & \binom{m}{4} &\leq& \frac{27}{96} \varepsilon \w(\frac{n}{k^2})^3\\
    \overset{\approx}{\Leftrightarrow} & m &\leq& \frac{3^{3/4}}{\sqrt{2}}\varepsilon^{1/4} \w(\frac{n}{k^2})^{3/4}\;,
\end{IEEEeqnarray}
where approximated $\binom{m}{4}$ by its leading order in $m$, because if $n \to \infty$ as then we also have $m \to \infty$. As expected, the maximum $m$ such that the Hamiltonian is \ac{scf} (almost surely with $1 - \varepsilon$) increases with $n$ and decrease with $k$.

\subsubsection{Majorana Hamiltonian}\label{s"majorana_analytical}

For the Hamiltonian $H_{\mathrm{M}}$ as in \cref{e"majorana_hamiltonian}, one can perform
a similar analysis as in \cref{s"analytical_k_local}. Let $m$ be the number of Majorana
operators in $H_{\mathrm{M}}$ such that the Hamiltonian is almost surely simplicial and
claw-free. For simplicity we are only interested in the the order of $m$ with respect to
$n$, in the limit $n \to \infty$. It is apparent that an analogous analysis as in
\cref{s"analytical_k_local} leads to $m \leq \order[n^{3/4}]$. Equivalently, for the
interaction probability $p$ the bound $p \leq \order[n^{-13/4}]$ implies that the
Hamiltonian is almost surely \ac{scf}.

\subsection{Exact Results for Periodic Models}

If the unit cell in periodic lattice models is small enough it is possible to calculate exact \ac{scf} probabilities for these models by considering all possible unit cells:

\newcommand\hs{\mathscr{H}}

\subsubsection{Periodic Brick lattice}\label{s"brick_exact}

In \cref{f"brick_lattice}, the unit cell is effectively defined by the five edges between the nodes $0$ to $5$, i.e., by $e_i = (i, i+)$ for $i = 0, \ldots 4$. For each edge, there are $9$ Pauli operators that can appear with probability $p \in (0, 1]$; given a Hamiltonian $H$, let $E_i$ the set of its operators on $e_i$ for $i = 0, \ldots 4$. Let $\hs$ be the set of all possible Hamiltonians - it is $m \coloneqq \log_2\w(\abs{\hs}) = 45$ - and define
\begin{equation}
    \supp \hs \to \bc{0, \ldots m},\; H \mapsto \sum_{i=0}^4 \abs{E_i}\;.
\end{equation}
Let $\hs(k) = \set{H \in \hs}{\abs{\supp(H)} = k}$, $k \in \bc{0, \ldots m}$. A Hamiltonian $H \in \hs(k)$ apparently appears with probability $p^{k} (1-p)^{m-k}$. However, we only want to consider models that are two-dimensional; therefore this probability has to be normalised. Define $\mathscr{H}_2 = \set{H \in \hs}{\forall i \in \bc{0, \ldots 4}: E_i \neq \emptyset}$ and $\mathscr{H}_2(k) = \set{H \in \hs_2}{\abs{\supp(H)} = k}$, $k \in \bc{0, \ldots m}$. Given a probability $p \in \bk{0, 1}$, the normalisation factor is $\Pr(H \in \hs_2 | p)\ui$ for $H \in \hs$; let $H \in \hs$, then
\begin{subalign}
    \Pr(H \in \hs_2 | p) &= \sum_{k=0}^m \Pr(H \in \hs_2(k))\\
    &= \sum_{k=0}^m \abs{\hs_2(k)} p^k(1-p)^{m-k}\;.
\end{subalign}
For $k \in \bc{0, \ldots m}$, $\abs{\hs_2(k)}$ can be calculated with the inclusion-exclusion principle: Define $A_i = \set{H \in \hs(k)}{E_i = \emptyset}$, then we have
\begin{subalign}
    \abs{\hs_2(k)} &= \abs{\hs(k)} - \abs{\bigcup_{i = 0}^4 A_i}\\
    &= \abs{\hs(k)} - \w(\sum_{\emptyset \neq S \subseteq \bc{0, \ldots 4}} (-1)^{\abs{S} + 1} \abs{\bigcap_{i \in S} A_i})\\
    &= \binom{45}{k} - \w(\sum_{j = 1}^5 \binom{5}{j} (-1)^{j+1} \binom{(5-j) 9}{k})\;.
\end{subalign}
Now define $\hs_{2, \mathrm{SCF}}(k) = \set{H \in \hs_2(k)}{\text{$H$ is SCF after recursively collapsing all twins and line-graph modules}}$, $k \in \bc{0, \ldots, m}$. By enumerating all $H \in \hs_2$ (e.g., via a binary tree construction corresponding to the drawn Paulis and an exhaustive depth-first-search) one calculates $\abs{\hs_{2, \mathrm{SCF}}(k)}$; this then gives the final result
\begin{equation}
    p_{\mathrm{SCF}}(p) = \frac{1}{\Pr(H \in \hs_2 | p)} \sum_{k = 0}^m \abs{\hs_{2, \mathrm{SCF}}(k)} p^k(1-p)^{m-k}\;.
\end{equation}
Analogously one calculates $\Delta p_{\mathrm{SCF}}(p)$.

\subsubsection{Periodic Square Lattice}\label{s"square_exact}

The calculation works analogously to the brick lattice in \cref{s"brick_exact}: The unit cell is effectively defined by three edges; let $e_h$ be the horizontal edge, $e_v$ the vertical edge, and $e_l$ the edge to the local spin, with according sets $E_h$, $E_v$ and $E_l$, respectively. Let $m = 27$. We allow $E_l$ to be empty, however, $E_h$ ad $E_v$ must be non-empty. Therefore, we have
\begin{equation}
    \abs{\hs_2(k)} = \binom{27}{k} - \w(\sum_{j = 1}^2 \binom{2}{j} (-1)^{j+1} \binom{(3-j) 9}{k})\;.
\end{equation}
The rest is analogue as in \cref{s"brick_exact}.

\section{Details on the Stone-von Neumann Theorem}\label{s"app_isomorphisms}

This section contains generalised versions and proofs of the results in \cref{s"res_stone_von_neumann}.
\begin{definition}[Polar commutator group]
    Let \en[d, n], $W \in \M_n\w(\Z_d)$ and set $\Omega = W - W\ut$. Furthermore, let $A$ be a multiplicative Abelian group, and define an embedding $\zeta: \Z_d \to A$. Let $F \subseteq A$ be a fixed set of representatives of $A / \zeta\w(\Z_d)$, without loss of generality $1 \in F$, and set $r: A \to F$, $u: A \to \Z_d$ such that $a = r(a) \zeta\w(u(a))$ for all $a \in A$.\\
    We define the group $\gk_d^n(A, \zeta, W) = \w(F, \Z_d, \Z_d^n, \cdot)$ via
    \begin{equation}
        \begin{split}
            \cdot: \gk_d^n \times \gk_d^n  &\to \gk_d^n,\\
            (a, p, x), (b, q, y) &\mapsto (r(ab), u(ab) + p + q + x\ut W y, x + y)\;.
        \end{split}
    \end{equation}
\end{definition}
We often just write $\gk = \gk_d^n$, potentially with some of $A, \zeta, W$ instead of the full form, $\gk_d^n(A, \zeta, W)$, if the rest is clear from context and or not important.
\begin{proposition}
    $\gk_d^n(A, \zeta, W)$ is indeed a group and the functions $r$ and $u$ are well defined. The identity element of the group is $e \coloneqq (1, 0, 0)$ and given $(a, p, x) \in \gk_d^n(A, \zeta, W)$, its inverse is $\w(r\w(a\ui), u\w(a\ui) -p + x\ut W x, -x)$.
\end{proposition}
\begin{proof}
    Each $a \in A$ has a unique decomposition into $f \cdot z$ with $f \in F$ and $z \in \zeta(\Z_d)$ since cosets do not overlap and $F$ is a fixed set, thus $r$ and $y$ are well defined. Only the inverse is not clear, but it can be verified elementarily by using that for all $a \in A$ we have
    \begin{equation}
        r\w(a r\w(a\ui)) \zeta\w(u\w(a r\w(a\ui)) + u\w(a\ui)) \overset{(*)}{=} a r\w(a\ui)\frac{a\ui}{r\w(a\ui)} = 1\;,
    \end{equation}
    where we used the homomorphic property of $\zeta$ in $(*)$.
\end{proof}

\begin{remark}
    \begin{enumerate}
        \item The decomposition of $a \in A$ into $r(a)$ and $u(a)$ is analogous to the polar decomposition for complex numbers. While we are mostly interested in the third and second component of the group elements - which will describe operators and their commutation relations, respectively - the first component will allow us to consider scalar factors in front of the operators. Splitting these scalar factors into the polar decomposition with respect to the second component will allow us to define the homomorphisms to the operators injectively. Furthermore, the first component also allows us to consider roots of the $\zeta$ embedding, which will be required for some isomorphisms later (e.g., when we cannot divide by $2$ in $\Z_{d=2}$).
        \item For all $m \in \Z_d$, it holds $\zeta\w(m) = \omega^m$ for some fixed $\omega \in A$ of order $d$, since $\Z_d$ is cyclic.
    \end{enumerate}
\end{remark}

\begin{proposition}
    Let \en[d, n], $W \in \M_n\w(\Z_d)$, and $A$ an abelian group with embedding $\zeta: \Z_d \to A$. For the polar commutator group $\gk_d^n(A, \zeta, W)$, we have
    \begin{enumerate}
        \item $\Omega\ut = -\Omega$ (i.e., $\Omega$ is skew-symmetric),
        \item $\forall x \in \Z^n_d: x\ut \Omega x = 0$ (i.e., $\Omega$ is alternating)
        \item $\forall x, y \in \Z_n^d: \tbk{(\cdot, \cdot, x), (\cdot, \cdot, y)} = (1, x\ut\Omega y, 0)$,
        \item The center is given by $\cen[\gk_d^n] = \set{(a, p, x)}{a \in F, p \in \Z_d, x \in \ker \Omega}$
    \end{enumerate}
\end{proposition}
\begin{proof}
    It is clear that $(\cdot, \cdot, 0) \in \cen[\gk]$. Let $(a, p, x), (b, q, y) \in \gk$. It is
    \begin{subalign}
        \tbk{(a, p, x), (b, q, y)} &= (a, p, x)(b, p, x) (r(ba), u(ba) + q+p + y\ut W x, y +x)\ui\\
        &= (a, p, x)(b, q, y) ((1, y\ut W x - x\ut W y, 0)(a, p, x)(b, p, y))\ui\\
        &= (a, p, x)(b, q, y) ((a, p, x)(b, q, y))\ui(1, x\ut W y - x\ut W\ut y, 0)\\
        &= (1, x\ut \Omega y, 0)\;.
    \end{subalign}
    For $(\cdot, \cdot, x) \in \cen[\gk]$, we therefore have $x \in \ker \Omega\ut = \ker \Omega$.
\end{proof}
\begin{remark}
    If, and only if $\Omega$ is non-degenerate (assuming $d$ is prime), i.e., $\cen[\gk_d^n] = \bc{(\cdot, \cdot, 0)}$, then $\Omega$ describes a symplectic form (and therefore $n$ must be even).
\end{remark}
\begin{proposition}\label{.odd_prime_polar_commutator_ismorphism}
    Let \en[d] be an odd prime (especially, $d \neq 2$), \en, $A$ be an abelian group with embedding $\zeta: \Z_d \to A$, $W_i \in \M_n\w(\Z_d)$ such that $\Omega_i = W_i - W_i\ut$ has full rank, i.e., is symplectic, for $i = 1, 2$. Then the groups $\gk_d^n(A, \zeta, W_1)$ and $\gk_d^n(A, \zeta, W_2)$ are isomorphic (\emph{the isomorphism is constructed in the proof}).
\end{proposition}
\begin{proof}
    Since $\Omega_i$ is a symplectic form, there is an invertible matrix $M_i \in M_n\w(\Z_d)$ such that $\Omega_i = M_i\ut \mathrm{\Omega}_n M_i$ for $i = 1, 2$, where $\mathrm{\Omega}_n = \smatp{0&\1\\-\1&0}$ is the standard symplectic form. Set $M = M_2\ui M_1$ (i.e, $\Omega_1 = M\ut \Omega_2 M$) and define $C = (M\ut W_2 M - W_1) / 2$. Note that $C$ is symmetric, since $2(C - C\ut) = M\ut \Omega_2 M - \Omega_1 = 0$, and therefore we have $C + C\ut + W_1 = M\ut W_2 M$. We state that the following mapping defines the wanted isomorphism
    \begin{equation}
        \phi: \gk_d^n(A, \zeta, W_1) \to \gk_d^n(A, \zeta, W_2)\;, (a, p, x) \mapsto (a, p +x\ut C x, M x)\;.
    \end{equation}
    Since $M$ is bijective, $\phi$ is bijective and for $(a, p, x), (b, q, y) \in \gk_d^n(A, \zeta, W_1)$ - without loss of generality, $a = b = 1, p = q = 0$ (since $(\cdot, \cdot, 0) \in \cen[\gk]$ these factors can be trivially factored out in $\phi$) - we have
    \begin{subalign}
        \phi\w((1, 0, x)_1 (1, 0, y)_1) &= \phi((1, x\ut W_1 y, x + y)_1)\\
        &= (1, (x+y)\ut C (x+y) + x\ut W_1 y, M(x + y))_2\\
        &= (1, x\ut C x + y\ut C y + x\ut \w(C + C\ut + W_1)y, M x + M y)_2\\
        &= (1, x\ut C x + y\ut C y + x\ut M\ut W_2 My, M x + M y)_2\\
        &= (1, x\ut C x, M x)_2 (1, y\ut C y, M y)_2\\
        &= \phi\w((1, 0, x)_1) \phi\w((1, 0, y)_1)\;.
    \end{subalign}
\end{proof}
\begin{remark}
    \Cref{.odd_prime_polar_commutator_ismorphism} fails for $d = 2$ since we cannot construct $C$ as we cannot divide by $2$. If the matrix $M\ut W_2 M - W_1$ has only zeros on the diagonal, we can fix the proof by defining $C$ to be the strictly lower (or upper) triangluar part of $M\ut W_2 M - W_1$. If this is not case, we can still construct the isomorphism with an additional requirement on $A$, namely that it contains the roots of $\zeta\w(\Z_d)$.
\end{remark}
\begin{proposition}[\cref{.odd_prime_polar_commutator_ismorphism} but allow $d =
2$]\label{.prime_polar_commutator_ismorphism}
    Let us be in the setting of \cref{.odd_prime_polar_commutator_ismorphism}, but now we allow $d = 2$, however, additionally require $\sqrt{a} \in A$ for all $a \in \zeta\w(\Z_d)$ (since $\Z_d$ is cyclic this is equivalent to $\sqrt{\zeta(0)} \in A$). Then the groups $\gk_d^n(A, \zeta, W_1)$ and $\gk_d^n(A, \zeta, W_2)$ are isomorphic (the isomorphism is constructed in the proof)
\end{proposition}
\begin{proof}
    Define $M$ analogously as in the proof of \cref{.odd_prime_polar_commutator_ismorphism}, set $C = M\ut W_2 M - W_1$ and set $\zeta': \Z_{2d} \to A, m \mapsto \sqrt{\zeta(0)}^m$. The isomorphism is then defined as
    \begin{equation}
        \phi: \gk_d^n(A, \zeta, W_1) \to \gk_d^n(A, \zeta, W_2)\;, (a, p, x) \mapsto \w(r\w(a\zeta'(x\ut Cx)), p + u\w(a\zeta'(x\ut Cx)), M x)\;,
    \end{equation}
    where the terms $x\ut C x$ are evaluated as quadractic form $\Z_{2d}^n \to \Z_{2d}$. Again, since $M$ is bijective, $f$ is bijective.
    Now let $(a, 0, x), (b, 0, y) \in \gk_d^n(A, \zeta, W_1)$ - we can trivially factor out $(1, \cdot, 0)$ - and set $d \coloneqq r(ab) \zeta'(x\ut C x + y\ut C y)$, $a' = a\zeta'(x\ut Cx)$, and $b' = b\zeta'(y\ut Cy)$; then
       \begin{subalign}
        \phi\w((a, 0, x)_1 (b, 0, y)_1) &= \phi((r(ab), u(ab) + x\ut W_1 y, x + y)_1)\\
        &= (r\w(d \zeta(x\ut C y)), u(ab) + u\w(d \zeta(x\ut C y)) + x\ut W_1 y, M(x + y))_2\\
        &\!\overset{(*)}{=} (r(d), u(ab) + u(d) + x\ut C y + x\ut W_1 y, M(x + y))_2\\
        &= (r(d), u(ab) + u(d) + x\ut M\ut W_2 My, M x + M y)_2\\
        &\!\!\overset{(**)}{=} (r\w(r(a')r(b')), u\w(r(a')r(b')) + u(a') + u(b') + x\ut M\ut W_2 My, M x + M y)_2\\
        &= (r(a'), u(a'), M x)_2 (r(b'), u(b'), M y)_2\\
        &= \phi\w((a, 0, x)_1) \phi\w((b, 0, y)_1)\;,
    \end{subalign}
    where in $(*)$ and $(**)$ we used that $r$ and $u$ are well defined together with
    \begin{subalign}
        r\w(d\zeta(x\ut Cy)) \zeta\w(u\w(d\zeta(x\ut Cy))) &= d\zeta(x\ut Cy)\\
        &= r(d) \zeta(u(d)) \zeta(x\ut Cy)\\
        &= r(d) \zeta(u(d) + x\ut Cy)
    \end{subalign}
    and
    \begin{subalign}
        r(d) \zeta\w(u(ab) + u(d)) &= d \zeta\w(u(ab))\\
        &= r(ab) \zeta'(x\ut Cx + y\ut Cy) \zeta\w(u(ab))\\
        &= ab \zeta'(x\ut Cx) \zeta'(y\ut Cy)\\
        &= a'b'\\
        &= r(a')r(b')\zeta\w(u(a') + u(b'))\\
        &= r\w(r(a')r(b')) \zeta\w(u\w(r(a')r(b')) + u(a') + u(b'))\;,
    \end{subalign}
    respectively.
\end{proof}
\begin{remark}
    If $W_1$ and $W_2$ are already symplectic, it holds $C = 0$ (for $d \neq 2$; for $d = 2$ it would be $\Omega = 0$).
\end{remark}
\begin{example}[\cite{room_synthesis_of_clifford_matrices}]\label{x"pauli_to_majorana}
    For $d = 2$ and \en, consider the case where $\Omega_1 = \mathrm{\Omega}_n$ (Pauli commutator) and $\Omega_2 = P - P\ut$ (Majorana commutator), where $P$ is the parity matrix, i.e, $P_{ij} = 1$ if $i > j$ and $0$ otherwise for all $i, j \in \bc{1, \ldots, n}$. Both matrices are symplectic, as discussed in \cref{s"conjugation_isomorphism} (or \cref{.parafermion_weyl_isomorphism}), and therefore, \cref{.prime_polar_commutator_ismorphism} holds. Reference~\cite{room_synthesis_of_clifford_matrices} gives us a direct construction of $M$: For $i \in \bc{1, \ldots, n}$ let $z_i, x_i \in \Z_2^{2n}$ be the $i$th and $n+i$th Euclidean basis vector, respectively. Then, $M$ is defined as follows
    \begin{equation}
        Mx_i = x_i + \sum_{j < i} z_j\;,\qquad Mz_i = z_i + Mx_i\;, \qquad \text{for all } i \in \bc{1, \ldots, n}\;.
    \end{equation}
    The image of $M$ is a generating set of $\Z_2^{2n}$ since we have $z_i = Mz_i + Mx_i$ and $x_i = Mx_i + \sum_{j < i} z_j$, $i \in \bc{1, \ldots, n}$, and with that, $M$ is indeed invertible. It is easy to check that $M\ut \Omega_1 M = \Omega_2$.
\end{example}

\begin{definition}[Polar commutator
representation]\label{.polar_commutator_representation}
    Let \en[d, n], $W \in \M_n\w(\Z_d)$, $R$ a commutative ring with a $d$th root of unity $\omega$ and $A \leq R^{\times}$ be a subgroup of the multiplicative unit group of $R$ (inclusion possibly via an embedding) with $\omega \in A$ (or $\sqrt{\omega} \in A$ if required). Set $\zeta: \Z_d \to A, m \mapsto \omega^m$. Furthermore, let $M$ be an associative $R$-algebra\footnote{That is, an $R$-module with a bilinear product as in an associative algebra over a field.} and $\tau: \Z_d^n \to M$ such that
    \begin{equation}
        \mu: \gk_d^n(A, \zeta, W) \to M, \; (a, p, x) \mapsto a \omega^p \tau(x)\;,
    \end{equation}
    is a monomorphism with respect to the multiplication in $M$.
\end{definition}
As shorthand, we write $\mu(\cdot, \cdot, \cdot) \coloneqq \mu((\cdot, \cdot, \cdot))$ and we call $\mu(\gk)$ the (isomorphic) ``representation'' of $\gk$ in $M$ via $\tau$. The commutators have a simple representation:
\begin{proposition}
    Let us be in the situation as in \cref{.polar_commutator_representation}. For $(a, p, x), (b, q, y) \in \gk$, we have
    \begin{equation}
        \tbk{\mu(a, p, x), \mu(b, q, y)} = \omega^{x\ut \Omega y}\;,
    \end{equation}
    and for the commutator Lie bracket, we have
    \begin{equation}
      \bk{\mu(a, p, x), \mu(b, q, y)} = \w(1 - \omega^{-x\ut \Omega y}) \mu(a, p, x)\mu(b, q,
      y)\;.
    \end{equation}
\end{proposition}
\begin{proof}
    Since for any invertible $g, h \in M$, it holds $\bk{g, h} = \w(1 - \tbk{h, g}) gh$, we only have to show one of the above equations (remember that $\Omega\ut = -\Omega$). Let $(a, p, x), (b, q, y) \in \gk$. We have
    \begin{equation}
        \tbk{\mu(a, p, x), \mu(b, q, y)} = \mu\w(\tbk{(a, p, x), (b, q, y)}) =\mu\w(1, x\ut \Omega y, 0) = \omega^{x\ut \Omega y}\;.
    \end{equation}
    Here we used that $\tau(0) = 1 \cdot 1 \cdot \tau(0) = \mu(1, 0, 0) = 1$, since $\mu$ is a homomorphism.
\end{proof}
\begin{remark}
    Since the multiplicative commutator is a scalar, the Lie bracket is equivalent to multiplication in $M$ for non-commuting elements, up to an R-scalar. Therefore, closing a subset of $\mu(\gk)$ under the Lie bracket is equivalent to closing it under multiplication, potentially including scalar factors (including $0$ if there are commuting elements).
\end{remark}
\begin{corollary}\label{.polar_commutator_similarity}
    Let \en[d] be prime, \en, $K$ be a field with a $d$th root of unity $\omega$ and $A \leq K^{\times}$ such that $\omega \in A$, and $\sqrt{\omega} \in A$ if $d = 2$. Let $W_1, W_2 \in \M_n\w(\Z_d)$ such that $\Omega_1, \Omega_2$ are symplectic forms. Let $V_1, V_2$ be finite-dimensional $K$-vector spaces, and $\zeta$ and $\mu_i: \gk_d^n(A, \zeta, W_i) \to \GL\w(V_i)$ be representations as defined in \cref{.polar_commutator_representation}.

    Then, it holds $\gk \cong \gk_d^n(A, \zeta, W_1) \cong \gk_d^n(A, \zeta, W_2)$, and if, and only if, the representations have the same character, then the representations are similar, i.e., with $V \cong V_1 \cong V_2$, it exists $S \in \GL(V)$ such that $\mu_2(g) = S \mu_1(g) S\ui$ for all $g \gk$.
\end{corollary}
\begin{proof}
    This is just \cref{.odd_prime_polar_commutator_ismorphism,.prime_polar_commutator_ismorphism} and a well known result from representation character theory, namely, that two representations are intertwining (similar, module-isomorphic) if, and only if, they have the same character.
\end{proof}

\subsection{Complex Representations}

    We now consider the special case of \cref{.polar_commutator_representation} where $R = \C$. Throughout this section let $\gk$, $\Omega$, $\omega$, $\zeta$ and $\mu$ be defined as in \cref{.polar_commutator_representation} (for given $d$, $A$, $n$, $W$ and $\tau$).
\begin{proposition}\label{.polar_commutator_similarity_unitary}
    Let us be in the situation of \cref{.polar_commutator_similarity} with $K = \C$ and similar representations via $S \in \GL(V)$. If $\gk$ contains a set $B$, such that $\mu_1(B)$ is a set of hermitian generators of the vector space $M_n(\C)$ and $\mu_2(g)$ is hermitian for all $g \in B$, then $S$ is unitary (after a potential normalisation). Instead of requiring that the operators in $\mu_1(B) \cup \mu_2(B)$ are hermitian, we can also require that they are unitary.
\end{proposition}
\begin{proof}
    Let $X = \mu_1(g)$, $g \in B$, and $Y = \mu_2(g) = S * X$. In the hermitian case, we have
    \begin{equation}
       \w(S\ur S) * X = S\ur * (S * X) = S\ur * Y = S\ur * Y\ur = (S\ui * Y)\ur = X \ur = X\;.
    \end{equation}
    and in the unitary case, we have
    \begin{equation}
        \w(S\ur S) * X = S\ur * (S * X) = S\ur * Y = S\ur * \w(Y\ui)\ur = (S\ui * Y\ui)\ur = \w((S\ui * Y)\ui)\ur = \w(X\ui)\ur = X\;.
    \end{equation}
    Therefore we have $S\ur S = \lambda \1$, with $\lambda \in \C$, since an operator is uniquely (up to a factor) defined by how it conjugates a generating set (elementary proof by using that euclidean matrix basis elements are mapped to themselves). After normalizing $S$ we have $S\ur S = \1$ (note that $\lambda = 0$ is not possible since both, $S$ and $S\ur$ have full rank).
\end{proof}
\begin{remark}
    \begin{enumerate}
       \item \cref{.odd_prime_polar_commutator_ismorphism,.prime_polar_commutator_ismorphism,.polar_commutator_similarity,.polar_commutator_similarity_unitary} give \cref{.main_isomorphism} in the main text.
      \item In the specific case where $\gk$ is the Heisenberg group, this is essentially the Stone-von Neumann theorem.
    \end{enumerate}
\end{remark}

\subsubsection{The Weyl-Heisenberg Group and
Parafermions}\label{s"weyl_heisenberg_and_parafermions}

We end the discussion on the discrete Stone-von Neumann theorem by quickly discussing the well known Weyl-Heisenberg group and some of its properties, as this is probably the most canonical polar commutator group and representation. Furthermore, we give an abstract definition of parafermions - as well as an explicit construction via the Weyl-Heisenberg group - as an important unitary equivalent group to the Weyl-Heisenberg group.
\begin{remark}
    Reminder: For odd $d$, $2$ has a multiplicative inverse in $\Z_d$; and for even $d$, $2$ does not have a multiplicative inverse in $\Z_d$.
\end{remark}
\begin{definition_proposition}[Weyl-Heisenberg group]\label{.weyl_heisenberg_group}
    Let \en[d] be odd with $d>2$, $n=2m$ for an \en[m], $A \leq \C^{\times}$ such that $\omega \in A$, and $W = \mathrm{\Omega}_n / 2 \in \M_n\w(\Z_d)$. $\gk_d^n(A, \zeta, W)$ is (sometimes) called the Heisenberg group (or Clifford collineation group). Define $\mu$ and $\tau$ linearly via
    \begin{equation}
       \tau: \Z_d^n \cong \Z_d^m \times \Z_d^m \to \M_{d^m}\w(\C), \; (z, x) \mapsto \w(\tau(z, x): \C^{d^m} \to \C^{d^m},\; \sum_{s \in \Z_d^m} \lambda_s e_s \mapsto \sum_{s \in \Z_d^m} \lambda_s \omega^{\braket{z, s + x/2}} e_{s + x})\;,
    \end{equation}
    where $\lambda_0, \ldots, \lambda_{d^m-1} \in \C$ and $e_0, \ldots, e_{d^m-1}$ is the euclidean basis (indices are encoded in base $d$), is its Weyl representation.\\
    Furthermore, we have $\Omega = \mathrm{\Omega}_n$, which is a symplectic form.
\end{definition_proposition}
\begin{proof}
    We have to show that $\tau$ is well defined and induces $\mu$ as a injective homomorphism. It is clear that the operators defined by $\tau$ are linear (per definition). To proof the homomorphism characteristic, let $(a, p, u = (z_u, x_u)), (b, q, v = (z_v, x_v)) \in \gk$ - without loss of generality, $a = b = 1, p = q = 0$ (since $(\cdot, \cdot, 0) \in \cen[\gk]$ these factors can be trivially factored out in $\mu$) - and $s \in \Z_d^m$. We have
    \begin{subalign}
        \mu((1, 0, u)(1, 0, v))e_s &= \mu(1, u\ut W v, u + v)e_s\\
        &= \omega^{(\braket{z_u, x_v} - \braket{x_u, z_v})/2} \omega^{\braket{z_u + z_v, s +(x_u + x_v)/2}} e_{s + x_u + x_v}\\
        &= \omega^{\braket{z_u, s + x_v + x_u/2}} \omega^{\braket{z_v, s + x_v/2}} e_{(s + x_v) + x_u}\\
       &= \mu(1, 0, u)\mu(1, 0, v)e_s\;.
    \end{subalign}
    Finally, it is clear, that the kernel of $\mu$ is $\bc{(1, 0, 0)}$.
\end{proof}
\begin{definition_proposition}[Continue \cref{.weyl_heisenberg_group}; alternative
definition]\label{.weyl_heisenberg_group_alt}
    Instead of $W = \mathrm{\Omega}_n / 2$, set $W = \smatp{0 & 0 \\ -\1^{m\times m} & 0} \in \M_n\w(\Z_d)$ and define $\mu$ via
    \begin{equation}
        \tau: \Z_d^n \cong \Z_d^m \times \Z_d^m \to \M_{d^m}\w(\C), \; (z, x) \mapsto \w(\tau(z, x): \C^{d^m} \to \C^{d^m},\; \sum_{s \in \Z_d^m} \lambda_s e_s \mapsto \sum_{s \in \Z_d^m} \lambda_s \omega^{\braket{z, s + x}} e_{s + x})\;.
    \end{equation}
    Furthermore, we can drop the condition that $d$ is odd and only require that $d \geq 2$.\\
    If $\sqrt{\omega} \in A$, this defines an isomorphic group to the group defined in \cref{.weyl_heisenberg_group} (if $d$ is odd); specifically, the isomorphism is given by
    \begin{equation}
        \begin{split}
            f: \gk_d^n\w(A, \zeta, \smatp{0 & \1^{m\times m} \\ -\1^{m\times m} & 0}) &\to \gk_d^n\w(A, \zeta, \smatp{0 & 0 \\ -\1^{m\times m} & 0})\\
             \w(a, p, (z, x)) &\mapsto \w(r\w(a\omega^{-\braket{z, x}/2}), p + u\w(a\omega^{-\braket{z, x}/2}), (z, x))\;.
        \end{split}
    \end{equation}
    Moreover, the representation via $\mu$ has the same image as the one in \cref{.weyl_heisenberg_group_alt}. Furthermore, we still have $\Omega = \mathrm{\Omega}_n$, which is a symplectic form.
\end{definition_proposition}
\begin{proof}
    Firstly, one proves analogously that $\mu$ is indeed a well-defined monomorphism. Now let $\gk_s, \mu_s$ be the group and representation defined in \cref{.weyl_heisenberg_group}, and $\gk_a, \mu_a$ the group and representation defined in this current definition. The isomorphy statement follows directly from \cref{.prime_polar_commutator_ismorphism} ($M$ being the identity). However, as an exercise, we also show it manually via the representations: It suffices to show that $\mu_s\w(\gk_s) = \mu_a\w(\gk_a)$, because then we can compose $f = \mu_a\ui \circ \mu_s$. Let $(a, p, v = (z, x))_s \in \gk_s$. It holds $\w(r\w(a\omega^{-\braket{z, x}/2}), p + u\w(a\omega^{-\braket{z, x}/2}), v) \in \gk_a$ and for $s \in \Z_d^m$ we have
    \begin{subalign}
        \mu_s\w(a, p, v)e_s &= a \omega^{p + \braket{z, s + x/2}}e_{s + x}\\
        &= a\omega^{-\braket{z, x}/2} \omega^{p + \braket{z, s + x}}e_{s + x}\\
        &= r\w(a\omega^{-\braket{z, x}/2}) \omega^{u\w(a\omega^{-\braket{z, x}/2})} \omega^{p + \braket{z, s + x}}e_{s + x}\\
        &= \mu_a\w(r\w(a\omega^{-\braket{z, x}/2}), p + u\w(a\omega^{-\braket{z, x}/2}), v)e_s\;.
    \end{subalign}
    Therefore we have $\mu_s\w(\gk_s) \subseteq \mu_a\w(\gk_a)$ and analogously one shows $\mu_a\w(\gk_a) \subseteq \mu_s\w(\gk_s)$.
\end{proof}
\begin{remark}
    \begin{enumerate}
        \item The Heisenberg-Weyl group is a generalization of the Pauli group for $d > 2$. Both definitions above have its advantages and disadvantages: In \cref{.weyl_heisenberg_group}, the group multiplication is already defined by a symplectic form, however, \cref{.weyl_heisenberg_group_alt} works for all $d \geq 2$ (not just odd numbers; but cf. next point) and especially for $d = 2$ it reduces to the Pauli group.
        \item One can adjust the definition in \cref{.weyl_heisenberg_group} to also work for even $d$ by putting the action of $W$ into the first tuple argument via $\zeta$ and taking the square root of $\omega$ (requiring $\sqrt{\omega} \in A$, analogously how it is done in \cref{.prime_polar_commutator_ismorphism} for $C$). This definition is, for example, used in ~\cite{gross_hudsons_theorem_finite_quantum_system}.
        \item For prime $d$, the Heisenberg group defines in a certain sense the only polar commutator group with symplectic $\Omega$; cf. \cref{.odd_prime_polar_commutator_ismorphism,.prime_polar_commutator_ismorphism}. Furthermore as we shall see below, the representation is unitary and provides a basis; therefore other unitary representations with the same character (trace zero, except for the neutral element) are unitary conjugations of this representation. The normaliser of this representation is the Clifford (transform) group, which is unitary (up to scalar factors) and isomorphic to the symplectic group (taken modulo the Weyl group) ~\cite{bolt_clifford_stuff_1,bolt_clifford_stuff_2,bolt_clifford_stuff_3}.
        \item In general, these groups have been extensively studied, e.g., ~\cite{bolt_clifford_stuff_1,bolt_clifford_stuff_2,bolt_clifford_stuff_3} and ~\cite{heinrich_stabiliser_techniques,gross_hudsons_theorem_finite_quantum_system}.
    \end{enumerate}
\end{remark}
\begin{proposition}\label{.weyl_heisenberg_properties}
    Let \en[d \geq 2], \en[n = 2m] and $\gk = \gk_d^n(A, \zeta, W)$ be the Heisenberg group with Weyl representation $\mu$ as defined in \cref{.weyl_heisenberg_group_alt} (or \cref{.weyl_heisenberg_group}). The representation has the following properties; let $(a, p, u) \in \gk$:
    \begin{enumerate}
        \item $\mu(a, p, u)$ is unitary up to the factor $a$, i.e., $\mu(a, p, u) \mu(a, p, u)\ur = \abs{a}^2$,
        \item $\mu(a, p, u)$ is traceless if, and only if, $u \neq 0$,
        \item $\tr \mu(a, p, 0) = a \omega^p d^m$,
        \item $\mu(\gk)$ is irreducible,
        \item the representatives of $\mu(\gk) / \cen\w(\mu(\gk)))$ form an orthogonal basis with respect to the Hilbert-Schmidt inner product, with norm $\abs{a}^2 d^m$ for $\mu(a, p, u) \in \mu(\gk)$.
        \item for $d = 2$, $\mu(a, p, u)$ is hermitian up to the factor $a^*a\ui \omega^{u\ut Wu} = \ee^{i\w(u\ut W u \pi - 2 \arg(a))}$,
        \item for $d > 2$, only $\mu(a, p, 0)$ is hermitian, up to the factor $a^*a\ui \omega^{-2p} = \ee^{-2i\w(p \arg(\omega) + \arg(a))}$.
    \end{enumerate}
\end{proposition}
\begin{proof}
    Regarding the unitarity and the trace, let $(a, p, u = (z, x)) \in \gk$, $s, t \in \Z_d^m$, and set $k = \abs{\supp z}$. We have
    \begin{equation}
        \braket{\mu(a, p, u)e_s, \mu(a, p, u)e_t} = \abs{a}^2 \omega^{\braket{z, s - t}}
        \braket{e_{s + x}, e_{t + x}}
        = \abs{a}^2 \omega^{\braket{z, s - t}} \delta_{s + x, t + x}
        = \abs{a}^2 \delta_{s, t}
        = \abs{a}^2 \braket{e_s, e_t}\;,
    \end{equation}
    and
    \begin{subalign}
        \tr \mu(a, p, u) &= \sum_{s \in \Z_d^m} \braket{e_s, \mu(a, p, u)e_s}\\
        &= a \omega^p \delta_{x, 0} \sum_{s \in \Z_d^m} \omega^{\braket{z, s}}\\
        &\!\overset{(*)}{=} a \omega^p \delta_{x, 0} d^{m-k} \w(\prod_{i \in \supp z} \w(\sum_{l=0}^{d-1} \omega^{z_i l}))\\
        &\!\!\overset{(**)}{=} a \omega^p \delta_{x, 0} d^{m-k} 0^k\\
        &= a \omega^p \delta_{x, 0} d^{m-k} \delta_{k, 0}\\
        &= a \omega^p d^m \delta_{u, 0}\;.
    \end{subalign}
    where in $(*)$ we decomposed the sum over $\Z_d^m$ into multiple sums over $\Z_d$ sorted w.r.t., the support of $z$, and in $(**)$ we used that $\omega$ is $d$th-root of unity, i.e., we know that for any $j \in \N_{>0}$ $\omega^j \neq 1$ solves the polynomial (in $X$; use the telescope sum)
    \begin{subalign}
        0 &= X^d - 1\\
        &= (X - 1) \sum_{l = 0}^{d-1} X^l\;.
    \end{subalign}
    Regarding the irreducibility, consider the subgroup $H = \set{\mu(0, p, u)}{p \in \Z_d, u \in \Z_d^n} \leq \gk$ we have
    \begin{equation}
        \frac{1}{\abs{H}} \sum_{h \in H} \abs{\tr \mu\w(h)}^2 = \frac{d \w(d^m)^2}{d^{n+1}} = 1\;,
    \end{equation}
    i.e., its character is normalised, and therefore $\mu(H)$ is irreducible and with that $\mu(\gk)$ is irreducible too.

    The orthogonality and norm follow from the unitarity and the previous trace calculation. Furthermore, we know that $\abs{\mu(\gk)/\cen\w(\mu\gk))} = d^n = d^{2m}$, and therefore the representatives form a maximal linearly independent set.

    Regarding the hermiticity for $d = 2$, let $(a, p, v) \in \gk_2$, then
    \begin{subalign}
        \mu(a, p, v)\ur &= \mu(a, p, v)\ui \abs{a}^2\\
        &= \mu\w((a, p, v)\ui) \abs{a}^2\\
        &= \mu\w(r\w(a\ui), p + u\w(a\ui) + v\ut W v, v) \abs{a}^2\\
        &= \mu\w(a, p + v\ut W v, v) a^* a\ui \omega^{v\ut Wv}\;,
    \end{subalign}
    Now for $d > 2$, let $(a, p, v) \in \gk_d$, we have
    \begin{equation}
        \mu(a, p, v)\ur = \mu\w(r\w(a\ui), -p + u\w(a\ui) + v\ut W v, -v) = \mu\w(a, p, -v) a^*a\ui \omega^{-2p + v\ut Wv}\;,
    \end{equation}
    so for hermicity (up to a factor), we require $v = 0$.
\end{proof}
\begin{definition}[Parafermions]\label{.parafermions}
Let \en[d] with $d\geq2$ and $n=2m$ for an \en[m]. The group of parafermions, $\mathcal{M}(d, m)$ is defined to be generated by the operators $\gamma_1, \ldots, \gamma_n \in \C^{d^m \times d^m}$, and by scalars in $A$, such that the following relations hold for all $i, j \in \{1, \ldots, n\}$ and $k, l \in \bc{0, \ldots, d-1}$:
\begin{equation}
  \gamma_i\ur = \gamma_i\ui = \gamma_i^{d-1}\;, \qquad \gamma_i^k \gamma_j^l =
  \omega^{\sign(j-i) k l} \gamma_j^l \gamma_i^k\;.
\end{equation}
Furthermore, we require that the trace of any product of those generators that does not yield the identity (up to a scalar factor) is zero.
\end{definition}
\begin{remark}
For $d = 2$, parafermions are Majorana fermions.
\end{remark}
\begin{proposition}\label{.parafermion_weyl_isomorphism}
Let \en[d] with $d\geq2$ and $n=2m$ for an \en[m]. The groups of parafermions, Define $W\in \Z_d^{n\times n}$ to be the negative lower left parity matrix, i.e., $W_{ij} = -1$ if $i > j$ and $0$ otherwise for all $i, j \in \bc{1, \ldots, n}$ and
\begin{equation}
  \tau: \Z_d^n \to \M_{d^m}\w(\C), \; x \mapsto \gamma_1^{x_1} \cdots \gamma_n^{x_n}\;.
\end{equation}
Then we have $\mathcal{M}(d, m) = \mu(\gk_d^n(A, \zeta, W))$ via $\tau$. Furthermore, if $d$ is prime, then $\Omega = P - P\ut$ is a symplectic form and we the parafermion group is unitarily equivalent to the Weyl-Heisenberg group.
\end{proposition}
\begin{proof}
It is an easy exercise to check that $\mu$ via $\tau$ is well-defined, which then proves the first statement. Now let $d$ be prime. One can convince oneself that
\begin{equation}
  \Omega = \matp{
    0 & 1 & \cdots & \cdots & 1 \\
    -1 & \ddots & \ddots & & \vdots \\
    \vdots & \ddots & \ddots & \ddots & \vdots \\
    \vdots & & \ddots & \ddots & 1 \\
    -1 & \cdots & \cdots & -1 & 0
  }
\end{equation}
has the inverse (alternating signs on the diagonals; note that this works only because $n$
is an even number)
\begin{equation}
  \Omega\ui = \matp{
    0 & -1 & 1 & -1 & \cdots & 1 & -1 \\
    1 & 0 & -1 & 1 & -1 & \ddots & 1 \\
    -1 & 1 & \ddots & \ddots & \ddots & \ddots & \vdots \\
    1 & -1 & \ddots & \ddots & \ddots & \ddots & -1 \\
    \vdots & 1 & \ddots & \ddots & \ddots & \ddots & 1 \\
    -1 & \ddots & \ddots & \ddots & \ddots & 0 & -1 \\
    1 & -1 & \cdots & 1 & -1 & 1 & 0
  }\;,
\end{equation}
which proves that $\Omega$ is symplectic. Furthermore, both parafermions and Weyl operators have a vanishing trace for non-trivial elements and it is easy to see that the unitarity condition in \cref{.polar_commutator_similarity_unitary} is fulfilled. Therefore, the parafermion group is unitary equivalent to the Weyl-Heisenberg group.
\end{proof}
\begin{proposition}[Explicit parafermions \cite{fendley_parafermionic_edge_zero_modes}]
Let \en[d] with $d\geq2$, $n=2m$ for an \en[m] and $\tau: \Z_d^m \times \Z_d^m \to \M_{d^m}\w(\C)$ be defined as in \cref{.weyl_heisenberg_group_alt}. Then, the operators $\gamma_1, \ldots, \gamma_n$ in \cref{.parafermions} can be explicitly constructed via (compare \cref{x"pauli_to_majorana})
\begin{align}
  \gamma_{2i-1} &= \tau\w(e_i, 0) \prod_{j<i} \tau(0, e_j)\;,\\
  \gamma_{2i} &= \tau\w(0, e_i) \gamma_{2i-1} \;,
\end{align}
for $i \in \bc{1, \ldots, m}$, where $e_1, \ldots, e_m$ are the Euclidean basis vectors of $\Z_d^m$.
\end{proposition}
\begin{proof}
We first prove the commutation relations. Let $i, j \in \bc{1, \ldots, m}$ with $i < j$; then
\begin{align}
  \gamma_{2i-1} \gamma_{2j-1} &= \tbk{\tau(e_i, 0), \tau(0, e_i)} \gamma_{2j-1}
  \gamma_{2i-1} = \omega \gamma_{2j-1}\\
  \gamma_{2i} \gamma_{2j} &= \tbk{\gamma_{2i-1}, \gamma_{2j-1}} \gamma_{2j} \gamma_{2i} =
  \omega \gamma_{2j} \gamma_{2i}\\
  \gamma_{2i-1} \gamma_{2j} &= \tbk{\gamma_{2i-1}, \gamma_{2j-1}} \gamma_{2j}
  \gamma_{2i-1} = \omega \gamma_{2j} \gamma_{2i-1}\\
  \gamma_{2i} \gamma_{2j-1} &= \tbk{\gamma_{2i-1}, \gamma_{2j-1}} \gamma_{2j-1}
  \gamma_{2i} = \omega \gamma_{2j-1} \gamma_{2i}\;.
\end{align}
Now let $i = j$:
\begin{align}
  \gamma_{2i-1} \gamma_{2i} &= \tbk{\tau(e_i, 0), \tau(0, e_i)} \gamma_{2i}
  \gamma_{2i-1} = \omega \gamma_{2i} \gamma_{2i-1}\\
  \gamma_{2i} \gamma_{2i-1} &= \tbk{\tau(0, e_i), \tau(e_i, 0)} \gamma_{2i-1}
  \gamma_{2i} = \omega^{-1} \gamma_{2i-1} \gamma_{2i}\;.
\end{align}
This proves the commutation relations. It is clear that $\gamma_{2i-1}^d = \1$ and furthermore we have
\begin{equation}
  \gamma_{2i}^d = \tbk{\gamma_{2i-1}, \tau(0, e_i)}^{\sum_{l=0}^{d-1} l} \1 =
  \omega^{d(d-1)/2} \1 = \1\;.
\end{equation}
Finally, the trace condition is fulfilled since it is fulfilled for the Weyl operators.
\end{proof}

\begin{acronym}[\textsf{QMA}]
\acro{scf}[SCF]{simplicial, claw-free}
\acro{qma}[\textsf{QMA}]{Quantum Merlin-Arthur\acroextra{ complexity class}}
\acro{np}[\textsf{NP}]{Nondeterministic polynomial time\acroextra{ complexity class}}
\end{acronym}

\end{document}